\documentclass[12pt, draftclsnofoot, onecolumn]{IEEEtran}
\usepackage{amsmath,graphicx}
\usepackage{color}
\usepackage{float}
\usepackage{amsmath,bm,amssymb,amsfonts,amsthm,epsfig}
\usepackage{mathtools}
\usepackage{graphicx}
\usepackage{psfrag}
\usepackage{comment}
\usepackage{url}
\usepackage{cite} 
\usepackage{epstopdf}
\usepackage{tabularx}
\newcounter{opteq}
\newenvironment{opteq}{\refstepcounter{opteq}\align}{\tag{P\theopteq}\endalign}
\usepackage{makecell}
\usepackage{arydshln}
\usepackage{mathrsfs}
\usepackage{bbm}

\DeclareMathOperator*{\argmin}{arg\,min}

\makeatletter
\newcommand*{\rom}[1]{\expandafter\@slowromancap\romannumeral #1@}
\makeatother

\newtheorem{definition}{Definition}
\newtheorem{lemma}{Lemma}
\newtheorem{proposition}{Proposition}
\newtheorem{thm}{Theorem}

\newtheorem{remark}{Remark}

\usepackage{subcaption}
\usepackage{algorithm}
\usepackage{algorithmicx}
\usepackage[noend]{algpseudocode}

\usepackage{threeparttable}

\newcommand{\txt}[1]{\text{\normalfont #1}}   

\title{Hybrid Beamforming for Active Sensing using Sparse Arrays}

\author{Robin~Rajam\"{a}ki$^\dagger$,~\IEEEmembership{Student Member,~IEEE,}
	Sundeep Prabhakar Chepuri$^\star$,~\IEEEmembership{Member,~IEEE,}
	and~Visa~Koivunen$^\dagger$,~\IEEEmembership{Fellow,~IEEE}
	\thanks{$^\dagger$Department of Signal Processing and Acoustics, Aalto University, Espoo, Finland. (e-mail: robin.rajamaki@aalto.fi, visa.koivunen@aalto.fi)
		\newline$^\star$Department of Electrical Communication Engineering, Indian Institute of Science, Bangalore, India (e-mail: spchepuri@iisc.ac.in).}
	\thanks{This work was supported in part by the Academy of Finland project ``Massive and Sparse Antenna Array Processing for Millimeter-Wave Communications'', and the Tata Trusts.}}%

\begin{document}

	\maketitle

	\begin{abstract}	
	This paper studies hybrid beamforming for active sensing applications, such as millimeter-wave or ultrasound imaging. Hybrid beamforming can substantially lower the cost and power consumption of fully digital sensor arrays by reducing the number of active front ends. Sparse arrays can be used to further reduce hardware costs. We consider phased arrays and employ linear beamforming with possibly sparse array configurations at both the transmitter and receiver. The quality of the acquired images is improved by adding together several component images corresponding to different transmissions and receptions. In order to limit the acquisition time of an image, we formulate an optimization problem for minimizing the number of component images subject to achieving a desired point spread function. Since this problem is not convex, we propose algorithms for finding approximate solutions in the fully digital beamforming case, as well as in the more challenging hybrid and analog beamforming cases that employ quantized phase shifters. We also determine upper bounds on the number of component images needed for achieving the fully digital solution using fully analog and hybrid architectures, and derive closed-form expressions for the beamforming weights in these cases. Simulations demonstrate that a hybrid sparse array with very few elements, and even fewer front ends, can achieve the resolution of a fully digital uniform array at the expense of {a} longer image acquisition time.
	\end{abstract}
	
	\section{Introduction}\label{sec:introduction}
	{Sensor} arrays are a key technology with several applications in radar, sonar, microwave imaging, medical ultrasound, and wireless communications, to list a few \cite{vantrees2002optimum}. The many advantages of arrays include high {signal-to-noise ratio} (SNR) {gain}, {spatial diversity,} and the capability to {cancel} interferences {by beamforming}. {The ability to resolve targets improves with increasing aperture, which encourages using short carrier wavelengths. This allows for designing electrically large arrays with small form factors by packing {a very large number (on the order of hundreds) of} elements into a tiny {(on the order of a cm$ ^2 $)} physical area.} On the other hand, the cost, power consumption, and computational load commonly associated with signal processing for many antenna elements and dedicated transceiver chains may become prohibitively large. These issues are especially pronounced for fully digital arrays, where each array element is connected to separate \emph{front end}, which includes {\emph{radio} {and} \emph{intermediate frequency}} (RF-IF) components and an \emph{analog-to-digital converter} (ADC) or {a} \emph{digital-to-analog converter} (DAC). For example, a planar antenna array operating in the THz frequencies of the radio spectrum may in principle even fit thousands of elements in {an} area of only a few {square} centimeters. The practical applicability of such fully digital systems is limited by the number of required {RF-IF} front ends, and the typical high sampling rates and bandwidths imposed on the DACs/ADCs.
	
	\emph{Sparse arrays} can be used to reduce the cost of large arrays with a regular geometry. By utilizing a virtual array model called the \emph{co-array}, the number of elements can be significantly reduced compared to a uniform array of equivalent aperture, without sacrificing the array's ability to resolve scatterers or signal sources \cite{hoctor1990theunifying,pal2010nested,wang2017coarrays}. The co-array is a virtual array structure typically consisting of the pairwise vector sums or differences of the physical array element positions. For instance, the \emph{sum co-array} commonly arises in active {sensing} applications, where linear processing (delay-and-sum beamforming and matched filtering) is used at the transmitter and receiver. Sparse arrays exploit the fact that the co-array of a uniform array is redundant. {{Redundancy implies that} the same co-array can be achieved using fewer physical elements} {by carefully placing the sensors in a sparse manner.}
	
	The support of the co-array ultimately determines the achievable set of \emph{point spread functions} (PSFs){, which determine the} {properties} of the imaging system. A particular PSF may be achieved by weighting the co-array using the so-called \emph{image addition} technique~\cite{hoctor1990theunifying}. Image addition produces a desired co-array weighting by adding together several images, which are acquired using different transmit-receive beamforming weights. Each of these component weights correspond to a separate transmission, or pulse, and reception, when transmitters operate coherently as in a phased array. In this case, it is critical to keep the number of component images as low as possible {to reduce the image acquisition time}, while controlling the distortions to the PSF. If transmissions are incoherent, as in synthetic aperture radar, or orthogonal multiple-input-multiple-output (MIMO) radar \cite{li2009mimo}, the number of component images is less important. In such cases, image addition may be applied during post-processing after data acquisition \cite{ahmad2004designandimplementation}.
	
	\emph{Hybrid beamforming} may be used to further {lower} the cost of a fully digital array. Hybrid beamforming architectures reduce the number of front ends by pre-processing the transmitted or received signals by an analog beamforming network. This network usually consists of inexpensive, low power phase shifters connecting every array element to all front ends. Fig.~\ref{fig:architecture} depicts this \emph{fully connected} architecture. In a partially connected architecture, each front end is connected to only a subset of all the available elements {\cite{ahmed2018asurvey}}. The total power consumption and cost of the system may further be reduced by applying coarser quantization at the ADCs/DACs \cite{venkateswaran2010analog,alkhateeb2014mimo}, or by using sparse arrays. The typical application of hybrid beamforming is \emph{millimeter-wave} (mmWave) communications, where linear processing is used to precode and decode multiple data streams sent over a MIMO channel, with the goal of maximizing the {spectral efficiency} \cite{alkhateeb2014mimo} or {minimizing the mean squared error of the received data \cite{ioushua2019afamily}}.
	\begin{figure*}[t]
		\includegraphics[width=1\textwidth]{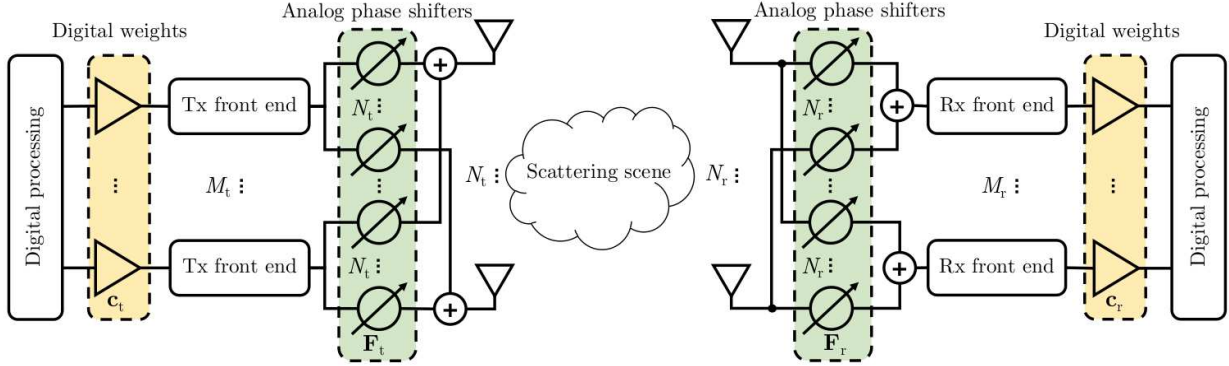} 
		\caption{Fully connected hybrid beamforming architecture of transmitter and receiver. Each digital weight (element of vector $ \mathbf{c}_\txt{x} $) and array element is connected via a front end and a phase shifter with discretized phase (element of matrix $\mathbf{F}_\txt{x}$).}\label{fig:architecture}
	\end{figure*}

	{The design of the hybrid beamformer is challenging as it requires solving a non-convex optimization problem.} In particular, non-convexity results from i) decomposing the fully digital beamformer into analog and digital parts; ii) introducing phase shifters in the analog beamforming network; and iii) using quantized phase shifts \cite{molisch2017hybrid}. Many authors have addressed these issues using {a {variety} of} analytical \cite{zhang2005variable,zhang2014onachieving,sohrabi2016hybrid,bogale2016onthenumber} {and} numerical tools \cite{sohrabi2016hybrid,bogale2016onthenumber,yu2016alternating,chen2017hybrid,jin2018hybridprecoding,ioushua2019afamily,arora2020hybrid}. Most analytical methods make use of the fact that any digital beamforming vector may be implemented by a fully connected hybrid beamformer using continuous phase shifters and two front ends \cite{zhang2005variable}. Actually, only a single front end per data stream is sufficient, if the number of phase shifters per front end is doubled \cite{zhang2014onachieving,sohrabi2016hybrid,bogale2016onthenumber}. {However, these} results are not applicable if the number of streams is greater than the number of front ends, or if the phase shifters are quantized. Consequently, several numerical approaches to solve the hybrid beamforming problem have been proposed, including alternating minimization {\cite{yu2016alternating,ioushua2019afamily}, majorization-minimization \cite{arora2020hybrid}}, quasi-Newton methods \cite{jin2018hybridprecoding}, Wirtinger {calculus} \cite{koochakzadeh2018beam}, {coordinate descent \cite{chen2017hybrid}}, and various heuristics \cite{sohrabi2016hybrid,bogale2016onthenumber}.
	
	\subsection{Contributions and organization of paper}
	The aforementioned works mainly consider hybrid beamforming in a mmWave {MIMO} communications context. In contrast, this paper {proposes a hybrid beamforming {phased array} architecture for} active sensing applications. The {co-located} transmitting and receiving arrays have a fully connected hybrid architecture and may be sparse. {Furthermore,} we utilize image addition to synthesize PSFs that are usually only achieved by uniform arrays employing fully digital beamforming. {To the best of our knowledge, the resulting multi-image joint transmit-receive beampattern matching problem {has not been studied before}. {In particular, it essentially differs from the typical hybrid beamformer design problem, where the optimization of the transmitter and receiver is decoupled, and spectral efficiency is used as the objective function \cite{ayach2014spatially}.}}
	
	{The} main contributions of the paper are threefold:
	\begin{enumerate}
		\item We formulate an optimization problem {to jointly find} the hybrid transmit and receive beamformers achieving a desired PSF using as few component images as possible.
		\item We develop {a greedy algorithm} for approximately solving this non-convex {hybrid} {beamformer design} problem. In the {special case of the} fully digital {beamformer}, {we propose using} an alternating minimization {algorithm, which {is also partly utilized} in the hybrid case.}
		\item We derive {closed-form beamforming weights yielding} upper bounds on the number of component images {required by the hybrid and fully analog beamformers to match the beampattern of the fully digital beamformer.}
	\end{enumerate}
	We address the general case when the analog beamforming network consists of phase shifters with quantized phases. In a related work, we study the special case of a single front end connected to phase shifters with continuous phases \cite{rajamaki2019analog}.

	The paper is organized as follows. Section~\ref{sec:definitions} introduces the signal model and defines key concepts, such as the point spread function and the image addition method. Section~\ref{sec:problem} formulates the hybrid beamformer weight optimization problem. Section~\ref{sec:prior_work} reviews key prior work that will be utilized in Section~\ref{sec:numerical}, where we propose algorithms for approximately solving the hybrid beamforming problem in both the fully digital and hybrid cases. Section~\ref{sec:analytical} develops closed-form expressions for the hybrid beamforming weights, which provide upper bounds on the number of component images in the case of continuous and discrete phase shifts. Finally, Section~\ref{sec:examples} demonstrates the performance of the proposed solutions via simulations using both linear and planar arrays. In particular, we show that sparse hybrid beamformers with quantized phase shifters can achieve image quality comparable to uniform fully digital beamformers, at the expense of an increase in the number of transmissions and a reduction in array gain.
	
	\subsection{Notation}
	Matrices are denoted {using} bold uppercase, vectors {using} bold lowercase, and scalars {using} unbolded letters. The $(n,m) $th element of matrix $ \mathbf{A} $ is denoted {as} $ A_{nm} $. {If the matrix is indexed by a subscript, say as $ \mathbf{A}_i $, the $(n,m) $th element is denoted as $ [\mathbf{A}_i]_{nm} $.} Furthermore, the $ n $th row and $ m $th column of matrix $ \mathbf{A} $ are denoted as $ \mathbf{A}_{n,:} $ and $ \mathbf{A}_{:,m} $. Subscripts ``$ \txt{t} $'' and ``$ \txt{r} $'' denote transmitter and receiver, respectively. We omit these subscripts, or use {``$\txt{x}$'' to indicate either of them to} avoid unnecessary repetition whenever possible. The $ N $-dimensional vector of ones is denoted by $\mathbf{1}_N $, and the $ N\times N $ identity matrix by $ \mathbf{I}_N $ {(subscripts are omitted when the dimensions are clear from the context)}. The standard unit vector, consisting of zeros except for the $ i $th entry, {which is unity,} is denoted by $ \mathbf{e}_i $ (dimension specified separately). The indicator function is denoted by $ \mathbbm{1}(\cdot) $. The $ \ell_p $ and Frobenius norms are{, respectively,} denoted {as} $ \|\cdot\|_p$ and $ \|\cdot\|_\txt{F} $, where $ p\geq 1 $. Operators $ (\cdot)^\txt{T} $, $ (\cdot)^\txt{H} $, $ (\cdot)^* $, and $ (\cdot)^\dagger ${, respectively,} denote the matrix transpose, complex conjugate transpose, complex conjugate, and pseudo-inverse. The Kronecker and Khatri-Rao products are denoted by $ \otimes $ and $ \odot $. The $ \txt{vec}(\cdot) $ operator stacks the columns of its matrix argument into a column vector, whereas, $ \txt{mat}_{N\times M}(\cdot) $ reshapes an $ NM $ dimensional vector into a $ N\times M $ matrix. The $ \txt{diag}(\cdot) $ operator constructs a diagonal matrix of its vector argument. {Basic operations}, such as rounding to the nearest integer $ \lceil \cdot \rfloor $ or the {the angle of a complex-valued number $ \measuredangle \cdot $,} are applied elementwise to matrix arguments. Table~\ref{tab:symbols} lists the symbols that are referred most frequently in the text.
	\begin{table}[t]
		\caption{Frequently used symbols. Subscript ``$ \txt{x} $'' denotes either the receiver (``$ \txt{r} $'') or transmitter (``$ \txt{t} $'').}\label{tab:symbols}
		\centering
		\begin{tabular}{c|c}
			Symbol&Description\\
			\hline
			$ N_\txt{x} $&No. of array elements\\
			$ M_\txt{x} $&No. of front ends\\
			$ {B} $&No. of bits (per phase shifter)\\
			$ Q $&No. of component images\\
			$ V $&No. of {spatial} sampling directions (angles)\\
			$ N_\Sigma $&No. of sum co-array elements\\
			\hline
		\end{tabular}
	\end{table}

	\section{Signal model and definitions} \label{sec:definitions}
	In this section, we introduce the signal model and key definitions. In particular, we consider active narrowband sensing of coherent far-field point scatterers using linear processing at both the transmitter and receiver. We first define matrix $ \mathbf{F}_\txt{x} $, which models the analog beamforming network consisting of phase shifters. We then briefly recall the key concept of the point spread function, which characterizes the performance of a linear imaging system. We also review the image addition method that extends the set of point spread functions that are achievable by hybrid or sparse arrays. Finally, we {introduce} the sum co-array, {which is fundamental in determining the set of achievable point spread functions}.
	
	\subsection{Signal model}
	Consider a phased planar array imaging system that sequentially scans a scattering scene by transmitting and receiving focused beams of narrowband signals. Such systems are typically used in, e.g., medical ultrasound imaging or radar. Let $N_\txt{t}$ denote the number of transmitting (Tx) and $N_\txt{r}$ the number of receiving (Rx) array elements. The number of Tx and Rx front ends are reduced using analog preprocessing networks comprising of phase shifters, as shown in Fig.~\ref{fig:architecture}. Specifically, we use a bank of $M_\txt{t} {\leq} N_\txt{t}$ Tx front ends and $M_\txt{r} {\leq} N_\txt{r}$ Rx front ends. We refer to the beamforming architecture as 
	\begin{itemize}
		\item {\emph{fully digital}, when $ M_\txt{x} =N_\txt{x}$}
		\item {\emph{hybrid}, when $ 2 \leq M_\txt{x}< N_\txt{x} $}
		\item {\emph{fully analog}, when $ M_\txt{x} =1$.}
	\end{itemize}
	Both the hybrid and analog architectures are assumed to be fully connected, whereas the digital architecture is partially connected, since each sensor has a dedicated front end.
	
	We transmit a modulated narrowband pulse using the transmit beamforming weights $\mathbf{w}_\txt{t} = \mathbf{F}_\txt{t}\mathbf{c}_\txt{t} \in \mathbb{C}^{N_\txt{t}}$, where $\mathbf{c}_\txt{t} \in \mathbb{C}^{M_\txt{t}}$ denotes the digital weight vector, and $\mathbf{F}_\txt{t} \in \mathbb{C}^{N_\txt{t} \times M_\txt{t}} $ the analog phase shift matrix (see section~\ref{sec:F} for details). The transmitted radiation is reflected off scatterers in the field-of-view of the transmit array and {observed by} the receiver, where it is processed by a hybrid beamforming network with the beamforming weights $\mathbf{w}_\txt{r} = \mathbf{F}_\txt{r}\mathbf{c}_\txt{r} \in \mathbb{C}^{N_\txt{t}}$. Here $\mathbf{c}_\txt{r} \in \mathbb{C}^{M_\txt{r}}$ denotes the digital, and $\mathbf{F}_\txt{r} \in \mathbb{C}^{N_\txt{r} \times M_\txt{r}}$ the analog beamforming weights of the receiver. The beamformed signal is then {processed} using a digital matched filter yielding 
	\begin{align}
	y(\mathbf{u})&=  \mathbf{w}_\txt{r}^\txt{T}(\mathbf{u}) {\mathbf{A}_\txt{r}\boldsymbol{\Gamma}\mathbf{A}_\txt{t}^\txt{T}}\mathbf{w}_\txt{t}(\mathbf{u})+\mathbf{w}_\txt{r}^\txt{T}(\mathbf{u})\mathbf{n},\label{eq:gamma_hat}
	\end{align}
	where $ \mathbf{u}\!\in\!\mathbb{R}^3$ is the scan direction {taking the form} $ \mathbf{u}\!=\![\sin\varphi\sin\theta, \cos\varphi\sin\theta, \cos\theta ]^\txt{T}$, when the array is focused in the far-field. Here $ \varphi\!\in\![-\pi/2,\pi/2]$ and $\theta\!\in\![0,\pi] $ are the azimuth and elevation angles of the scan direction, {respectively}. Matrix $ \boldsymbol{\Gamma} = \text{diag}(\boldsymbol{\gamma}) \in \mathbb{C}^{K\times K}$ is a diagonal matrix, with $\boldsymbol{\gamma} = [\gamma_1, \ldots, \gamma_K]^\txt{T} \in \mathbb{C}^K$ containing the {scattering coefficients} of the $ K $ {reflectors}, and $ \mathbf{n}\in\mathbb{C}^{N_\txt{r}}$ is a vector of {spatio-temporally white complex circular Gaussian noise with zero mean and covariance matrix $ \sigma^2\bm{I} $}. {Furthermore,} the {$ N_\txt{x} \times K $} steering matrix {of the} {Tx or Rx} {array is}
	\begin{align}
	\mathbf{A}_{\txt{x}}\!=\![\mathbf{a}_\txt{x}(\mathbf{v}_1), \ldots, \mathbf{a}_\txt{x}(\mathbf{v}_K)], \label{eq:Ax}
	\end{align}
	{where} the steering vector $  \mathbf{a}_{\txt{x}}(\mathbf{v}_k)\!\in\!\mathbb{C}^{N_\txt{x}}$ {is} evaluated in scatterer direction $ \{\mathbf{v}_k \in\mathbb{R}^3\}_{k=1}^K $. When the scatterers are located in the far-field of both the transmitting and receiving array, we have $ \mathbf{v}_k = [\sin\varphi_k\sin\theta_k, \cos\varphi_k\sin\theta_k, \cos\theta_k]^\txt{T}$, where $ \varphi_k\in[-\pi/2,\pi/2]$ and $\theta_k\in[0,\pi] $ denote the azimuth and elevation angles of the $ k $th scatterer. {Eq.~\eqref{eq:gamma_hat}, or its magnitude $ |y(\mathbf{u})| $, may be interpreted as an image of the scattering scene. Typically, this image is evaluated for a discrete set of steering directions, i.e., pixels.} 
	 
	\subsection{Analog phase shift matrix $ \mathbf{F}_\txt{x} $} \label{sec:F}
	The entries of the analog phase shift matrix $ \mathbf{F}_{\txt{x}} \in \mathscr{F}_\txt{x}({B})$ are complex exponentials with discrete phases. Specifically, let
	\begin{align}
	\mathscr{F}_\txt{x}({B}) &= \{\mathbf{F}\!=\!\exp({j\boldsymbol{\Phi}})\ |\ \boldsymbol{\Phi}\in\mathbb{R}^{N_\txt{x}\times M_\txt{x}}, \Phi_{nm}\!\in\!\varPhi({B}) \}, \label{eq:B_quantized}\\
	\varPhi({B}) &= \{0,2\pi/2^{{B}},\ldots,(2^{B}-1)2\pi/2^{{B}}\}, \label{eq:phases_set}
	\end{align}
	where the exponential function in \eqref{eq:B_quantized} is applied elementwise, and $ {B} $ denotes the number of bits used to uniformly quantize the phase of each entry of $ \mathbf{F} $ over the interval $ [0,2\pi) $. Note that \eqref{eq:phases_set} ensures that $ \varPhi({B}+1) \supset \varPhi({B})$, and thereby $ \mathscr{F}_\txt{x}({B}+1) \supset \mathscr{F}_\txt{x}({B})$. It also follows from \eqref{eq:phases_set} that the phase quantization operator $ {\mathcal{P}_{{B}}}(\boldsymbol{\Psi}) $, i.e., the projection of the elements of some matrix $ \boldsymbol{\Psi}\in[0,2\pi)^{N\times M} $ to set $ \varPhi({B}) $, {can be expressed as}
	\begin{align}
	{\mathcal{P}_{{B}}}(\boldsymbol{\Psi})\! =\frac{\pi}{2^{B-1}} \Big\lceil \frac{2^{B-1}}{\pi} \boldsymbol{\Psi} \Big\rfloor {\bmod\ 2\pi}.\label{eq:proj_Phi}
	\end{align}
	Letting the number of bits go to infinity yields the special case of continuous phase shifters: $ \mathscr{F}_\txt{x}(\infty) =\lim_{{B}\to\infty}\mathscr{F}_\txt{x}({B})  $; $ \varPhi(\infty)=\lim_{{B}\to\infty}\varPhi({B}) = [0,2\pi) $; and $ {\mathcal{P}_{\infty}}(\boldsymbol{\Psi}) = \boldsymbol{\Psi} $.
	
	\subsection{Point spread function}\label{sec:psf}
	The \emph{point spread function} (PSF) {fully defines the} spatial impulse response of a linear imaging system, {and {is} key in characterizing the achievable resolution and interference suppression capability. The effective PSF of an active array} is the product of the Tx and Rx PSFs, {as illustrated in Fig.~\ref{fig:psf_concept}}. Specifically, for {the array focused in the direction $ \mathbf{u}\!\in\!\mathbb{R}^3$, and a unit reflectivity point-scatterer in the direction $ \mathbf{v}\!\in\!\mathbb{R}^3$,} the PSF is defined as
	\begin{align*}
	{\psi(\mathbf{u},\mathbf{v}) =(\mathbf{a}_\txt{t}(\mathbf{v})\otimes \mathbf{a}_\txt{r}(\mathbf{v}))^\txt{T}\text{vec}(\mathbf{W}(\mathbf{u})),}
	\end{align*}
	where $ {\mathbf{W}(\mathbf{u})=\mathbf{w}_\txt{r}(\mathbf{u})\mathbf{w}_\txt{t}^\txt{T}(\mathbf{u})} \in\mathbb{C}^{N_\txt{r}\times N_\txt{t}} $ is a rank-1 matrix. {For a fixed $\mathbf{u}$ and a discrete set of $ V $ scatterer directions $ \{\mathbf{v}_i\}_{i=1}^V $ {determined by the desired imaging region and resolution,} the PSF may be expressed as a vector $\boldsymbol{\psi} \in\mathbb{C}^{V}$ satisfying}
	\begin{align}
	{\boldsymbol{\psi} = \mathbf{A}^\txt{T}\txt{vec}(\mathbf{W}).} \label{eq:psf_vec}
	\end{align}
	{Here the $ i $th column of the \emph{effective steering matrix} $ \mathbf{A}_{}\in\mathbb{C}^{N_\txt{r}N_\txt{t}\times V} $ is given by the Kronecker product of the Tx and Rx steering vectors evaluated in direction $\mathbf{v}_i $, i.e., $ \mathbf{a}_\txt{t}(\mathbf{v}_i)\otimes\mathbf{a}_\txt{r}(\mathbf{v}_i) $. Consequently, $ \mathbf{A} $ can be expressed as the Khatri-Rao product of the Tx and Rx steering matrices $ \mathbf{A}_\txt{x}\in\mathbb{C}^{N_\txt{x}\times V} $ in \eqref{eq:Ax}:
	\begin{align}
	\mathbf{A} = \mathbf{A}_\txt{t}\odot\mathbf{A}_\txt{r}.\label{eq:A}
	\end{align}
	Any feasible PSF $ \boldsymbol{\psi} $ thus lies in the row space of $\mathbf{A}$. For examples of typical PSFs, see Section~\ref{sec:num_def_psf_det}.}
\begin{figure}[t]
		\centering
		\includegraphics[width=.7\textwidth]{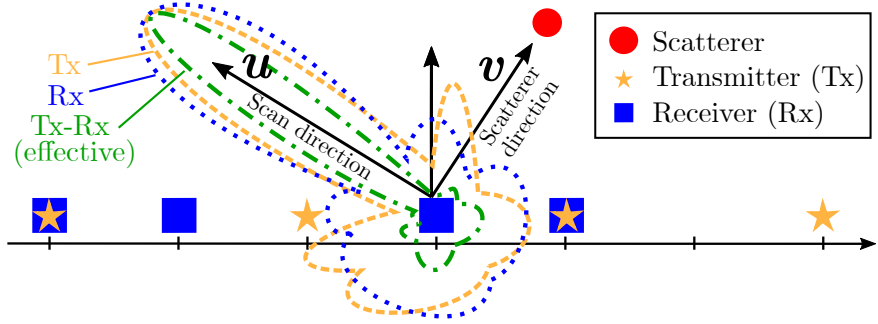}
	\caption{The PSF is the spatial response of the array to a point scatterer in direction $\mathbf{v}$, with the array steered in direction $ \mathbf{u} $. The effective PSF is the product of {the} Tx and Rx array PSFs.}\label{fig:psf_concept}
\end{figure}
		
	\subsection{Image addition}
	A single Tx-Rx weight pair {$\{\mathbf{w}_\txt{t},\mathbf{w}_\txt{r}\}$} may not always suffice to achieve a desired PSF. In this case, the {range of feasible PSFs may be extended} by \emph{image addition} \cite{hoctor1990theunifying}. Image addition synthesizes a composite image {with improved resolution and} {lower} {side lobe levels, by summing together several \emph{component images} that} are formed using different Tx-Rx weight pairs. {This corresponds to using a rank-$ Q $ \emph{co-array weight matrix} $\mathbf{W}\in \mathbb{C}^{N_\txt{r}\times N_\txt{t}} $ in \eqref{eq:psf_vec}, defined as \cite{kozick1991linearimaging}:}
	\begin{align}
	\mathbf{W} &= \sum_{q=1}^Q \mathbf{w}_{\txt{r},q}\mathbf{w}_{\txt{t},q}^\txt{T}=\mathbf{W}_\txt{r}\mathbf{W}_\txt{t}^\txt{T}, \label{eq:W_w} 
	\end{align}
	where $ \mathbf{W}_\txt{x}= [\mathbf{w}_{\txt{x},1},\ldots,\mathbf{w}_{\txt{x},Q}] \in\mathbb{C}^{N_\txt{x}\times Q}$. Each {rank-$ 1 $ matrix} $ \mathbf{w}_{\txt{r},q}\mathbf{w}_{\txt{t},q}^\txt{T} $ in \eqref{eq:W_w} corresponds to a transmission and reception with a different pair of Tx and Rx weight vectors $\mathbf{w}_{\txt{t},q}$ and $ \mathbf{w}_{\txt{r},q}$. These vectors may be found from the \emph{singular value decomposition} (SVD) of matrix $ \mathbf{W} $ in the case of {a fully digital beamformer} \cite{kozick1991linearimaging}. The smaller the \emph{number of component images} $ Q $ is, the shorter the image acquisition time, since fewer {transmissions} are required {in forming} an image. In the case of hybrid beamforming, \eqref{eq:W_w} becomes
	\begin{align}
	\mathbf{W}\!=\!\sum_{q=1}^Q \mathbf{F}_{\txt{r},q}\mathbf{c}_{\txt{r},q}\mathbf{c}_{\txt{t},q}^\txt{T}\mathbf{F}_{\txt{t},q}^\txt{T}\!=\! \mathbf{F}_\txt{r}({\mathbf{I}}\odot \mathbf{C}_\txt{r})({\mathbf{I}}\odot \mathbf{C}_\txt{t})^\txt{T}\mathbf{F}_\txt{t}^\txt{T},\label{eq:W_Bg} 
	\end{align}
	where {matrices $ \mathbf{F}_\txt{x}\!=\![\mathbf{F}_{\txt{x},1},\ldots,\mathbf{F}_{\txt{x},Q}] \in	\mathscr{F}_\txt{x}({B})\!\subset\!\mathbb{C}^{N_\txt{x}\times M_\txt{x}Q}$ and $ \mathbf{C}_\txt{x}\!=\![\mathbf{c}_{\txt{x},1},\ldots,\mathbf{c}_{\txt{x},Q}]\!\in\!\mathbb{C}^{M_\txt{x}\times Q}$, {respectively,} contain the analog and digital beamforming weights of all {the} $ Q $ component images.} {Note that \eqref{eq:W_Bg} is not necessarily unique, as is shown in Section~\ref{sec:bounds_h_cont}. Furthermore, although $\mathbf{c}_{\txt{x},q}$ could also be defined as a real-valued vector, allowing for complex-valued entries {is} more convenient for our purposes.}
	
	\subsection{Sum co-array}\label{sec:sum co-array}
	The sum co-array, {$\mathcal{D}_\Sigma$,} is a virtual array structure {defined as the set of} pairwise sums of the transmit and receive {element positions $ \mathcal{D}_\txt{x}=\{\mathbf{d}_{\txt{x},1},\ldots,\mathbf{d}_{\txt{x},N_\txt{x}}\}\subset \mathbb{R}^3$:}
	\begin{align}
	\mathcal{D}_\Sigma &= \{\mathbf{d}_\txt{t}+\mathbf{d}_\txt{r}\ |\ \mathbf{d}_\txt{x}\in \mathcal{D}_\txt{x}\}.\label{eq:co-array_supp}
	\end{align}
	The sum co-array ultimately determines the set of PSFs that the array can achieve \cite{hoctor1990theunifying}. The utility of the sum co-array stems from the fact that it has at least as many virtual elements $ N_\Sigma= |\mathcal{D}_\Sigma| $ as either of the physical arrays, since $N_\Sigma\geq N_\txt{t}+N_\txt{r}-1$. If the transceivers are co-located, that is, {$\mathcal{D}_\txt{t}=\mathcal{D}_\txt{r} $ and} $ N_\txt{x}=N $, a simple counting argument shows that $N_\Sigma \leq N(N+1)/2$. The array is redundant if $N_\Sigma < N(N+1)/2$. If the transmitting and receiving elements can be placed independently of each other, the redundancy condition is $ N_\Sigma < N_\txt{t}N_\txt{r}$.
	
	Steering matrix $\mathbf{A}$ in \eqref{eq:A} may have fewer than $ N_\txt{t}N_\txt{r} $ unique rows when the array configuration is redundant. In particular, the number of unique rows equals the number of sum co-array elements $ N_\Sigma $, for example, when the array elements are identical and mutual coupling is negligible. The unique rows of $\mathbf{A}$ are therefore contained in the \emph{{sum co-array} steering matrix} $ \mathbf{A}_\Sigma\in\mathbb{C}^{N_\Sigma\times V} $ satisfying
	\begin{align}\label{eq:A_S}
	\mathbf{A} =\boldsymbol{\Upsilon}^\txt{T}\mathbf{A}_\Sigma.
	\end{align}		
	Here $  \boldsymbol{\Upsilon}$ is {an $ N_\Sigma\times N_\txt{t}N_\txt{r} $} binary matrix {mapping} the set of virtual elements arising from the Kronecker structure in {\eqref{eq:A}} to the set of sum co-array elements in \eqref{eq:co-array_supp}. The sum co-array and {the multiplicities of its element} uniquely determine $\boldsymbol{\Upsilon}$.
	\begin{definition}[{Sum co-array} selection matrix]\label{def:sel_mat}
		Map $\boldsymbol{\Upsilon}:\mathbb{C}^{N_\txt{t}N_\txt{r}} \mapsto \mathbb{C}^{N_\Sigma}$ is a binary matrix $ \boldsymbol{\Upsilon}\in\{0,1\}^{N_\Sigma\times N_\txt{t}N_\txt{r}} $, where
		\begin{align*}
		\Upsilon_{n,m}
		=\begin{cases}
		1,\txt{ if } \mathbf{d}_{\Sigma,n} = \mathbf{d}_ {\txt{t},\lceil m/N_\txt{r} \rceil}+\mathbf{d}_ {\txt{r},1+(m-1)\bmod N_\txt{r}}\\
		0,\txt{ otherwise}.
		\end{cases}
		\end{align*} 
		Here $ \mathbf{d}_{\txt{x},i}\in \mathcal{D}_\txt{x} $ denote the elements of the physical array and $ \mathbf{d}_{\Sigma,n}\in\mathcal{D}_\Sigma $ the elements of the sum co-array.
	\end{definition}
	{Any feasible PSF may be expressed as $\boldsymbol{\psi}\!=\!\mathbf{A}_\Sigma^\txt{T}\mathbf{w}_\Sigma$,} where 
	\begin{align}
	\mathbf{w}_\Sigma = \boldsymbol{\Upsilon}\,\txt{vec}(\mathbf{W})\label{eq:w_sigma}
	\end{align}
	{is the $ N_\Sigma $-dimensional \emph{sum co-array beamforming weight vector.} Note that if $ V<N_\Sigma $, then \eqref{eq:w_sigma} is only sufficient for satisfying \eqref{eq:psf_vec}. If $ \txt{rank}(\mathbf{A}_\Sigma)=N_\Sigma $, which holds only if $ V\geq N_\Sigma $, then \eqref{eq:w_sigma} is also necessary.} When the physical array is a uniform array with $N$ co-located transceivers, then $N_\Sigma\propto N $. Conversely, a sparse array {may have} $ N_\Sigma\propto N^2 $. {Consequently, these arrays would typically require (at least) $ V\propto N $, respectively, $ V\propto N^2 $ angular samples.}

	\section{Problem formulation}\label{sec:problem}
	{In this section,} we formulate the hybrid beamformer weight optimization problem {as a {spatial} filter design problem. Our fidelity measure of choice is the {$ \ell_2 $ approximation error}
	\begin{align*}
	\varepsilon = \|\boldsymbol{\psi}-\mathbf{A}^\txt{T}\txt{vec}(\mathbf{W})\|_2^2. 
	\end{align*}
	Here $ \boldsymbol{\psi}\in\mathbb{C}^V $ is the desired PSF {(sampled in $ V $ directions)}, which is also assumed to be feasible (in the {row} space of {$\mathbf{A}$}). The realized PSF $ \mathbf{A}^\txt{T}\txt{vec}(\mathbf{W}) $ is determined by {the measurement (or steering) matrix $ \mathbf{A} $ given by \eqref{eq:A}, and} the co-array weight matrix $\mathbf{W}\in\mathbb{C}^{N_\txt{r}\times N_\txt{t}}$, which factorizes as \eqref{eq:W_Bg}.}
	
	The objective is to minimize the number of component images $ Q$, while achieving a desired PSF. This leads to the following non-convex optimization problem\footnote{{We do not constrain the transmit power for simplicity. However, the Rx and Tx weight vectors are normalized post optimization as $\mathbf{w}_{\txt{r},q}\gets \mathbf{w}_{\txt{r},q} \|\mathbf{w}_{\txt{t},q}\|_\infty $ and $\mathbf{w}_{\txt{t},q}\gets \mathbf{w}_{\txt{t},q}\|\mathbf{w}_{\txt{t},q}\|^{-1}_\infty $ to ensure that the transmitters are operated at saturation and SNR is maximized for each of the $ Q $ component images.}}:
	\begin{opteq}
		&\underset{\{\mathbf{F}_{\txt{x},q}\in \mathscr{F}_\txt{x}({B}),\mathbf{c}_{\txt{x},q}\in\mathbb{C}^{M_\txt{x}}\}_{q=1}^{Q\in\mathbb{N}_+}}{\text{minimize}} Q\nonumber \\
		&\text{subject to}\  \Bigg\|\boldsymbol{\psi}_{}\!-\!\mathbf{A}^\txt{T}_{}\text{vec}\Bigg(\sum_{q=1}^{Q} \mathbf{F}_{\txt{r},q}\mathbf{c}_{\txt{r},q}\mathbf{c}_{\txt{t},q}^\txt{T}\mathbf{F}_{\txt{t},q}^\txt{T}\Bigg)\Bigg\|_2^2\!\leq\!\varepsilon_{\max}. \label{p:h}
	\end{opteq}
	In \eqref{p:h}, $ \varepsilon_{\max} \geq 0 $ is an {approximation} error tolerance parameter, and $ \mathscr{F}_\txt{x} $ denotes the analog weight matrix constraint set in \eqref{eq:B_quantized}. The fact that $ Q $ is unknown further complicates problem \eqref{p:h}. If we instead fix $ Q $, we obtain the following slightly simpler optimization problem:
	\begin{opteq}
		\underset{\begin{subarray}{c}
				\{\mathbf{F}_{\txt{x},q}\in \mathscr{F}_\txt{x}({B}),\ \\
				\quad \mathbf{c}_{\txt{x},q}\in\mathbb{C}^{M_\txt{x}}\}_{q=1}^{Q}
		\end{subarray}}{\text{minimize}} & \Bigg\|\boldsymbol{\psi}_{}-\mathbf{A}^\txt{T}_{}\text{vec}\Bigg(\sum_{q=1}^{Q} \mathbf{F}_{\txt{r},q}\mathbf{c}_{\txt{r},q}\mathbf{c}_{\txt{t},q}^\txt{T}\mathbf{F}_{\txt{t},q}^\txt{T}\Bigg)\Bigg\|_2^2.\label{p:h_alt}	
	\end{opteq}
	{Note that $ Q $ determines the number of optimization variables in both optimization problems, which implies that the optimal value of the objective function of \eqref{p:h_alt} is non-increasing in $ Q $. We may therefore} recover the solution to \eqref{p:h} from \eqref{p:h_alt} by {finding} the smallest $ Q $ for which the objective function of \eqref{p:h_alt} does not exceed {the approximation error {tolerance}} $ \varepsilon_{\max} $. {This can easily be accomplished using binary search (bisection) at a small additional cost, given a search interval of feasible values of $ Q $. Consequently, we will henceforth focus on \eqref{p:h_alt} instead of  \eqref{p:h}.} If we know the weight vectors $ \mathbf{w}_{\txt{x},q} \in\mathbb{C}^{N}$, which may be the case when a fully digital solution is available, we may solve the even simpler optimization problem
	\begin{opteq}
		\underset{\mathbf{F}\in \mathscr{F}({B}),\mathbf{c}\in\mathbb{C}^{M}}{\text{minimize}}\ \|\mathbf{w}-\mathbf{F}\mathbf{c}\|_2^2 \label{p:h_d}
	\end{opteq}
	 {independently for the transmitter, receiver}, and each {component image} (we omit the subscripts in \eqref{p:h_d} for simplicity). Problem \eqref{p:h_d} recovers a {hybrid} solution to \eqref{p:h}, if a fully digital solution to \eqref{p:h} is available, and if this solution {can be factorized as in} \eqref{eq:W_Bg} for the same number of component images as in the {fully} digital case. {However, since this is generally not the case, \eqref{p:h_alt} needs to be solved instead.}
	
	\section{Key results in prior work}\label{sec:prior_work}
	In this section, we review two key results {(used in Sections~\ref{sec:numerical} and \ref{sec:analytical})} related to solving optimization problem \eqref{p:h_d} using hybrid and analog beamformers with continuous phase shifters. {We again} omit subscripts $ \txt{x} $ {and $ q $} for simplicity, since the results are independently applicable to both the transmitter and receiver, as well as any component image.
	
	Zhang et al. {\cite{zhang2005variable}} showed that two front ends with continuous phase shifts are sufficient for factorizing any $ \mathbf{w}\in\mathbb{C}^N$ as $ \mathbf{w} =\mathbf{F}\mathbf{c} $, thus optimally solving \eqref{p:h_d} when $ M\geq 2 $. The hybrid beamforming weights can {also} be expressed in closed form, as shown by the following lemma adapted\footnote{This is a reformulation of \cite[Theorem~1 and Appendix~A]{zhang2005variable}, where we give a slightly more general expression for the phase of $ \mathbf{F} $ in Appendix~\ref{proof:thm:M2F2_zhang}.} from \cite{zhang2005variable}.
	\begin{lemma}[Solution to \eqref{p:h_d} using two front ends and continuous phase shifters \cite{zhang2005variable}] \label{thm:M2F2_zhang}
		Let $ M\!=\!2$ and $ {B}\!\to\!\infty$. Given any $ \mathbf{w}\in\mathbb{C}^N $, {an} optimal solution to \eqref{p:h_d} achieving $ \mathbf{w}=\mathbf{Fc}$, where $ \mathbf{c}\in\mathbb{C}^2 $ and $ \mathbf{F}\!\in\!\mathscr{F}(\infty)$ following \eqref{eq:B_quantized}, is given by
		{\begin{align}
		\mathbf{F} &=\exp\Bigg(j\Big(\measuredangle\mathbf{w}\mathbf{1}_2^\txt{T}+\cos^{-1}\Big(\frac{|\mathbf{w}|}{\|\mathbf{w}\|_\infty}\Big)(\mathbf{1}_2-2\mathbf{e}_2)^\txt{T}\Big)\Bigg) \label{eq:B_opt}\\
		\mathbf{c} &= \frac{\|\mathbf{w}\|_\infty}{2}\mathbf{1}_2.\label{eq:g_opt}
		\end{align}
		Here $\mathbf{1}_2$ is a vector of ones, $ \mathbf{e}_2 $ is the standard unit vector, {and} $ \cos^{-1}$, $ |\cdot| $, and {the angle operator} $ \measuredangle $ are applied elementwise.}
	\end{lemma}
	\begin{proof}\let\qed\relax 
		See Appendix~\ref{proof:thm:M2F2_zhang}.
	\end{proof}
	{Fig.~\ref{fig:zhang} illustrates Lemma~\ref{thm:M2F2_zhang}, showing how each entry $ w_n \in\mathbb{C}$ of vector $\mathbf{w}\in\mathbb{C}^N $ can be expressed as the sum of two phasors with magnitude $ \|\mathbf{w}\|_\infty/2$ and phases depending on $ w_n $ and $  \|\mathbf{w}\|_\infty $. Note that \eqref{eq:B_opt} and \eqref{eq:g_opt} are not unique, since any $c_1=c_2\geq  \|\mathbf{w}\|_\infty/2 $ yields a feasible factorization $ \mathbf{w}=\mathbf{F}\mathbf{c} $. In general, unequal magnitudes $ |c_1|\neq |c_2| $ are also possible if $ \min_n|w_n|>0 $ holds (see Appendix~\ref{proof:thm:M2F2_zhang}). We may also {easily} extend Lemma~\ref{thm:M2F2_zhang}} to the case $ M>2 $ by appending zeros to $ \mathbf{c} $ and columns with arbitrary phases to $\mathbf{F}$ \cite{zhang2005variable}.
	\begin{figure}[t]
			\centering
			\includegraphics[width=.7\linewidth]{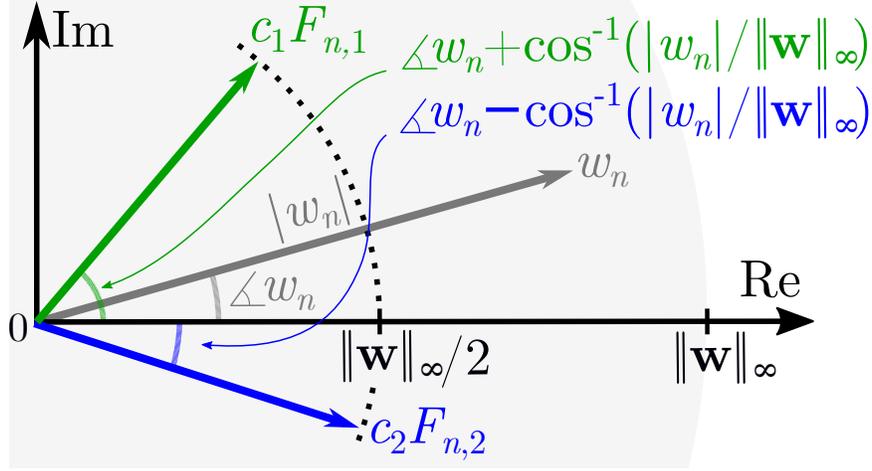}
		\caption{Illustration of Lemma~\ref{thm:M2F2_zhang}. Any point within the shaded disk can be expressed as the sum of two phasors with fixed equal magnitudes and appropriately chosen phases.}\label{fig:zhang}
	\end{figure}
	
	Lemma~\ref{thm:M2F2_zhang} implies that the number of front ends $ M $ required to implement any fully digital weight vector $ \mathbf{w}\in\mathbb{C}^N $ is independent of the number of array elements $ N $, provided continuous phase shifts are used. The number of front ends may actually be reduced to just one, if the digital weight vector $ \mathbf{c} \in\mathbb{C}^M$ can be selected as the scaled unit vector $ \mathbf{c}={c}\mathbf{1}_M$, where $ {c}\in\mathbb{C} $ \cite{zhang2014onachieving,sohrabi2016hybrid,bogale2016onthenumber}{, as in \eqref{eq:g_opt}}. However, a modification to the canonical fully-connected architecture is required. Namely, all $ MN $ phase shifters need to be connected to a single front end, as explained in the following remark.
	\begin{remark}[Analog beamformer with modified architecture \cite{zhang2014onachieving,sohrabi2016hybrid,bogale2016onthenumber}]\label{thm:M1F2}
		Consider a fully connected hybrid beamformer with $ M $ front ends connected to $ N $ phase shifters each. If $ \mathbf{c}={c}\mathbf{1}_M, {c}\in\mathbb{C} $, we may form an equivalent analog beamformer with a single front end connected to all $ NM $ phase shifters.
	\end{remark}
	Perfect factorization of $ \mathbf{w} $ is not generally possible, if $ M=1 $ and the number of phase shifters equals the number of array elements $N $. Nevertheless, \eqref{p:h_d} actually admits a closed-form solution when $ {B}\to \infty ${, as shown by the following lemma.}
	\begin{lemma}[Solution to \eqref{p:h_d} using single front end and continuous phase shifters] \label{thm:M1F1_approx}
		Let $ M\!=\!1$ and $ {B}\!\to\!\infty$. Given any $ \mathbf{w}\!\in\!\mathbb{C}^N $, an optimal solution to \eqref{p:h_d} that minimizes $ \|\mathbf{w}\!-\!{c}\mathbf{f}\|_2^2,$ where ${c}\!\in\!\mathbb{C} $ and $ \mathbf{f}\!\in\!\mathscr{F}(\infty)$ following \eqref{eq:B_quantized}, is given by
		\begin{align}
		\mathbf{f} &=\exp(j\measuredangle \mathbf{w}) \label{eq:b_opt_1}\\
		{c} &= \|\mathbf{w}\|_1/N.\label{eq:g_opt_1}
		\end{align}
		Furthermore, the optimal value of \eqref{p:h_d} is $\|\mathbf{w}\|_2^2-\|\mathbf{w}\|_1^2/N$.
	\end{lemma}
	\begin{proof}\let\qed\relax 
		See Appendix~\ref{proof:thm:M1F1_approx}.
	\end{proof}
	{In general, Lemma~\ref{thm:M1F1_approx} only yields an approximate factorization $ \mathbf{w}\approx c\mathbf{f} $. Equality $\mathbf{w}=c\mathbf{f}$ holds if and only if the entries of $\mathbf{w}$ have equal magnitude, i.e., $ |w_1|=|w_2|=\ldots=|w_N| $.}
	
	\section{Algorithms for finding beamformer weights} \label{sec:numerical}
	In this section, we develop three algorithms for approximately solving optimization problems \eqref{p:h_alt} and \eqref{p:h_d}. In the fully digital beamforming case, we address \eqref{p:h_alt} using alternating minimization (Algorithm~\ref{alg:altmin_d}). In the hybrid and analog beamforming cases, we use a greedy approach to approximately solve both \eqref{p:h_d} (Algorithm~\ref{alg:greedylemma}) and \eqref{p:h_alt} (Algorithm~\ref{alg:greedy}). Table~\ref{tab:numerical} summarizes the proposed algorithms. 
	\begin{table*}[h]
		\resizebox{1\textwidth}{!}{
		\centering
		\begin{threeparttable}[t]
			\centering
			\caption{Summary of beamformer weight optimization algorithms proposed in Section~\ref{sec:numerical}.}\label{tab:numerical}
			\begin{tabular}{c|c|c|c|c}
				Algorithm &Beamforming architecture &Problem &Worst case complexity\tnote{a}, $ \mathcal{O}(\cdot) $& Explanation\\
				\hline
				\ref{alg:altmin_d}&Digital &\eqref{p:d_altmin}&${k_{\max}VQN\max(V,QN)}$&Alternating minimization\\
				\ref{alg:greedylemma}&Hybrid or Analog &\eqref{p:h_d}&$  {M^2N^2} $&Greedy + Lemma~\ref{thm:M2F2_zhang} and \ref{thm:M1F1_approx}\\
				\ref{alg:greedy}&Hybrid or Analog &\eqref{p:h_alt}&${k_{\max}VQ\max(NV,MQV,M^2Q^2)+QN^2M^2}$ &Greedy + {Algorithm~\ref{alg:altmin_d} and} \ref{alg:greedylemma} + Alt. min.\\
				 \hline
			\end{tabular}
			\begin{tablenotes}
				\item[a] Assuming $ N_\txt{x} \propto N; M_\txt{x} \propto M; $ {and $V\geq N $}.
			\end{tablenotes}
		\end{threeparttable}
		}
	\end{table*}

	\subsection{Fully digital beamformer}\label{sec:algorithms_d}
	{The fully digital beamformer serves as a natural {baseline for the hybrid beamformer.} This is because the digital {beamformer} imposes the fewest constraints on the solutions {to the optimization problems in Section~\ref{sec:problem}}. In the following, we develop an alternating minimization algorithm for solving \eqref{p:h_alt} in the fully digital case. We will also utilize this algorithm in the hybrid case developed in Section~\ref{sec:alg_h}.}
	\subsubsection{Alternating minimization}
	Digital beamformer design is substantially simplified by the {absence of phase shifters,} {which reduces} \eqref{p:h_alt} to the {following} biconvex problem
	\begin{opteq}
	\underset{\mathbf{W}_\txt{x}\in\mathbb{C}^{N_\txt{x}\times Q}}{\text{minimize}}\  \|\boldsymbol{\psi}-\mathbf{A}^\txt{T}\text{vec}(\mathbf{W}_\txt{r}\mathbf{W}_\txt{t}^\txt{T})\|_2^2.\label{p:d_altmin}
	\end{opteq}
	{Here} the columns of $ \mathbf{W}_\txt{x}=[\mathbf{w}_{\txt{x},1},\mathbf{w}_{\txt{x},2},\ldots,\mathbf{w}_{\txt{x},Q}] \in\mathbb{C}^{N_\txt{x}\times Q}$ are the unknown weight vectors, each corresponding to a specific component image. Problem \eqref{p:d_altmin} is non-convex due to the product of the unknown matrices $ \mathbf{W}_\txt{r} $ and $ \mathbf{W}_\txt{t} $. However, we may find a local minimum of \eqref{p:d_altmin} in a straightforward fashion by \emph{alternating minimization}. {The \emph{low rank matrix sensing problem} \eqref{p:d_altmin} was actually studied by Jain et. al \cite{jain2013low} in a more general (albeit real-valued) setting. Next, we {describe in} detail a slightly modified version of their ``AltMinSense'' algorithm, adapted to the beamforming application considered in this paper.} {The alternating minimization algorithm, {summarized in Algorithm~\ref{alg:altmin_d}},} starts with an initial guess for $ \mathbf{W}_\txt{t} $ and proceeds by computing the least squares solutions:
	\begin{align}
	\mathbf{W}_\txt{r} &= \txt{mat}_{N_\txt{r}\times Q}((\mathbf{A}^\txt{T}(\mathbf{W}_\txt{t}\otimes \mathbf{I}_{N_\txt{r}}))^\dagger\boldsymbol{\psi}_{}) \label{eq:Wr}\\	
	\mathbf{W}_\txt{t} &= \txt{mat}_{Q\times N_\txt{t}}^\txt{T}((\mathbf{A}^\txt{T}(\mathbf{I}_{N_\txt{t}}\otimes \mathbf{W}_\txt{r}))^\dagger\boldsymbol{\psi}_{}).\label{eq:Wt}
	\end{align}
	Equations \eqref{eq:Wt} and \eqref{eq:Wr} are iteratively solved until a desired error $ \varepsilon_{\max} $ or maximum number of iterations $ k_{\max} $ is achieved. Although alternating minimization is guaranteed to converge to a local minimum, which local minimum is found depends on the initialization. We choose to use the spectral initialization $ \mathbf{W}\!=\!\sum_{v=1}^V\psi_v \mathbf{a}_{\txt{r},v}\mathbf{a}_{\txt{t},v}^\txt{T}\!=\!\sum_{v=1}^V\psi_v\txt{mat}_{N_\txt{r}\times N_\txt{t}}({\mathbf{A}_{:,v}})$ to confine the initialization to the solution subspace \cite{jain2013low}. We then initialize $ \mathbf{W}_\txt{t} $ using the right singular vectors corresponding to the $ Q $ largest singular values of $ \mathbf{W} $. {Alternatively, multiple {different} initializations of $ \mathbf{W}_\txt{t}  $ could be used to increase the chances of finding a good local, or even global, minimum.} Note that when $ Q\!=\!\min(N_\txt{r},N_\txt{t}) $, we may simply obtain $ \mathbf{W}_\txt{t} $ and $ \mathbf{W}_\txt{r} $ from the SVD of the least squares solution $ \mathbf{W}\!=\!\txt{mat}_{N_\txt{r}\times N_\txt{t}}((\mathbf{A}^\txt{T})^\dagger\boldsymbol{\psi}) $.
	\begin{algorithm}[t]
		\caption{Digital beamformer: alternating min. {for} \eqref{p:d_altmin}}
		\label{alg:altmin_d}
		\begin{algorithmic}[1]		
			\Procedure{AltMin}{$\mathbf{A},\boldsymbol{\psi}_{},Q,k_{\max},\varepsilon_{\max}$}
			\State $\mathbf{W} \gets \sum_{v=1}^V\psi_v\txt{mat}_{N_\txt{r}\times N_\txt{t}}({\mathbf{A}_{:,v}})$ \Comment{initialization \cite{jain2013low}}
			\State $\{\mathbf{U},\boldsymbol{\Sigma},\mathbf{V}\}\gets$ \Call{SVD}{$ \mathbf{W} $,$ Q $} \Comment{{$ Q $ princ. {components}}}
			\State $ \{\mathbf{W}_\txt{t},k,\varepsilon\}\gets \{\mathbf{V}^*,0,\infty\} $
			\While{$k < k_{\max} \wedge \varepsilon > \varepsilon_{\max}$}
			\State Update $ \mathbf{W}_\txt{r} $ and $ \mathbf{W}_\txt{t} $ using \eqref{eq:Wr} and \eqref{eq:Wt} \label{line:d_Wx}
			\State $ \varepsilon \gets
			\|\boldsymbol{\psi}_{}-\mathbf{A}^\txt{T}\text{vec}(\mathbf{W}_\txt{r}\mathbf{W}_\txt{t}^\txt{T})\|_2^2$
			\State $ k\gets k+1 $
			\EndWhile
			\State \Return $ \mathbf{W}_\txt{r},\mathbf{W}_\txt{t}$
			\EndProcedure
		\end{algorithmic}
	\end{algorithm}
	
	\subsubsection{Worst-case complexity}\label{sec:d_complexity}
	The most expensive operation in Algorithm~\ref{alg:altmin_d} is {on line~\ref{line:d_Wx}}, {where} {the {worst-case} time complexity of computing} the pseudo-inverse {{is proportional to that of the} SVD. In general, the SVD of a full $ m$ by $ n $ matrix has complexity $\mathcal{O}(m^2n)$, where $ m\geq n $ \cite[p.~493]{golub2013matrix}.} The worst-case complexity of Algorithm~\ref{alg:altmin_d} {is therefore $\mathcal{O}(k_{\max}m^2n)$, where $ m=\max(V,QN_{\max}), n= \min(V,QN_{\max}) $, and $ N_{\max} = \max(N_\txt{r},N_\txt{t}) $.} Assuming that $ N_\txt{x} \propto N$, the {complexity of Algorithm~\ref{alg:altmin_d} simplifies to 
	\begin{align*}
		\mathcal{O}(k_{\max}VQN\max(V,QN)).
	\end{align*}
	As {indicated} in Section~\ref{sec:sum co-array}, typically $ V \propto N $ for a uniform array, and $ V \propto N^2 $ for a sparse array. This {simplifies} the complexity of Algorithm~\ref{alg:altmin_d} to $\mathcal{O}(k_{\max}Q^2N^3)$, or $\mathcal{O}(k_{\max}QN^5)$, respectively. In either case, the complexity is at most on the order of $N^5$ or $N^6 $, since $ Q\leq N $.} This may {become} prohibitively large for an array with hundreds of elements. {Note that the number of iterations $ k \leq k_{\max}$ required by Algorithm~\ref{alg:altmin_d} also depends on the properties of matrix $\mathbf{A}$ and the desired error tolerance $ \varepsilon_{\max} $ \cite{jain2013low}. A more detailed complexity analysis is however out of the scope of this paper.}
	
	\subsection{Hybrid beamformer} \label{sec:alg_h}
	We will show in Section~\ref{sec:analytical} that it is possible to construct a hybrid beamformer that solves \eqref{p:h} using continuous phase shifts and exactly two front ends, if a fully digital solution to \eqref{p:h} is given. In case of discrete phase shifts, it may be sufficient to quantize this hybrid solution, provided that the number of {phase shift} bits is sufficiently large. However, poor results may follow if {too} few bits are used \cite{sohrabi2016hybrid}. Furthermore, this approach does not benefit from increasing $ M_\txt{x} > 2$. 
	
	\subsubsection{Greedy subroutine}
	As a first step towards addressing the issues outlined above, we consider problem \eqref{p:h_d}, and adopt a greedy method {in Algorithm~\ref{alg:greedylemma}} to approximately solve it. Starting with a fully digital solution $\mathbf{w} $, we initialize the residual as $\mathbf{w}'\!=\!\mathbf{w} $, and apply Lemma~\ref{thm:M2F2} to find $ \mathbf{F}\!\in\!\mathscr{F}(\infty)\!\subset\!\mathbb{C}^{N\times 2} $ satisfying $ \mathbf{w}'\!=\!\mathbf{F}\mathbf{1}_2 \|\mathbf{w}'\|_\infty /2$. Using \eqref{eq:proj_Phi}, we quantize the phase of $ \mathbf{F} $, yielding $ \mathbf{F}\!=\!\exp(j{\mathcal{P}_{{B}}}(\measuredangle \mathbf{F}))\!\in\!\mathscr{F}({B})\!\subset\!\mathbb{C}^{N\times 2}$. We then compute the least squares solution of the digital weights $\mathbf{c} $ before updating the residual on line~\ref{line:residual_gl}. The process is repeated a total of $ \lfloor M/2 \rfloor$ times. In case $ M $ is odd, we quantize the analog least squares solution given by Lemma~\ref{thm:M1F1_approx} in the final iteration. The solution found by Algorithm~\ref{alg:greedylemma} becomes exact, i.e., $ \mathbf{w}'\!\to\!0 $, when either (i) ${B}\!\geq\!1 , M\!\to\!\infty $, or (ii) ${B}\!\to\!\infty$, $ M\!\geq\!2$. We note that Algorithm~\ref{alg:greedylemma} is similar to \cite[Algorithm~1]{koochakzadeh2018beam}, although we consider a different problem with quantized phase shifts, as well as utilize Lemma~\ref{thm:M2F2_zhang}. {Other algorithms addressing \eqref{p:h_d}, or modifications thereof, also exist \cite{chen2017hybrid}. However, solving \eqref{p:h_d} alone is insufficient for solving the problem of interest \eqref{p:h_alt}, as we will show next. A detailed comparison of alternatives to Algorithm~\ref{alg:greedylemma} is therefore beyond the scope of this paper and left for future work.}
	\begin{algorithm}[h]
		\caption{Hybrid beamformer: greedy {method} {for} \eqref{p:h_d}}
		\label{alg:greedylemma}
		\begin{algorithmic}[1]		
			\Procedure{{GreedySub}}{$\mathbf{w},M,{B}$}
			\State $ \mathbf{w}'\gets \mathbf{w} $
			\For{$m\in\{1,2,\ldots,\lfloor M/2 \rfloor\}$} \Comment{Lemma~\ref{thm:M2F2}}
			\State {$\boldsymbol{\Phi}\gets \measuredangle\mathbf{w}'\mathbf{1}_2^\txt{T}+\cos^{-1}(|\mathbf{w}'|/\|\mathbf{w}'\|_\infty)(\mathbf{1}_2-2\mathbf{e}_2)^\txt{T} $}
			\State $ \mathbf{F}_{:,(2m-1):2m}\gets \exp(j  {\mathcal{P}_{{B}}}(\boldsymbol{\Phi})) $   \Comment{{quantization}}
			\State $ \mathbf{c}_{1:2m}\gets  \mathbf{F}_{:,1:2m}^\dagger\mathbf{w}$ \Comment{{LS solution}} \label{line:pinv_h}
			\State $\mathbf{w}'\gets \mathbf{w}-\mathbf{F}_{:,1:2m}\mathbf{c}_{1:2m}$\Comment{update residual} \label{line:residual_gl}
			\EndFor
			\If{$M \bmod 2 = 1$} \Comment{Lemma~\ref{thm:M1F1_approx}}
			\State $\mathbf{F}_{:,M}\gets \exp(j{\mathcal{P}_{{B}}}(\measuredangle \mathbf{w}')) $ \label{line:update_odd}
			\State $\mathbf{c}\gets  \mathbf{F}^\dagger\mathbf{w}$
			\EndIf
			\State \Return $\mathbf{F},\mathbf{c}$
			\EndProcedure	
		\end{algorithmic}
	\end{algorithm}
	
	\subsubsection{Greedy main routine} \label{sec:greedy}
	{If} the number of component images of the hybrid beamformer {is} $ Q {\leq\txt{rank}(\mathbf{W})} $, we can directly apply Algorithm~\ref{alg:greedylemma} to each fully digital weight {vector} $ \{\mathbf{w}_{\txt{x},q}\}_{q=1}^{{\txt{rank}(\mathbf{W})}} $ and find an approximate solution to \eqref{p:h_alt}. However, this solution does not improve if $ Q $ is increased beyond $ {\txt{rank}(\mathbf{W})} $. Consequently, we propose Algorithm~\ref{alg:greedy}, which uses Algorithm~\ref{alg:altmin_d} and \ref{alg:greedylemma} to iteratively compute {and quantize} {the rank-1 matrix $\mathbf{W}$ that} {minimizes the $ \ell_2 $-norm of the residual PSF vector $ \boldsymbol{\psi}'\in\mathbb{C}^{V} $}. The residual is initialized as {the desired PSF $ \boldsymbol{\psi}'=\boldsymbol{\psi}$, and updated at the end of each iteration by subtracting the $ q $-component realized PSF from $ \boldsymbol{\psi} $. } The hybrid weights {of the $ q $th iteration} $ \mathbf{F}_{\txt{x},q}\!\in\!\mathbb{C}^{N_\txt{x} \times M_\txt{x}}$ and $\mathbf{c}_{\txt{x},q}\!\in\!\mathbb{C}^{M_\txt{x} }$ {are {found} by} applying Algorithm~\ref{alg:greedylemma} to the {fully digital single component solution} obtained by {calling Algorithm~\ref{alg:altmin_d} with $ Q=1 $. Finally,} the digital weights $ \{c_{\txt{x},l}\}_{l=1}^q $ are recomputed by solving \eqref{p:h_alt} with $ \mathbf{F}_\txt{x} $ fixed, {i.e., the following problem:}
	\begin{opteq}	
		\underset{\mathbf{C}_{\txt{x}}\in\mathbb{C}^{M_\txt{x}\times {q}}}{\text{minimize}}\ \|\boldsymbol{\psi}_{}\!-\!\mathbf{A}^\txt{T}\text{vec}(\mathbf{F}_\txt{r}({\mathbf{I}}\!\odot\!\mathbf{C}_\txt{r})({\mathbf{I}}\!\odot\!\mathbf{C}_\txt{t})^\txt{T}\mathbf{F}_\txt{t}^\txt{T})\|^2_2. \label{p:d_altmin_G}	
	\end{opteq}
	Problem \eqref{p:d_altmin_G} is biconvex, since $ \txt{vec}(\mathbf{W}) $ can be rewritten as
	\begin{align*}
	\txt{vec}(\mathbf{W}) 
	&= ((\mathbf{F}_{\txt{t}}({\mathbf{I}}\odot \mathbf{C}_{\txt{t}}) \otimes \mathbf{1}_{M_\txt{r}}^\txt{T}) \odot \mathbf{F}_{\txt{r}}) \txt{vec}(\mathbf{C}_{\txt{r}})\\
	&=  (\mathbf{F}_{\txt{t}} \odot (\mathbf{F}_{\txt{r}}({\mathbf{I}}\odot \mathbf{C}_{\txt{r}}) \otimes \mathbf{1}_{M_\txt{t}}^\txt{T}) ) \txt{vec}(\mathbf{C}_{\txt{t}}).
	\end{align*}
    This follows from {\eqref{eq:W_Bg} and the easily verifiable} identities (i) $\txt{vec}(\mathbf{Xx}\mathbf{y}^\txt{T}\mathbf{Y}^\txt{T})\!=\!(\mathbf{Yy1}^\txt{T}\!\odot\!\mathbf{X})\mathbf{x}\!=\!(\mathbf{Y}\!\odot\!\mathbf{Xx1}^\txt{T})\mathbf{y}$, and (ii) $\sum_{{l}=1}^{{q}}\mathbf{Z}_{{l}}\mathbf{z}_{{l}}\!=\![\mathbf{Z}_1,\ldots,\mathbf{Z}_{{q}}][\mathbf{z}_1^\txt{T} ,\ldots,\mathbf{z}_{{l}}^\txt{T}]^\txt{T}$ after {simplifications}. A local minimum of \eqref{p:d_altmin_G} {is} found using alternating minimization, which iterates between the following least squares solutions:
	\begin{align}
	\mathbf{C}_\txt{r} &= \txt{mat}_{M_\txt{r}\times {q}}((\mathbf{A}^\txt{T}((\mathbf{F}_{\txt{t}}({\mathbf{I}} \odot \mathbf{C}_{\txt{t}}) \otimes \mathbf{1}_{M_\txt{r}}^\txt{T}) \odot \mathbf{F}_{\txt{r}}))^\dagger\boldsymbol{\psi}) \label{eq:Cr}\\
	\mathbf{C}_\txt{t} &= \txt{mat}_{M_\txt{t}\times {q}}((\mathbf{A}^\txt{T}(\mathbf{F}_{\txt{t}} \odot (\mathbf{F}_{\txt{r}}({\mathbf{I}} \odot \mathbf{C}_{\txt{r}}) \otimes \mathbf{1}_{M_\txt{t}}^\txt{T}) )^\dagger\boldsymbol{\psi}).\label{eq:Ct}
	\end{align}
	{Algorithm~\ref{alg:greedy} does not necessarily converge to the fully digital solution found by Algorithm~\ref{alg:altmin_d}, even as $ B\to\infty $. This is because Algorithm~\ref{alg:altmin_d} is called sequentially using a single component image (instead of $ Q $) on line~\ref{line:greedy_altmin} of Algorithm~\ref{alg:greedy}. Consequently, when $ B $ is large, a better solution may be obtained by applying Algorithm~\ref{alg:greedylemma} to the fully digital solution found by Algorithm~\ref{alg:altmin_d} (provided $ M\geq2 $). However, {we will show} {in Section~\ref{sec:num_val}} {that} Algorithm~\ref{alg:greedy} produces better results in the interesting regime of small {to} moderate $ B $ {(say, $ B\leq 6 $)}.}
	\begin{algorithm}[t]
		\caption{Hybrid beamformer: greedy method {for} \eqref{p:h_alt}}
		\label{alg:greedy}
		\begin{algorithmic}[1]		
			\Procedure{Greedy}{$\mathbf{A},\boldsymbol{\psi},M_\txt{r},M_\txt{t},{B},Q,k_{\max},\varepsilon_{\max}$}
			\State {$ \boldsymbol{\psi}'\gets\boldsymbol{\psi} $}
			\For{$q\in\{1,2,\ldots,Q\}$}	
			\State {$\{\mathbf{w}_\txt{r},\mathbf{w}_\txt{t}\} \gets $ \Call{AltMin}{$\mathbf{A},\boldsymbol{\psi}',1,k_{\max},\varepsilon_{\max}$}}\label{line:greedy_altmin}
			\For{{$\txt{x}\in\{\txt{t},\txt{r}\}$}}
			\State $ \{\mathbf{F}_{\txt{x},q},\mathbf{c}_{\txt{x},q}\} \gets$ \Call{{GreedySub}}{${\mathbf{w}_\txt{x}},M_\txt{x},{B} $}\label{line:gl_x}
			\State {$ \mathbf{F}_\txt{x} \gets [\mathbf{F}_{\txt{x},1},\ldots,\mathbf{F}_{\txt{x},q}]$}
			\EndFor
			\State $ \{\mathbf{C}_\txt{t},k,\varepsilon\}\gets \{{[\mathbf{c}_{\txt{t},1},\ldots,\mathbf{c}_{\txt{t},q}]},0,\infty\} $
			\While{$k < k_{\max} \wedge \varepsilon > \varepsilon_{\max}$}\Comment{{alternating min.}}
			\State Update $ \mathbf{C}_\txt{r} $ and $\mathbf{C}_\txt{t}$ using \eqref{eq:Cr} and \eqref{eq:Ct} \label{line:Cx_altmin}
			\State {$ \boldsymbol{\psi}' \gets \boldsymbol{\psi}-\mathbf{A}^\txt{T}\text{vec}(\mathbf{F}_\txt{r}({\mathbf{I}}\odot \mathbf{C}_\txt{r})({\mathbf{I}}\odot \mathbf{C}_\txt{t})^\txt{T}\mathbf{F}_\txt{t}^\txt{T})$}\label{line:residual_psi}
			\State $ {\{\varepsilon,k\}\gets \{\|\boldsymbol{\psi}'\|_2^2,k+1\}} $
			\EndWhile
			\State \Return $ \mathbf{F}_{\txt{r}},\mathbf{F}_{\txt{t}},\mathbf{C}_{\txt{r}},\mathbf{C}_{\txt{t}} $
			\EndFor
			\EndProcedure
		\end{algorithmic}
	\end{algorithm}
	
	\subsubsection{Worst-case complexity}\label{sec:h_complexity}
	{The {worst-case} complexity of Algorithm~\ref{alg:greedylemma} is dominated by the pseudo-inverse on line~\ref{line:pinv_h}. Since the dimension of matrix $\mathbf{F}_{:,1:2m}$ changes at each iteration, the total complexity of Algorithm~\ref{alg:greedylemma} is proportional to}
	\begin{align*}
	{\sum_{m=1}^{\lfloor M/2\rfloor }N^2m = \mathcal{O}(N^2M^2).}
	\end{align*}
	{The most expensive operations in Algorithm~\ref{alg:greedy} are (i) calling Algorithm~\ref{alg:altmin_d} on line~\ref{line:greedy_altmin}; (ii) calling Algorithm~\ref{alg:greedylemma} on line~\ref{line:gl_x}; (iii) evaluating the pseudo-inverses on line~\ref{line:Cx_altmin}; and (iv) computing the matrix-vector product on line~\ref{line:residual_psi}. Assuming again that $ N_\txt{x}\propto N $ and $ M_\txt{x}\propto M $, the complexity of (i) reduces to $ \mathcal{O}(k_{\max}VNQ\max(V,N))$, or $\mathcal{O}(k_{\max}V^2NQ)$ if $ V\geq N $. Furthermore, (ii) has complexity $\mathcal{O}(QN^2M^2)$. {Denoting $ m= \max(V,qM)$, $ n= \min(V,qM)$, and $\varrho=\min(V/M,Q)$,} the complexity of (iii) is proportional to
	 \begin{align*}
	 {k_{\max}\sum_{q=1}^Q m^2n}&=k_{\max}VM\Big(V\sum_{q=1}^{\lfloor \varrho \rfloor} q+M\sum_{q=\lfloor \varrho \rfloor+1}^Q q^2\Big)\\
	 &= \mathcal{O}(k_{\max}VM\max(V\varrho^2,M(Q^3-\varrho^3))).
	\end{align*}
	 {This is upper bounded by $ \mathcal{O}(k_{\max}VMQ^2\max(V,MQ)) $, since $ \varrho\!\in\![0,Q]$}. Finally, the cost of (iv) is $\mathcal{O}(k_{\max}VN^2Q)$. {Consequently,} the {worst-case} complexity of Algorithm~\ref{alg:greedy} {is}
	 \begin{align*}
	 \mathcal{O}(k_{\max}VQ\max(VN,VMQ,M^2Q^2)+N^2M^2Q).
	 \end{align*}
 	{As} $ V\propto N $ for a uniform array {and} $ V\propto N^2 $ for a sparse array, the complexity of Algorithm~\ref{alg:greedy} is at most on the order of $ N^3 $ (uniform array) or $ N^5 $ (sparse array), and $ Q^3 $ (both).} 
 	
	\subsection{Fully analog beamformer}
	Algorithms~\ref{alg:greedylemma} and \ref{alg:greedy} are directly applicable to {fully} analog beamformer design {by setting $ M_\txt{x}\!=\!1 $}. {We note that problems~\eqref{p:h} and \eqref{p:h_alt} also simplify significantly in the analog case. This was exploited in the companion paper \cite{rajamaki2019analog} to develop a more efficient algorithm when $ B\to\infty $. Investigations into improved fully analog beamforming algorithms for finite $ B $ are however beyond the scope of this paper.}
	
	\subsection{Remarks on the computational complexity} \label{sec:remarks_on_comp_complexity}
	{We conclude this section with two remarks regarding the computational complexity of solving the optimization problems formulated in Section~\ref{sec:problem}.} 
	
	{Firstly,} {we may have to solve \eqref{p:h} for several steering directions $\mathbf{u}$ {in practice}. {For example, if the number of phase shift bits $ B $ is small, the number of desired steering directions may be larger than that {accommodated by the $ 2^B $ quantized phase shifts}. Even if $ B $ is large, the} PSF {may} not {be} translation invariant, {as}, for {instance}, when sensors {have directive gain patterns} or scatterers are located in the near field of the array.}
	
	{Secondly, {sparsity can be leveraged to speed up computations. For example, \eqref{eq:Wr}--\eqref{eq:Ct} contain sparse matrices resulting from the Kronecker and Khatri-Rao products with the identity matrix. This can be exploited when computing the pseudo-inverse using, e.g., \emph{power} or \emph{orthogonal iterations} \cite[pp.~366-368]{golub2013matrix}, especially if there are only a few dominant singular values. Furthermore,} solving \eqref{p:h} for a desired co-array weight vector in \eqref{eq:w_sigma}, instead of the PSF in \eqref{eq:psf_vec} {allows us to} replace $\boldsymbol{\psi}$ by $ \mathbf{w}_\Sigma $, $ \mathbf{A} $ by $ \boldsymbol{\Upsilon}^\txt{T} $, and $ V $ by $ N_\Sigma $ in Algorithms~\ref{alg:altmin_d} and \ref{alg:greedy}. The computational advantage follows from the fact that $ \boldsymbol{\Upsilon} $ is sparse with only $ N_\txt{t}N_\txt{r} $ non-zero entries, which is a factor of $ V $ less than the $ VN_\txt{t}N_\txt{r} $ entries of full matrix $ \mathbf{A} $. The solutions obtained using $\boldsymbol{\Upsilon}$ {and} $\mathbf{A}$ are equivalent when the {sum co-array} steering matrix $\mathbf{A}_\Sigma$ in \eqref{eq:A_S} is a (scaled) unitary matrix. This is the case, e.g., when the columns of $ \mathbf{A} $ are sampled uniformly in $ V=N_\Sigma $ unique directions, and the array elements are ideal and omnidirectional.}
	
	\section{Bounds on the number of component images $ Q $}\label{sec:analytical}
	In this section, we derive closed-form solutions for $ \mathbf{F}_{\txt{x},q},\mathbf{c}_{\txt{x},q} $ assuming {zero approximation error} ($ \varepsilon_{\max} = 0 $). {These solutions then yield} upper bounds on $ Q $ in problem~\eqref{p:h} for $ M_\txt{x}\in\{1,2,N_\txt{x}\}$ and $ {B} \in \{1,\infty\}$. Each {beamformer} makes a different trade-off between the number of front ends $ M_\txt{x}$, phase shift bits $ {B} $, and component images $ Q $, as summarized in Table~\ref{tab:array_configs}. We find that for any number of phase shift bits $ {B} $, the number of component images required by a hybrid beamformer satisfies $ {\txt{rank}(\mathbf{W})}\leq Q\leq N_\txt{r}N_\txt{t} $, where $ {\txt{rank}(\mathbf{W})} \leq \min(N_\txt{r},N_\txt{t})$ is the number of component images required by the fully digital beamformer. Similarly, for the fully analog beamformer we have $ {\txt{rank}(\mathbf{W})}\leq Q\leq 4N_\txt{r}N_\txt{t}$. 
	\begin{table*}
	\resizebox{1\textwidth}{!}{
		\centering
		\begin{threeparttable}[t]
		\centering
		\caption{Properties of closed-form beamformers {proposed} in Section~\ref{sec:analytical}.}\label{tab:array_configs}
		\begin{tabular}{c|c|c|c|c|c}
			Theorem & \makecell{Beamforming architecture} &\makecell{\# of front ends, $ M_\txt{x}$} &\makecell{\# of phase shifters}&\makecell{\# of {phase shift} bits, $ {B} $}&\makecell{\# of component images\tnote{b}, $ Q $}\\
			\hline
			n/a&Digital &{$ N_\txt{x}$}&{$ 0 $}&n/a&{${\txt{rank}(\mathbf{W})}$}\\
			\hdashline
			\ref{thm:M2F2}&Hybrid&{$ 2$}&{$ 2N_\txt{x}$}&{$ \infty $}&{${\txt{rank}(\mathbf{W})} $}\\
			\ref{thm:M2F2J1}&Hybrid&{$ 2 $}&{$ 2N_\txt{x}$}&{$ 1 $}&{$N_\txt{r}N_\txt{t} $}\\
			\ref{thm:M1F1Jinf}&Analog&{$ 1$}&{$ N_\txt{x} $}&{$ \infty $}&{$4\, {\txt{rank}(\mathbf{W})}$}\\
			\ref{thm:M1F1J1}&Analog&{$ 1 $}&{$ N_\txt{x} $}&{$ 1 $}&{$4N_\txt{r}N_\txt{t}$}\\
			\hdashline
			\ref{thm:M2F2} + Remark~\ref{thm:M1F2}&Analog\tnote{c}&{$ 1 $}&{$ 2N_\txt{x} $}&{$ \infty $}&{${\txt{rank}(\mathbf{W})}$}\\
			\ref{thm:M2F2J1} + Remark~\ref{thm:M1F2}&Analog\tnote{c}&{$ 1 $}&{$ 2N_\txt{x} $}&{$ 1 $}&{$ N_\txt{r}N_\txt{t}$}\\
			\hline
		\end{tabular}
			\begin{tablenotes}
				\item[b] {Note that} ${\txt{rank}(\mathbf{W})}\leq  \max(N_\txt{r},N_\txt{t})$.
				\item[c] Requires modification to the architecture in Fig.~\ref{fig:architecture} (see Remark~\ref{thm:M1F2} in Section~\ref{sec:prior_work}).
			\end{tablenotes}
		\end{threeparttable}
		}
	\end{table*}
	
	\subsection{Fully digital beamformer}\label{sec:bounds_d}
	In the case of fully digital beamforming, the SVD guarantees that any co-array weight matrix $ \mathbf{W}\in\mathbb{C}^{N_\txt{r}\times N_\txt{t}} $ in \eqref{eq:W_w} can be factorized using {$ Q $ component images, where
	\begin{align}
	Q=\text{rank}(\mathbf{W})\leq \min(N_\txt{t},N_\txt{r}). \label{eq:Q_d_0}
	\end{align}
	We may also obtain a lower bound on $ Q $ by considering the number of degrees of freedom available for realizing a desired PSF. Specifically, assuming \eqref{eq:A_S} holds, a simple comparison of the number of equations and unknowns in \eqref{eq:w_sigma} yields the following necessary condition on $ Q $.
	\begin{proposition}[Lower bound on $ Q $]\label{thm:Q_d}
		{Let $\mathbf{W}\in\mathbb{C}^{N_\txt{r}\times N_\txt{t}}$ be a rank-$ Q $ matrix. Then \eqref{eq:w_sigma} holds for any $\mathbf{w}_\Sigma\in \mathbb{C}^{N_\Sigma}$ only if}
		\begin{align}
		{Q \geq \frac{N_\txt{t}+N_\txt{r}-\sqrt{(N_\txt{t}+N_\txt{r})^2-4N_\Sigma}}{2}.} \label{eq:Q_d}
		\end{align}
	\end{proposition}
	\begin{proof}
		In \eqref{eq:w_sigma}, the number of equations is $ N_\Sigma\!\leq\!N_\txt{t}N_\txt{r} $, and the number of free variables is $ Q(N_\txt{t}+N_\txt{r}-Q) $, since $\mathbf{W}$ is a rank-$ Q $ matrix. A necessary condition for \eqref{eq:w_sigma} to hold for any $\mathbf{w}_\Sigma\!\in\!\mathbb{C}^{N_\Sigma}$ is therefore that $Q(N_\txt{t}+N_\txt{r}-Q)\!\geq\!N_\Sigma$. This quadratic inequality directly yields the bound in \eqref{eq:Q_d}.
	\end{proof}
	 Eq.~\eqref{eq:Q_d} is a necessary lower bound {on $ Q $} only when \eqref{eq:w_sigma} is required to hold for any co-array beamforming {weight} vector $\mathbf{w}_\Sigma\in\mathbb{C}^{N_\Sigma}$. For a fixed $\mathbf{w}_\Sigma$, a rank-$ Q $ matrix $\mathbf{W}$ may exist that satisfies \eqref{eq:w_sigma} but not \eqref{eq:Q_d}.}
	
	{It is instructive to evaluate \eqref{eq:Q_d} for some {typical} array configurations. Firstly, if the array is non-redundant, every Tx-Rx pair uniquely maps to one sum co-array element, which implies that $ N_\Sigma = N_\txt{t}N_\txt{r} $. Consequently, \eqref{eq:Q_d} reduces $ Q\geq \min(N_\txt{t},N_\txt{r}) $, which together with \eqref{eq:Q_d_0} yields
		\begin{align*}
			Q=\min(N_\txt{t},N_\txt{r}).
		\end{align*}
	Secondly, if the transceivers are co-located, we have {$ N_\txt{x}=N $}, and \eqref{eq:Q_d} simplifies to
	\begin{align}
	Q\geq N-\sqrt{N^2-N_\Sigma}. \label{eq:Q_d_colocated}
	\end{align}
	For example, the {\emph{uniform linear array}} has $ N_\Sigma = 2N-1 $, which yields $ Q\geq 1 $. Consequently, a single component image may suffice to achieve any PSF supported on the sum co-array of {this array}. Similar results can be shown to hold for higher dimensional uniform arrays, such as the \emph{uniform rectangular array}. In contrast, \eqref{eq:Q_d_colocated} scales linearly with $ N $ for sparse arrays, {since} $ N_\Sigma \propto \zeta N^2 $, with $0<\zeta < 1$. {This} leads to the inequality $ Q\gtrapprox N(1-\sqrt{1-\zeta})$ that holds approximately for large $ N $. In practice, the linear dependence between $ N $ and $ Q $ is weak. For instance, the square \emph{boundary array} \cite{hoctor1990theunifying} with $ \zeta = 1/4$ yields the bound $ Q\gtrapprox N(2-\sqrt{3})/2> 0.13 N$.} {The {greatest} lower bound in \eqref{eq:Q_d_colocated} is given by the array configuration maximizing $N_\Sigma$ given $ N $. This is the \emph{minimum-redundancy array} \cite{moffet1968minimumredundancy,hoctor1996arrayredundancy,kohonen2014meet,kohonen2018planaradditive}, if the sum co-array is required to be uniform.}

	\subsection{Hybrid beamformer}
	It is not evident for which values of {$ Q,M $ and $ B $} factorization \eqref{eq:W_Bg} is feasible, given a general co-array weight matrix $ \mathbf{W}\in\mathbb{C}^{N_\txt{r}\times N_\txt{t}} $. Next, we show that $ M_\txt{x}\!=\!2 $ Tx/Rx front ends are sufficient for feasibility, irrespective of the number of phase shifter bits $ {B} $. 
	\subsubsection{Continuous phase shifters} \label{sec:bounds_h_cont}
	Lemma~\ref{thm:M2F2_zhang} implies that \eqref{eq:W_Bg} provides a feasible factorization when $ Q=\text{rank}(\mathbf{W}) $, provided $ M_\txt{x} = 2 $ and $ {B}\to\infty $. In this case, the hybrid beamforming weights are given by the following theorem:
	\begin{thm}[{Hybrid beamformer, continuous phase shifters, two Tx/Rx front ends}] \label{thm:M2F2}
		Let $ M_\txt{x}\!=\!2$ and ${B}\!\to\!\infty$. Any $ \mathbf{W}\!=\!\sum_{q=1}^{Q}\mathbf{w}_{\txt{r},q}\mathbf{w}_{\txt{t},q}^\txt{T}\in\mathbb{C}^{N_\txt{r}\times N_\txt{t}} $ may be factorized as $ \mathbf{W}\!=\!\sum_{q=1}^{Q} \mathbf{F}_{\txt{r},q}\mathbf{c}_{\txt{r},q}\mathbf{c}_{\txt{t},q}^\txt{T}\mathbf{F}_{\txt{t},q}^\txt{T} $, with $ \mathbf{c}_{\txt{x},q}\!\in\!\mathbb{C}^2 $; and $ \mathbf{F}_{\txt{x},q}\!\in\!\mathscr{F}_\txt{x}(\infty)$ following \eqref{eq:B_quantized}. For example, a valid factorization is
		\begin{align}
		{\mathbf{F}_{\txt{x},q}} \!&{=\!\exp\Bigg(j\Big(\measuredangle\mathbf{w}_{\txt{x},q}\mathbf{1}_2^\txt{T}\!+\!\cos^{-1}\Big(\frac{|\mathbf{w}_{\txt{x},q}|}{\|\mathbf{w}_{\txt{x},q}\|_\infty}\Big)(\mathbf{1}_2\!-\!2\mathbf{e}_2)^\txt{T}\Big)\Bigg)} \label{eq:Bx_opt_inf_q}\\
		\mathbf{c}_{\txt{x},q} &=\frac{ \|\mathbf{w}_{\txt{x},q}\|_\infty}{2}\mathbf{1}_2. \label{eq:gx_opt_inf_q}
		\end{align}
	{Here $\mathbf{1}_2$ is a vector of ones, $ \mathbf{e}_2 $ is the standard unit vector, {and} $ \measuredangle $, $ \cos^{-1}, $ and $ |\cdot| $ are applied elementwise.}
	\end{thm}
	\begin{proof}
		{This follows directly from Lemma~\ref{thm:M2F2_zhang}, since each $ \mathbf{w}_{\txt{x},q} $ can be factorized as $ \mathbf{w}_{\txt{x},q}=\mathbf{F}_{\txt{x},q}\mathbf{c}_{\txt{x},q} $.}
	\end{proof}
	{Note that the factorization in Theorem~\ref{thm:M2F2} is not unique, {and more general expressions for $\mathbf{F}_{\txt{x},q}$ and $ \mathbf{c}_{\txt{x},q}$ are easily obtained (see Lemma~\ref{thm:M2F2_zhang} in Section~\ref{sec:prior_work} and the proof in Appendix~\ref{proof:thm:M2F2_zhang}).}}
	
	\subsubsection{One-bit phase shifters}
	The phases of the phase shifters may be coarsely quantized in practice \cite{molisch2017hybrid}. In this case, Theorem~\ref{thm:M2F2} no longer holds even approximately. However, any co-array weight matrix $ \mathbf{W}\in \mathbb{C}^{N_\txt{r}\times N_\txt{t}} $ can still be achieved using only two Tx/Rx front ends and one-bit phase quantization. This is accomplished at the expense of increasing the number of component images to $ Q=N_\txt{r}N_\txt{t} \gg \min(N_\txt{r},N_\txt{t})\geq \text{rank}(\mathbf{W})$. The hybrid weight matrices in \eqref{eq:W_Bg} are again obtained in closed form, as shown by the following theorem.
	\begin{thm}[{Hybrid beamformer, $ 1 $-bit phase shifters, two Tx/Rx front ends}] \label{thm:M2F2J1}
		Let $ M_\txt{x}\!=\!2$ and ${B}\!=\!1$. Any $ \mathbf{W}\!\in\!\mathbb{C}^{N_\txt{r}\times N_\txt{t}} $ may be factorized as $ \mathbf{W}\!=\!\sum_{q=1}^{N_\txt{r}N_\txt{t}} \mathbf{F}_{\txt{r},q}\mathbf{c}_{\txt{r},q}\mathbf{c}_{\txt{t},q}^\txt{T}\mathbf{F}_{\txt{t},q}^\txt{T} $, with $ \mathbf{c}_{\txt{r},q},\mathbf{c}_{\txt{t},q}\!\in\!\mathbb{C}^2 $, and $ \mathbf{F}_{\txt{x},q}\!\in\!\mathscr{F}_\txt{x}(1)$ following \eqref{eq:B_quantized}. {For example, a valid factorization is}
		\begin{align}
		\mathbf{F}_{\txt{x},q} &=[\mathbf{1}_{N_\txt{x}}, 2\mathbf{e}_{n_\txt{x}}-\mathbf{1}_{N_\txt{x}}] \label{eq:Bx_opt_q}\\
		\mathbf{c}_{\txt{x},q} &= \frac{\sqrt{{W}_{n_\txt{r}n_\txt{t}}}}{2}\mathbf{1}_{2}, \label{eq:gx_opt_q}
		\end{align}
		where $ n_\txt{r}\!=\!1\!+\!(q\!-\!1) \bmod N_\txt{r}$ and $n_\txt{t}\!=\!\lceil q/N_\txt{r}\rceil $.
	\end{thm}
	\begin{proof}
		{Eq.~\eqref{eq:Bx_opt_q} ensures that} each of the $ Q\!=\!N_\txt{r}N_\txt{t} $ terms in \eqref{eq:W_Bg} contribute to exactly one {entry} of matrix $ \mathbf{W}\!\in\!\mathbb{C}^{N_\txt{r}\times N_\txt{t}}  $. {In particular, substituting} \eqref{eq:Bx_opt_q} and {$ {c}_{\txt{x},q}\mathbf{1}_{2} $} into \eqref{eq:W_Bg} {yields} $ \mathbf{F}_{\txt{x},q}\mathbf{c}_{\txt{x},q}\!=\!2{c}_{\txt{x},q}\mathbf{e}_{n_\txt{x}}${, where} $ \mathbf{e}_{n_\txt{x}}\!\in\!\{0,1\}^{N_\txt{x}}$ is the standard unit vector of length $ N_\txt{x} $ {with a unit entry at index $ n_\txt{x}\in\{1,2,\ldots,N_\txt{x}\} $}. Consequently, the $q $th term in \eqref{eq:W_Bg} becomes
		\begin{align*}
			\mathbf{F}_{\txt{r},q}\mathbf{c}_{\txt{r},q}\mathbf{c}_{\txt{t},q}^\txt{T}\mathbf{F}_{\txt{t},q}^\txt{T}\!=\!4{c}_{\txt{r},q}{c}_{\txt{t},q}\mathbf{e}_{n_\txt{r}}\mathbf{e}_{n_\txt{t}}^\txt{T},
		\end{align*}
		where $ q\!=\!n_\txt{r}+(n_\txt{t}-1)N_\txt{r}$. {If} $ 4{c}_{\txt{r},q}{c}_{\txt{t},q}\!=\!W_{n_\txt{r}n_\txt{t}} $ {then, as desired,}
		\begin{align*}
			\sum_{q=1}^{N_\txt{r}N_\txt{t}} \mathbf{F}_{\txt{r},q}\mathbf{c}_{\txt{r},q}\mathbf{c}_{\txt{t},q}^\txt{T}\mathbf{F}_{\txt{t},q}^\txt{T}\!=\!\sum_{n_\txt{r}=1}^{N_\txt{r}}\sum_{n_\txt{t}=1}^{N_\txt{t}} \mathbf{e}_{n_\txt{r}}\mathbf{e}_{n_\txt{t}}^\txt{T} {W}_{n_\txt{r}n_\txt{t}}\!=\!\mathbf{W}.
		\end{align*}
		{Choosing $ {c}_{\txt{r},q}\!=\!{c}_{\txt{t},q}\!=\!\sqrt{{W}_{n_\txt{r},n_\txt{t}}}/2 $ then yields \eqref{eq:gx_opt_q}.}
	\end{proof}
	Theorem~\ref{thm:M2F2J1} implies that the hybrid beamformer with at least two Tx/Rx front ends can achieve the PSF of the fully digital beamformer, regardless of the number of bits used to quantize the phase shifters. This is facilitated by image addition, which trades off an increase in the number of component images $ Q $ for lower quantization precision $ {B} $, and fewer Tx/Rx front ends $ M_\txt{x}$. As a corollary of Theorem~\ref{thm:M2F2J1}, we see that the number of component images of the hybrid beamformer is always upper bounded by $ Q\leq N_\txt{r}N_\txt{t} $, since a trivial solution with $ M_\txt{x}\geq 2 $; $ Q= N_\txt{r}N_\txt{t}$; $ {B}\geq 1 $ is achieved by appending columns with arbitrary phases to \eqref{eq:Bx_opt_q}, and zeros to \eqref{eq:gx_opt_q}.
	
	\subsection{Fully analog beamformer}
	A fully analog beamformer may be constructed directly from a hybrid architecture by either increasing the number of component images $ Q $, or by modifying the beamforming architecture as in Remark~\ref{thm:M1F2} of Section~\ref{sec:prior_work}. In the latter case, the number of phase shifters is still $ M_\txt{x}N_\txt{x} $, although only a single Tx/Rx front end is used. Actually, the number of phase shifters can be reduced to half by doubling $ Q $. More generally, the following lemma shows that {the} total number of phase shifters can be reduced from $ M_\text{t}N_\text{t}+M_\text{r}N_\text{r} $ to $ N_\text{t}+N_\text{r} $ by increasing the number of component images from $ Q $ to $ M_\text{t}M_\text{r}Q $.
	\begin{lemma}[Analog beamforming weights from hybrid beamforming weights]\label{thm:M1F1}
		Any $ \mathbf{W}=\\ \sum_{\tilde{q}=1}^{Q}\mathbf{F}_{\txt{r},\tilde{q}}\mathbf{c}_{\txt{r},\tilde{q}}\mathbf{c}_{\txt{t},\tilde{q}}^\txt{T}\mathbf{F}_{\txt{t},\tilde{q}}^\txt{T} \in\mathbb{C}^{N_\txt{r}\times N_\txt{t}}$, where $\mathbf{F}_{\txt{x},\tilde{q}}\in\mathbb{C}^{N_\txt{x}\times M_\txt{x}} $ and $\mathbf{c}_{\txt{x},\tilde{q}}\in\mathbb{C}^{M_\txt{x}}  $, can be factorized as $ \mathbf{W} = \sum_{{q} = 1}^{M_\txt{r}M_\txt{t}Q}{c}_{\txt{r},{q}}{c}_{\txt{t},{q}}\mathbf{f}_{\txt{r},{q}}\mathbf{f}_{\txt{t},{q}}^\txt{T} $. {For example, a valid choice is} 
		\begin{align}
		\mathbf{f}_{\txt{x},{q}} &= [\mathbf{F}_{\txt{x},{\lceil{q}/(M_\txt{r}M_\txt{t}) \rceil}}]_{:,m_\txt{x}} \label{eq:bx_analog}\\
		{c}_{\txt{x},{q}} &=[\mathbf{c}_{\txt{x},{\lceil{q}/(M_\txt{r}M_\txt{t}) \rceil}}]_{m_\txt{x}},\label{eq:gx_analog}
		\end{align}
		where $ m_\txt{r} = \lceil (1+({q}-1) \bmod M_\txt{r}M_\txt{t})/M_\txt{t} \rceil $ and $ m_\txt{t}= 1+({q}-1) \bmod M_\txt{t}$.
	\end{lemma}
	\begin{proof}\let\qed\relax 
		See Appendix~\ref{proof:thm:M1F1}.
	\end{proof}
	
	\subsubsection{Continuous phase shifters}
	Recall from Theorem~\ref{thm:M2F2} that a hybrid beamformer with continuous phase shifters can achieve any fully digital beamforming vectors using only two Tx/Rx front ends. By Lemma~\ref{thm:M1F1}, the number of front ends may further be halved by quadrupling the number of component images $ Q $, as shown by the following theorem (cf. \cite[Theorem~1]{rajamaki2019analog}).
	\begin{thm}[{Analog beamformer, continuous phase shifters}]\label{thm:M1F1Jinf}
		Let $ M_\txt{x}=1$ and $ {B}\to \infty$. Any $ \mathbf{W}=\sum_{\tilde{q}=1}^{Q}\mathbf{w}_{\txt{r},\tilde{q}}\mathbf{w}_{\txt{t},\tilde{q}}^\txt{T} \in\mathbb{C}^{N_\txt{r}\times N_\txt{t}}$ may be factorized as $ \mathbf{W} = \sum_{{q} = 1}^{4Q}{c}_{\txt{r},{q}}{c}_{\txt{t},{q}}\mathbf{f}_{\txt{r},{q}}\mathbf{f}_{\txt{t},{q}}^\txt{T} $, with $ {c}_{\txt{x},q}\in\mathbb{C} $; and $ \mathbf{f}_{\txt{x},q}\in \mathscr{F}_\txt{x}(\infty)$ following \eqref{eq:B_quantized}. For example, a valid factorization is 
		\begin{align}
		\mathbf{f}_{\txt{x},{q}}&\!=\!{\exp\Bigg(j\Big(\measuredangle \mathbf{w}_{\txt{x},\tilde{q}}\!+\!(-1)^{i_\txt{x}+1}\cos^{-1}\Big(\frac{|\mathbf{w}_{\txt{x},\tilde{q}}|}{\|\mathbf{w}_{\txt{x},\tilde{q}}\|_\infty}\Big)\Big)\Bigg)}\label{eq:bx_analog_inf}\\
		{c}_{\txt{x},{q}}&=\|\mathbf{w}_{\txt{x},\tilde{q}}\|_\infty/2, \label{eq:gx_analog_inf}
		\end{align}
		where $ \tilde{q}= \lceil{q}/4 \rceil$; $i_\txt{r}\!=\!\lceil (1+({q}-1) \bmod 4)/2\rceil $; and $ i_\txt{t}\!=\!1+({q}-1) \bmod 2$. 
	\end{thm}
	\begin{proof}
		By Theorem~\ref{thm:M2F2}, we have
		\begin{align*}
		\mathbf{W}=\sum_{\tilde{q}=1}^{Q}\mathbf{w}_{\txt{r},\tilde{q}}\mathbf{w}_{\txt{t},\tilde{q}}^\txt{T}
		=\sum_{\tilde{q}=1}^{Q}\mathbf{F}_{\txt{r},\tilde{q}}\mathbf{c}_{\txt{r},\tilde{q}}\mathbf{c}_{\txt{t},\tilde{q}}^\txt{T} \mathbf{F}_{\txt{t},\tilde{q}}^\txt{T}, 
		\end{align*}
		where $ \mathbf{F}_{\txt{x},\tilde{q}}\!\in\!\mathscr{F}_\txt{x}(\infty)\!\subset\!\mathbb{C}^{N_\txt{x}\times 2}$ and $ \mathbf{c}_{\txt{x},\tilde{q}}\!\in\!\mathbb{C}^{2}$. Factorization into analog beamforming weights using Lemma~\ref{thm:M1F1} then yields 
		\begin{align*}
		\mathbf{W}&=\!\sum_{\tilde{q}=1}^{Q}\sum_{i=1}^{2}\sum_{l=1}^{2}[\mathbf{c}_{\txt{r},\tilde{q}}]_i[\mathbf{c}_{\txt{t},\tilde{q}}]_l[\mathbf{F}_{\txt{r},\tilde{q}}]_{:,i}[\mathbf{F}_{\txt{t},\tilde{q}}]_{:,l}^\txt{T}\!=\!\sum_{{q} = 1}^{4Q}{c}_{\txt{r},{q}} {c}_{\txt{t},{q}}\mathbf{f}_{\txt{r},{q}}\mathbf{f}_{\txt{t},{q}}^\txt{T}.
		\end{align*}
		Substituting \eqref{eq:Bx_opt_inf_q} and \eqref{eq:gx_opt_inf_q} into this expression, and properly accounting for the summation indices yields \eqref{eq:bx_analog_inf} and \eqref{eq:gx_analog_inf}.
	\end{proof}
	
	\subsubsection{One-bit phase shifters}
	According to Remark~\ref{thm:M1F2} in Section~\ref{sec:prior_work}, we may reduce the number of Tx/Rx front ends in Theorem~\ref{thm:M2F2J1} to one, since the digital weight vector in \eqref{eq:gx_opt_q} is a scaled unit vector. Similarly to Theorem~\ref{thm:M1F1Jinf}, the number of phase shifters may further be reduced to half.
	\begin{thm}[{Analog beamformer, $ 1 $-bit phase shifters}]\label{thm:M1F1J1}
		Let $ M_\txt{x}=1$ and $ {B}=1$. Any $ \mathbf{W} \in\mathbb{C}^{N_\txt{r}\times N_\txt{t}}$ may be factorized as $ \mathbf{W}\!=\!\sum_{q = 1}^{4N_\txt{r}N_\txt{t}}{c}_{\txt{r},q}{c}_{\txt{t},q}\mathbf{f}_{\txt{r},q}\mathbf{f}_{\txt{t},q}^\txt{T} $, with $ {c}_{\txt{x},q}\!\in\!\mathbb{C}$, and $ \mathbf{f}_{\txt{x},q}\!\in\!\mathscr{F}_\txt{x}(1)$ following \eqref{eq:B_quantized}. {For example, a valid factorization is}
		\begin{align}
		{\mathbf{f}_{\txt{x},q}} &{= (\mathbf{e}_{n_\txt{x}}-\mathbf{1}_{N_\txt{x}})(-1)^{i_\txt{x}}+\mathbf{e}_{n_\txt{x}}} \label{eq:bx_analog_q}\\
		{{c}_{\txt{x},q}} &{= \sqrt{W_{n_\txt{r}n_\txt{t}}}/2,} \label{eq:gx_analog_q}
		\end{align}
		where $ i_\txt{r} = \lceil (1+(q-1) \bmod 4)/2 \rceil$; $ i_\txt{t} =  1+(q-1) \bmod 2$; $ n_\txt{r}=1+(\lceil q/4 \rceil-1) \bmod N_\txt{r} $; and $n_\txt{t} = \lceil q/(4N_\txt{r})\rceil $.
	\end{thm}
	\begin{proof}
		By Theorem~\ref{thm:M2F2J1} and Lemma~\ref{thm:M1F1}, we have
		\begin{align*}
			\mathbf{W}=\sum_{\tilde{q}=1}^{N_\txt{r}N_\txt{t}} \mathbf{F}_{\txt{r},\tilde{q}}\mathbf{c}_{\txt{r},\tilde{q}}\mathbf{c}_{\txt{t},\tilde{q}}^\txt{T}\mathbf{F}_{\txt{t},\tilde{q}}^\txt{T}=\sum_{q=1}^{4N_\txt{r}N_\txt{t}}{c}_{\txt{r},q}{c}_{\txt{t},q}\mathbf{f}_{\txt{r},q}\mathbf{f}_{\txt{t},q}^\txt{T},
		\end{align*}
		 where $ \mathbf{F}_{\txt{x},\tilde{q}}\in \mathscr{F}_\txt{x}(1)\!\subset\!\mathbb{C}^{N_\txt{x}\times 2}$; $ \mathbf{c}_{\txt{x},\tilde{q}}\!\in\!\mathbb{C}^{2}$; $ \mathbf{f}_{\txt{x},{q}}\!\in\!\mathscr{F}_\txt{x}(1)\!\subset\!\mathbb{C}^{N_\txt{x}}$; and $ {c}_{\txt{x},{q}}\in\mathbb{C} $ (cf. Theorem~\ref{thm:M2F2J1}). Equations \eqref{eq:bx_analog_q} and \eqref{eq:gx_analog_q} then follow from \eqref{eq:Bx_opt_q}, \eqref{eq:gx_opt_q}{, and the indexing in Lemma~\ref{thm:M1F1}}.
	\end{proof}
	A direct corollary of Theorem~\ref{thm:M1F1J1} is that the number of component images of the analog beamformer is upper bounded by $ Q\leq 4N_\txt{r}N_\txt{t} $, since $ \mathscr{F}_\txt{x}(1)\subseteq  \mathscr{F}_\txt{x}({B}\geq1)$. The {bound is not necessarily tight, as the} gap between the bounds presented in this section {and the solutions found by Algorithm~\ref{alg:greedy}} can be significant, as we will show in the next section. Establishing tighter bounds is therefore an important topic for future work.
	
	\section{Numerical experiments} \label{sec:examples}
	This section presents numerical {results} using the beamforming weight optimization algorithms developed in Section~\ref{sec:numerical} and {the} closed-form beamformer designs derived in Section~\ref{sec:analytical}. {We first introduce the necessary preliminaries and describe the simulation setup. We then evaluate the performance of Algorithm~\ref{alg:altmin_d} and \ref{alg:greedy} for randomly drawn target PSFs, and} study how trade-offs {among} the main parameters $ M,{B}, $ and $ Q $ affect the {realized} PSF. {Lastly}, we simulate a planar array imaging far-field scatterers. We {show} that a sparse hybrid array {with coarsely quantized phase shifters} can achieve comparable {image quality} to a fully digital uniform array.
	
	\subsection{Preliminaries and simulation setup}\label{sec:num_def}
		
	\subsubsection{Linear array model}\label{sec:num_def_lin}
	{The linear array is a useful model for illustrating the impact of different design parameters in a simple and intuitive manner. Consequently, in Sections~\ref{sec:num_val} to \ref{sec:num_to}, we} consider two linear array configurations with co-located transceivers: the \emph{uniform linear array} (ULA), and the \emph{minimum-redundancy array} (MRA) \cite{moffet1968minimumredundancy,hoctor1996arrayredundancy}. The MRA has the largest uniform sum co-array for a given number of elements. {Since} each {sensor} is used for both transmission and reception, we denote $N\!=\!N_\txt{x} $ and $ M\!=\!M_\txt{x}$. {We assume that the elements are identical and omnidirectional with a {unit} inter-element spacing of {half a wavelength} (${d=} \lambda/2 $).} No mutual coupling between the elements is {considered}. {We {particularly} study the $ N\!=\!11 $ {element ULA, and $ N\!=\!7 $ element MRA in more detail (Fig.~\ref{fig:arrays_lin}). The two arrays span the same aperture $5 \lambda$ and are {sum} co-array equivalent.}} The {Tx} and {Rx} steering vectors {of the arrays are} given by {$ \mathbf{a}\!=\!\mathbf{a}_\txt{t}\!=\!\mathbf{a}_\txt{r}$, where}
	\begin{align*}
	\mathbf{a}(\varphi)= \exp(j\pi {\mathbf{d}}\sin\varphi).
	\end{align*}
	{Here,} $ \mathbf{d}\in\mathbb{Z}^N$ denotes the element positions on the $ x $-axis normalized by $ \lambda/2 $. For the ULA $ \mathbf{d}\!=\![-5, -4, \ldots,5]^\txt{T} $, and for the MRA $ \mathbf{d}\!=\![-5, -4,-2,0,2,4,5]^\txt{T} $. A complete list of MRAs with $ N\!\leq\!42 $ elements can be found in \cite{kohonen2014meet}.
	\begin{figure}[t]
		\centering
		\includegraphics[width=.7\linewidth]{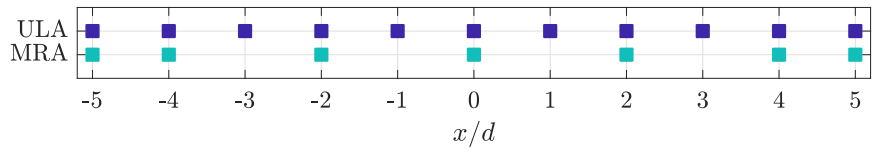}
		\caption{Linear array configurations (co-located transceivers). The MRA with $ N\!=\!7 $ elements is (sum) co-array equivalent to the ULA with $ N\!=\!11 $ elements.}\label{fig:arrays_lin}
	\end{figure}
	
	\subsubsection{Planar array model}\label{sec:num_def_plan}
	{Planar arrays are {often used} in active sensing and imaging applications. Consequently,} {in Section~\ref{sec:examples_plan}, we consider} the square uniform rectangular array (URA) and boundary array (BA) \cite{hoctor1990theunifying} {shown} in Fig.~\ref{fig:planar_array}. Both {arrays have} a side length of $16 $ unit inter-element spacings and an equivalent uniform sum co-array. The unit distance between elements is set to {$d= \lambda/2 $}, and the number of elements is $ N=289 $ in the case of the URA, and $ N=64 $ in case of the BA. All elements are used for both transmission and reception, which means that the fully digital beamforming architecture requires $ N $ ADCs/DACs. The BA in Fig.~\ref{fig:planar_array}~(b) also satisfies the minimum-redundancy property, which implies that it has the fewest elements {of} all arrays that are sum co-array equivalent with the URA in Fig.~\ref{fig:planar_array}~(a) \cite{kohonen2018planaradditive}. {Note that other sparse array configurations with this property also exist \cite{rajamaki2018sparseactive}.}
	\begin{figure}[t]
		\begin{minipage}[b]{.49\linewidth}
			\centering
			\centerline{\includegraphics[width=.7\textwidth]{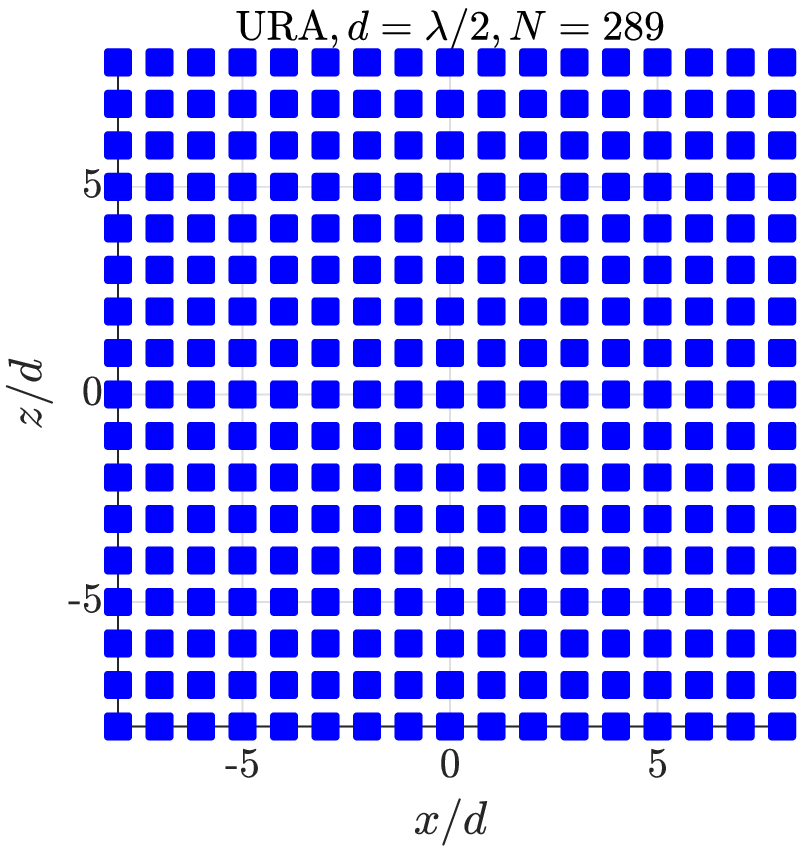}}
			\centerline{(a) Uniform rectangular array}\medskip
		\end{minipage}
		\hfill
		\begin{minipage}[b]{.49\linewidth}
			\centering
			\centerline{\includegraphics[width=.7\textwidth]{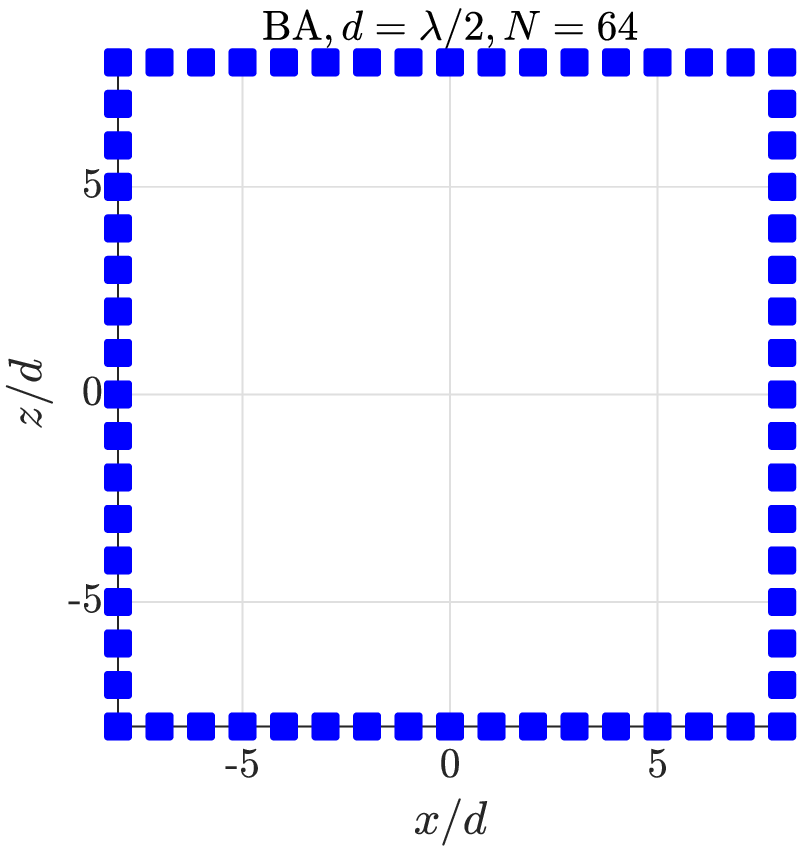}}
			\centerline{(b) Boundary array}\medskip
		\end{minipage}
		\caption{Planar square arrays (co-located transceivers). The two arrays are co-array equivalent, although the BA in (b) uses almost $ 78\% $ fewer physical elements than the URA in (a).}\label{fig:planar_array}
	\end{figure}

	We ignore mutual coupling and assume that the array elements have identical {sinusoidal gain patterns} $ g(\varphi,\theta) =  \cos \varphi \sin\theta$. Consequently, the (transmit and receive) steering vectors assume the form
	\begin{align*}
	\mathbf{a}(\varphi,\theta) = \cos \varphi \sin\theta\exp(j\pi (\mathbf{d}_x\sin\varphi \sin \theta+\mathbf{d}_z\cos\theta )),
	\end{align*}
	where $ \mathbf{d}_x {\in\mathbb{Z}^N}$ and $ \mathbf{d}_z {\in\mathbb{Z}^N}$ are the $ x $ and $ z $ coordinates of the elements {normalized by $ \lambda/2 $, {as illustrated in Fig.~\ref{fig:planar_array}}}.
	
	\subsubsection{Stochastic PSF model} \label{sec:num_def_psf_sto}
	{For performance evaluation purposes, we generate the desired co-array weight vector $ \mathbf{w}_\Sigma $ randomly from a uniform distribution (within the complex unit sphere). Specifically,} the $ i $th entry of the vector {$ \mathbf{w}_\Sigma $ becomes}
	\begin{align*}
	[\mathbf{w}_\Sigma]_i\!=\!\sqrt{r_i}e^{j\phi_i}, \txt{where\,}  r_i\!\sim\!\mathcal{U}(0,1) \txt{ and } \phi_i\!\sim\!\mathcal{U}(0,2\pi).
	\end{align*}
	Using this model, we may {conveniently} sample the parameter space of feasible PSFs uniformly at random {in Section~\ref{sec:num_val}}.
	
	\subsubsection{Deterministic PSF model}\label{sec:num_def_psf_det}
	{From the application point of view,} the stochastic model in {Section~\ref{sec:num_def_psf_sto}} may not generate interesting PSFs {that have} a narrow main lobe and low side lobe levels. We {therefore} also consider {four} deterministic PSFs that are commonly used in beamforming \cite{vantrees2002optimum} and power spectrum estimation \cite{stoica2005spectral}.  {Fig.~\ref{fig:psf_det} shows} the {magnitudes} of the rectangular, triangular, Hann, and Dolph-Chebyshev \cite{dolph1946acurrent} {beamforming weights $\mathbf{w}_\Sigma$ {and} the corresponding PSFs $\boldsymbol{\psi}$.}
	\begin{figure}[t]
		\begin{minipage}[b]{.49\linewidth}
			\centering
			\centerline{\includegraphics[width=1\textwidth]{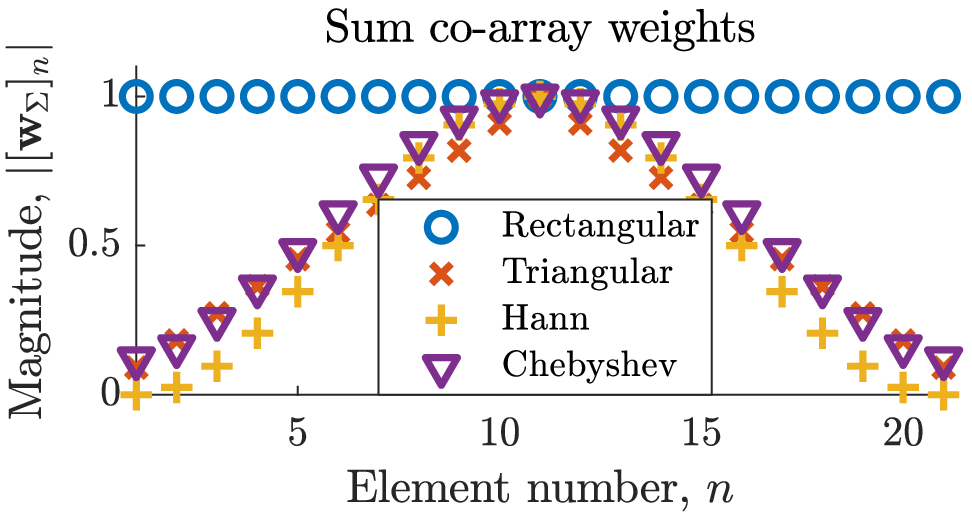}}
		\end{minipage}
		\hfill
		\begin{minipage}[b]{.49\linewidth}
			\centering
			\centerline{\includegraphics[width=1\textwidth]{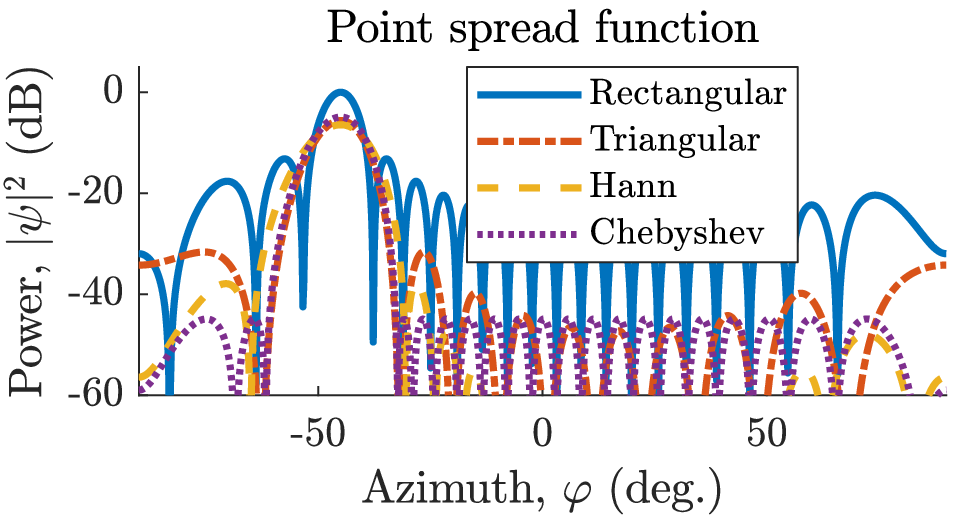}}
		\end{minipage}
		\caption{Typical sum co-array beamforming weights (left) and corresponding PSFs (right) {of the arrays in Fig.~\ref{fig:arrays_lin} (steering direction} $ \varphi = -45^\circ $). Each PSF makes a different trade-off among the main lobe width, array gain, and side-lobe levels.}\label{fig:psf_det}
	\end{figure}	
			
	\subsubsection{Algorithm parameters and performance criterion}\label{sec:num_def_alg}
	{In order to speed up {computations}, we solve \eqref{p:h_alt} for a desired co-array beamforming weight vector $\mathbf{w}_\Sigma\!\in\!\mathbb{C}^{N_\Sigma}$ instead of the sampled PSF $\boldsymbol{\psi}\!\in\!\mathbb{C}^V$ (see Section~\ref{sec:remarks_on_comp_complexity}).} {We set the maximum number of iterations in Algorithm~\ref{alg:altmin_d} to $ k_{\max}\!=\!100 $, and to $ k_{\max}\!=\!10 $ in Algorithm~\ref{alg:greedy}. We use an approximation error tolerance of $ \varepsilon_{\max}\!=\!10^{-16} \|\mathbf{w}_\Sigma\|_2^2$, {except for Section VII-E, where we use} $ \varepsilon_{\max}\!=\!10^{-6} \|\mathbf{w}_\Sigma\|_2^2$.} {Our performance criterion of choice is the \emph{relative approximation error}}
	\begin{align*}
	\epsilon = \|\mathbf{w}_{\Sigma}-\boldsymbol{\Upsilon}\,\txt{vec}(\mathbf{W})\|_2/\|\mathbf{w}_{\Sigma}\|_2.
	\end{align*}
	For an ensemble of realizations of $\epsilon$, we {evaluate the sample} mean, or {alternatively} the median and $ 90\% $ confidence interval ($ 5\% $ and $ 95\% $ percentiles) {of the sample}.
	
	\subsubsection{Computation of pseudo-inverse}
	{For numerical stability,} we compute the (approximate) pseudo-inverse of a matrix $\mathbf{X}$ using diagonal loading ({ridge regression}) as
	\begin{align}
	(\mathbf{X})^\dagger_\alpha  =  (\mathbf{X}^\txt{H}\mathbf{X}+\alpha\mathbf{I})^{-1}\mathbf{X}^\txt{H},\label{eq:tikh}
	\end{align}
	where $ (\mathbf{X})^\dagger_\alpha \approx \mathbf{X}^\dagger $ holds for small values of the diagonal loading parameter $ \alpha\!>\!0 $. {Heuristics, such as regularization or truncated SVD, are often employed when $\mathbf{X}^\txt{H}\mathbf{X}$ is ill-conditioned.} We choose the value of $ \alpha $ by trial-and-error, {since determining a rigorous selection rule is out of scope of this paper.} We {generally} set $ \alpha = 10^{-9} $, with the exception of Section~\ref{sec:examples_plan}, where we use $ \alpha= 10^{-4} $.
		
	\subsection{Validation of beamforming algorithms}\label{sec:num_val}
	{In the following, we evaluate {the two main algorithms}, Algorithm~\ref{alg:altmin_d} and \ref{alg:greedy}, for $ 100 $ random i.i.d. realizations of the desired co-array weight vector $\mathbf{w}_\Sigma $ {following {the stochastic} model in Section~\ref{sec:num_def_psf_sto}}.}
	
	\subsubsection{Algorithm~\ref{alg:altmin_d} (alternating minimization)}
	Fig.~\ref{fig:alt_min} shows the {{mean} relative approximation error} of the fully digital {beamforming weights} found by Algorithm~\ref{alg:altmin_d} as a function of the number of array elements $ N $ and component images $ Q $. The lower bound on $ Q $ in \eqref{eq:Q_d} is {also} shown in red (dashed line). For the ULA, this bound is constant (equal to one), and for the MRA it has a weak linear dependence on $ N $, {since the MRA has larger sum co-array than the ULA for given $ N $} (see Section~\ref{sec:bounds_d}). {We observe in both cases that a \emph{phase transition}, where the error drops rapidly, coincides with the lower bound in \eqref{eq:Q_d}. This empirically validates Algorithm~\ref{alg:altmin_d}.}
	\begin{figure}[t]
	\begin{minipage}[b]{.49\linewidth}
		\centering
		\centerline{\includegraphics[width=1\textwidth]{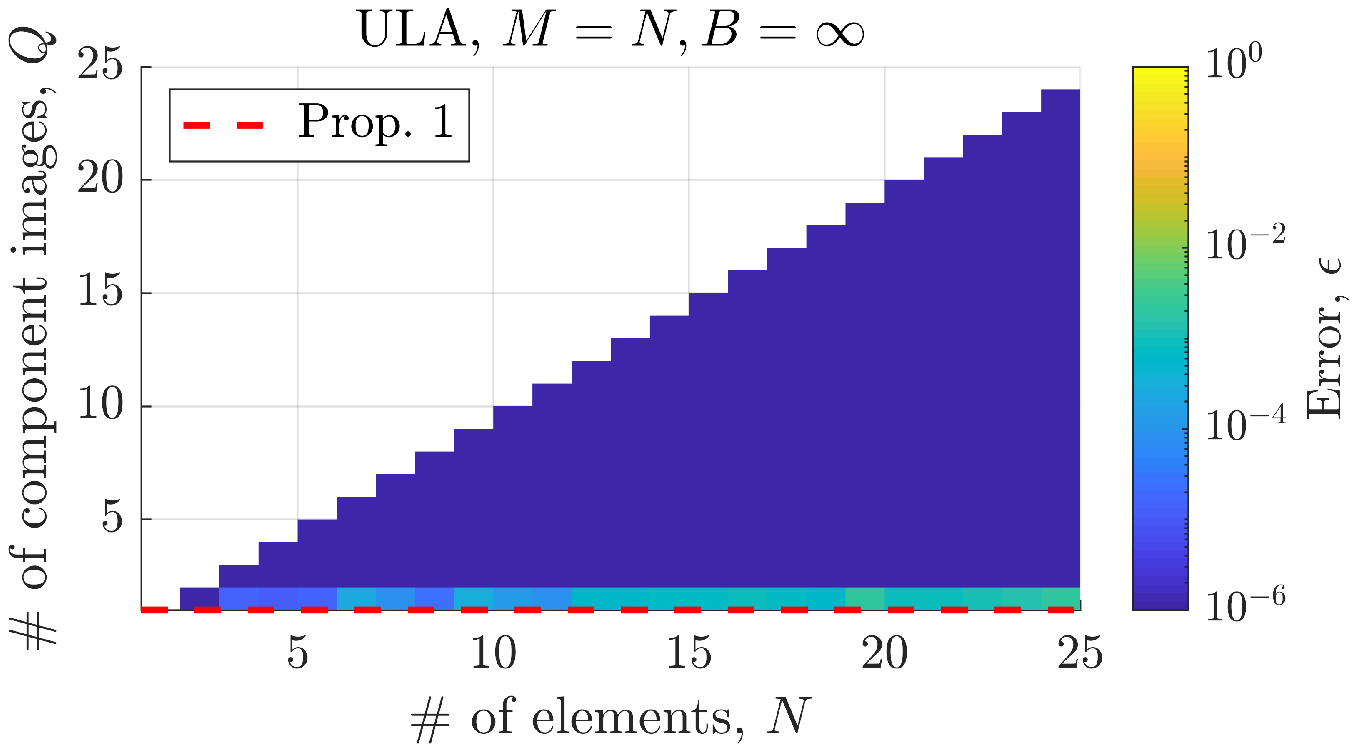}}
	\end{minipage}
	\hfill
	\begin{minipage}[b]{.49\linewidth}
		\centering
		\centerline{\includegraphics[width=1\textwidth]{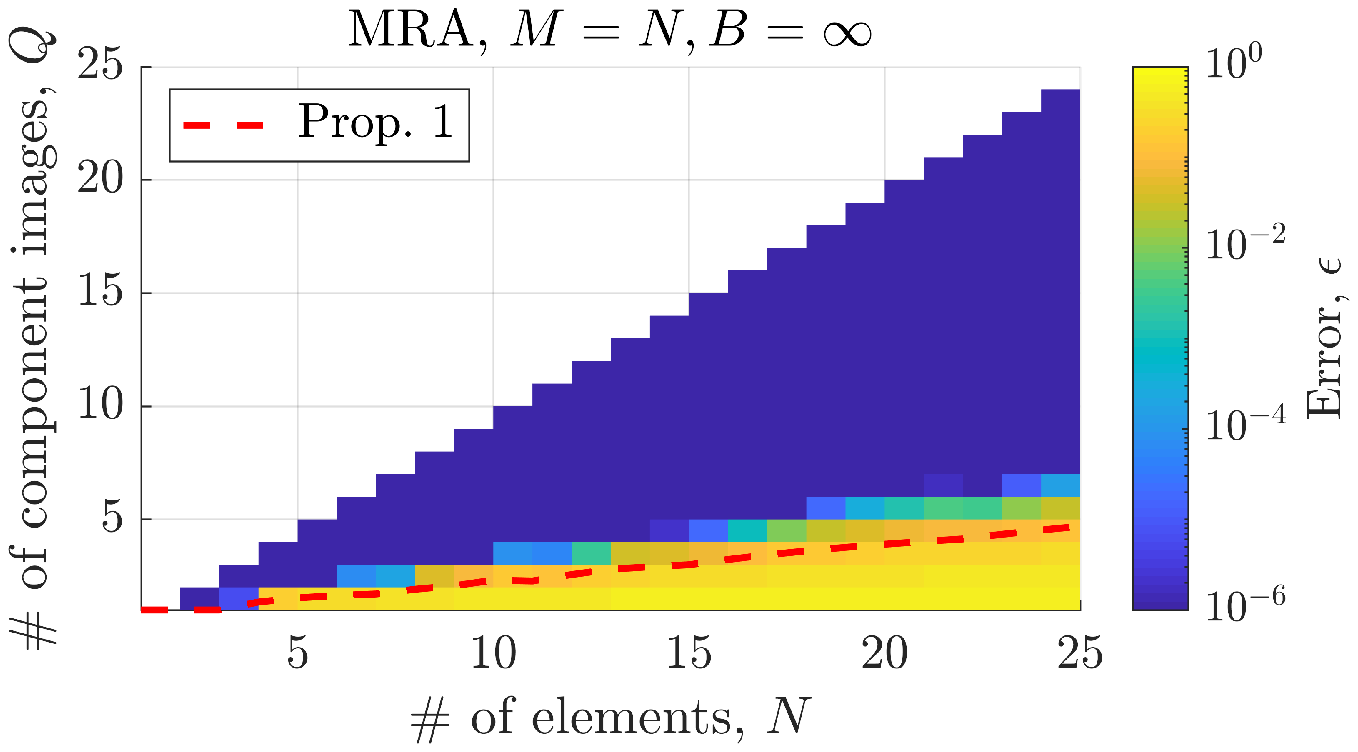}}
	\end{minipage}
	\caption{Mean error of fully digital {ULA (left) and MRA (right)} beamformers (Algorithm~\ref{alg:altmin_d}). The {error} shows a phase transition {approximately following} the lower bound on $ Q $ in \eqref{eq:Q_d}. This empirically validates the efficiency of Algorithm~\ref{alg:altmin_d}.}\label{fig:alt_min}
	\end{figure}

	{Note that by Theorem~\ref{thm:M2F2}, Fig.~\ref{fig:alt_min} also applies to the hybrid beamformer with continuous phase shifters ($ B\to \infty $) and at least two Tx-Rx front-ends {($ M=2 $)}. By Theorem~\ref{thm:M1F1Jinf}, a fully analog beamformer with continous phase shifters ($ M=1 $ and $ B\to\infty $) would achieve the same {error} {level} as the fully digital beamformers in Fig.~\ref{fig:alt_min} using at most four times as many component images.}
	
	\subsubsection{Algorithm~\ref{alg:greedy} (greedy method)}
	{Fig.~\ref{fig:greedy_h} shows the {{mean} error} of the hybrid {beamforming weights} found by Algorithm~\ref{alg:greedy}. The number of Tx/Rx front-ends is $ M=2 $, and the number of phase shift bits is {$ B=5 $ (top row) and $ B=1 $ (bottom row)}. The quantization of the phase shifts degrades the quality of the solution compared to the fully digital beamformer {shown} in Fig.~\ref{fig:alt_min}. However, even in the one bit case, the phase transition boundary of the {error} {is far below} the upper bound $ Q\leq N^2 $ suggested by Theorem~\ref{thm:M2F2J1}. {In fact, the phase transition obeys {a tighter} bound $ Q<N_\Sigma$ (dotted line), where {$N_\Sigma \leq N^2$} is the number of co-array elements.} {These findings suggest the possibility of both algorithmic improvements and tighter bounds {in future work.}}}
	\begin{figure}[t]
	\begin{minipage}[b]{1\linewidth}
	\centering
	{\includegraphics[width=.49\textwidth]{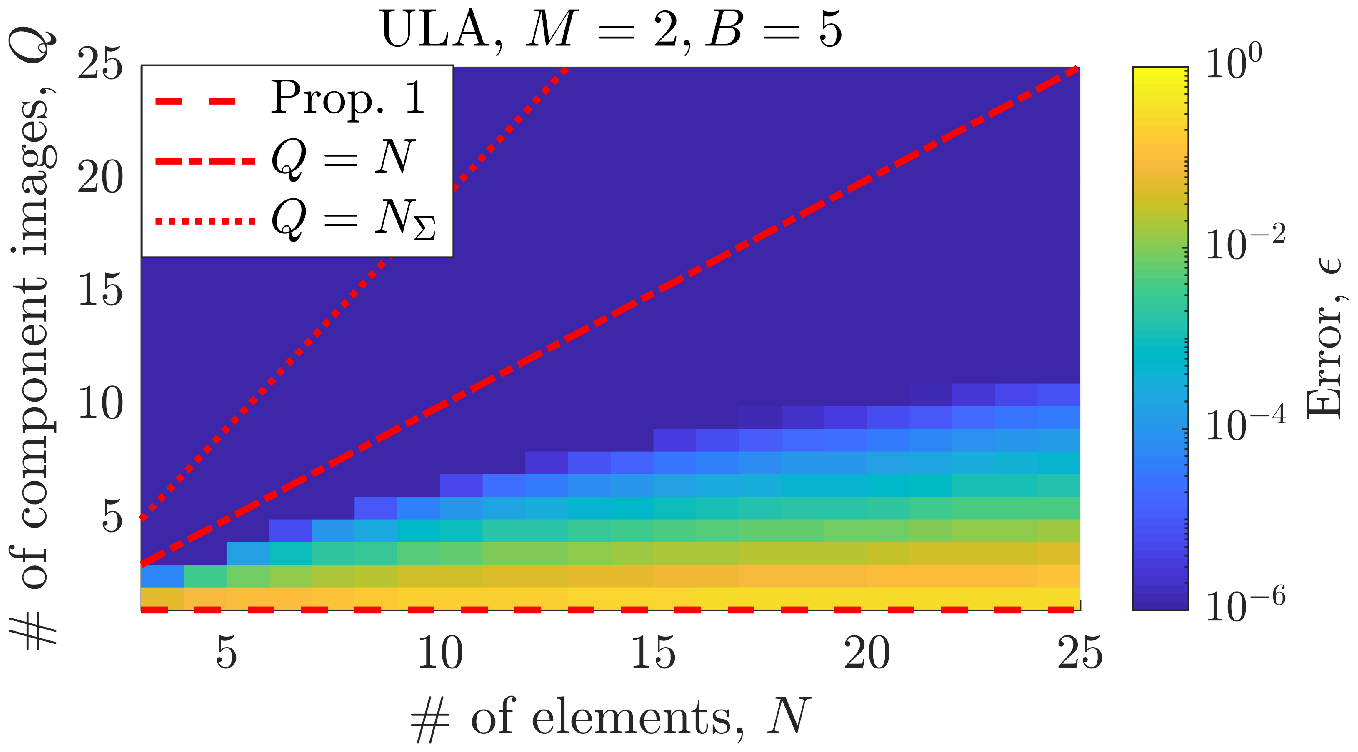}}
	{\includegraphics[width=.49\textwidth]{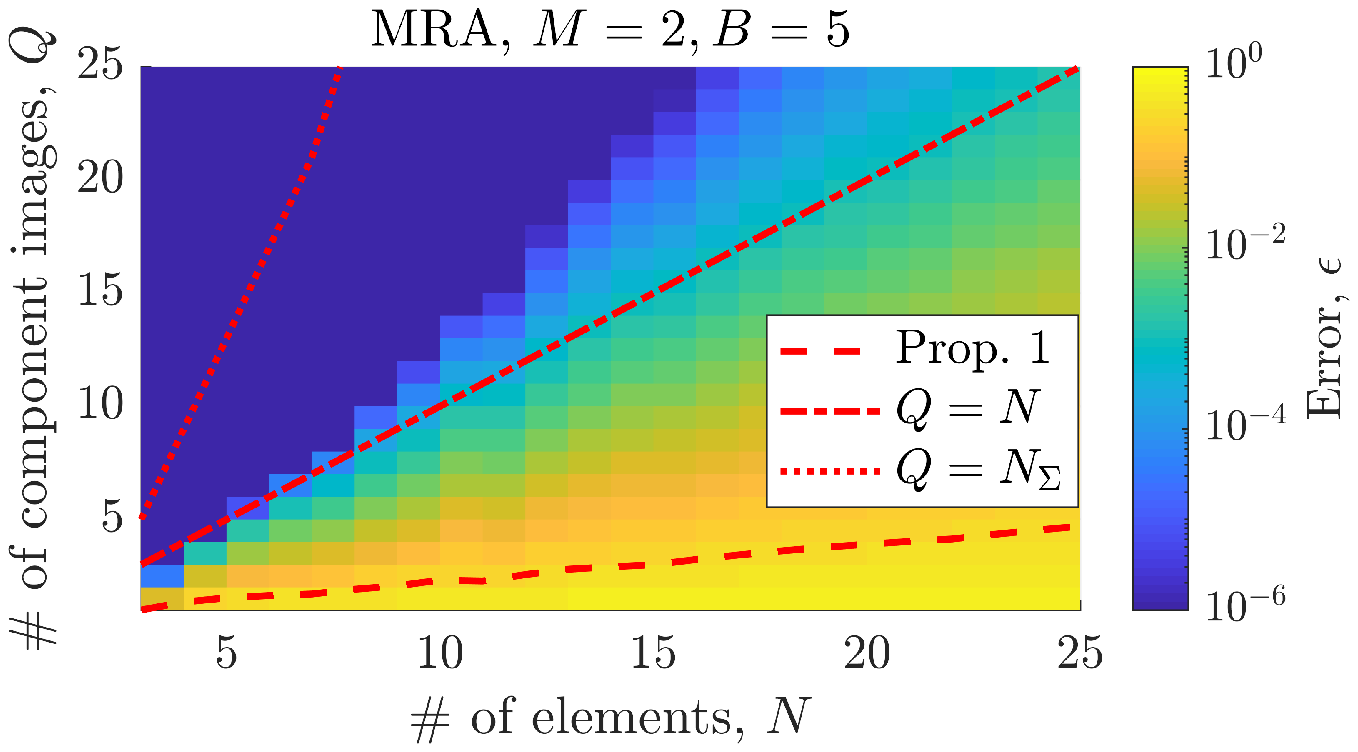}}
	{\includegraphics[width=.49\textwidth]{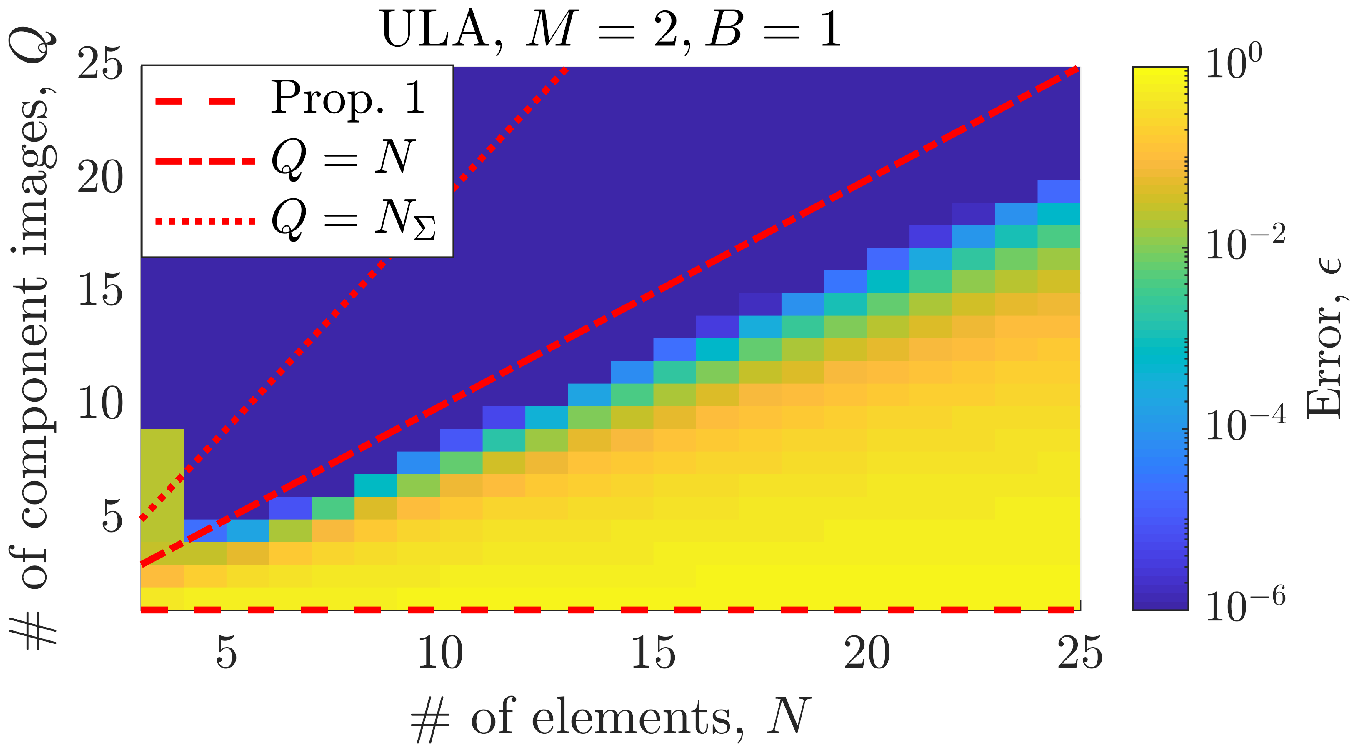}}
	{\includegraphics[width=.49\textwidth]{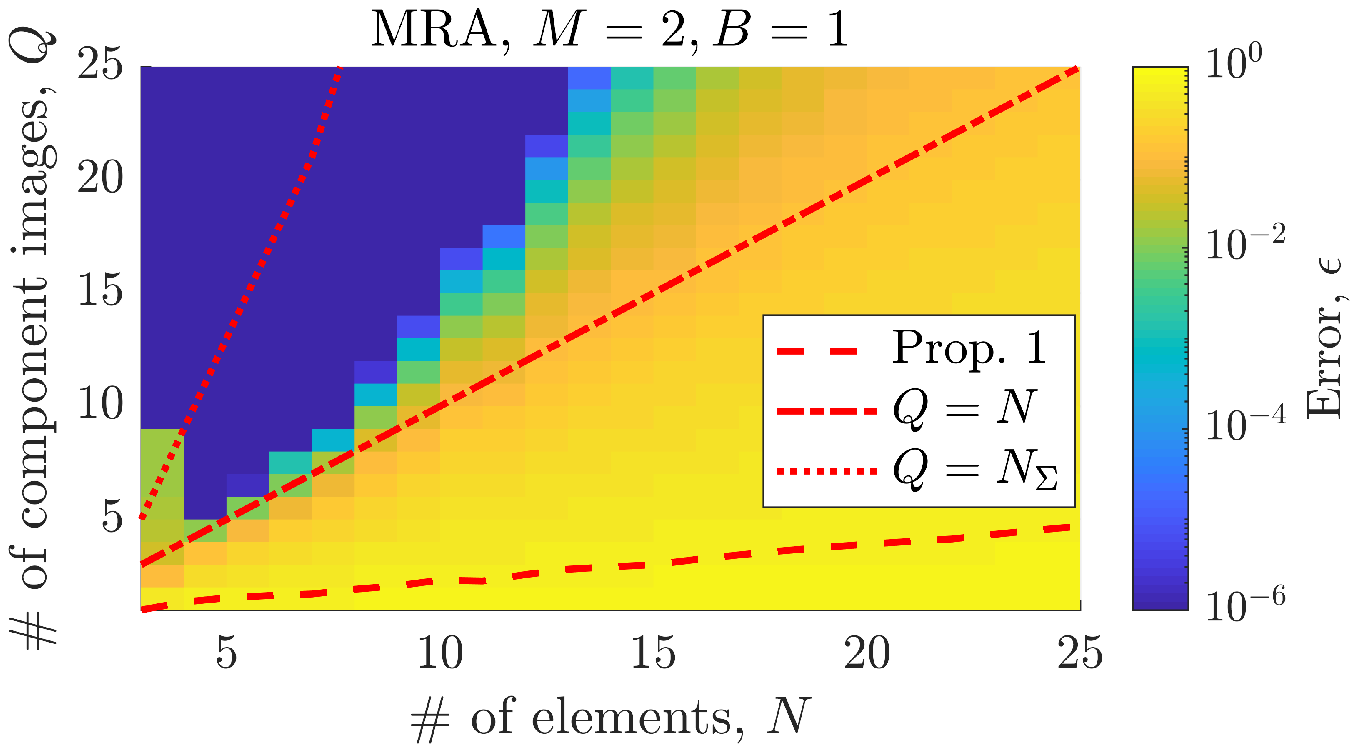}}
\end{minipage}
		\caption{Mean error of hybrid ULA {(left) and MRA (right)} beamformers with $ M\!=\!2 $ {Tx/Rx} front-ends, {and $ B\!=\!5$  (top) and $ B\!=\!1 $ (bottom) phase shift bits} (Algorithm~\ref{alg:greedy}). The phase transition occurs for {$Q\!<\!N_\Sigma $} even in the one bit case.}\label{fig:greedy_h}
	\end{figure}

	Fig.~\ref{fig:greedy_B_2} shows the {median} error {and 90\% confidence intervals} as a function of $ B $ and $ Q $ for the $ N=11 $ element ULA {and $ N=7$ element MRA. For the ULA (left),} the error decreases rapidly as $ B $ increases up to approximately $ B=8 $. After this, {we see diminishing returns in} increasing $ B $. For the MRA (right), increasing $ B $ beyond $ B=6 $ leads to little or no improvement in the error. This {point of diminishing returns} is lower than {for} the ULA, since the MRA has fewer elements.
	
	{A better solution may {sometimes} be obtained by quantizing the {phase shifts of the} beamformer provided by Theorem~\ref{thm:M2F2}. Unlike Algorithm~\ref{alg:greedy}, this solution converges to that of Algorithm~\ref{alg:altmin_d} (the fully digital beamformer) when $ B\to \infty $, provided $ M\geq 2 $. However, the solution does not improve by increasing $ Q $ or $ M $. As shown in Fig.~\ref{fig:greedy_B_2}, Algorithm~\ref{alg:greedy} (colorful non-solid lines) achieves a lower {median} error than directly quantizing Theorem~\ref{thm:M2F2} (black solid line\footnote{We use the fully digital solution found by Algorithm~\ref{alg:altmin_d} ($ Q=2 $), and recompute the digital weight matrices $ \mathbf{C}_\txt{x} $ using alternating minimization as on line~\ref{line:Cx_altmin} of Algorithm~\ref{alg:greedy} {with $ k_{\max}=100 $ iterations}.}), for example, when {$ B \leq 11 $} and $ Q=2 $ {in case of the ULA (left), or $ B \leq 6 $ and $ Q=3 $ in case of the MRA (right)}. Even in the worst case, based on the {$ 90\% $ confidence intervals} (shaded area), Algorithm~\ref{alg:greedy} {mostly} yields {a} better solution when {$ B \leq 8 $} {(ULA) or $ B \leq 4 $ (MRA)}. This cross-over point shifts to higher values of $ B $ as the number of component images $ Q $ increases.}
	\begin{figure}[t]
	\begin{minipage}[b]{1\linewidth}
		\centering
		{\includegraphics[width=.49\linewidth]{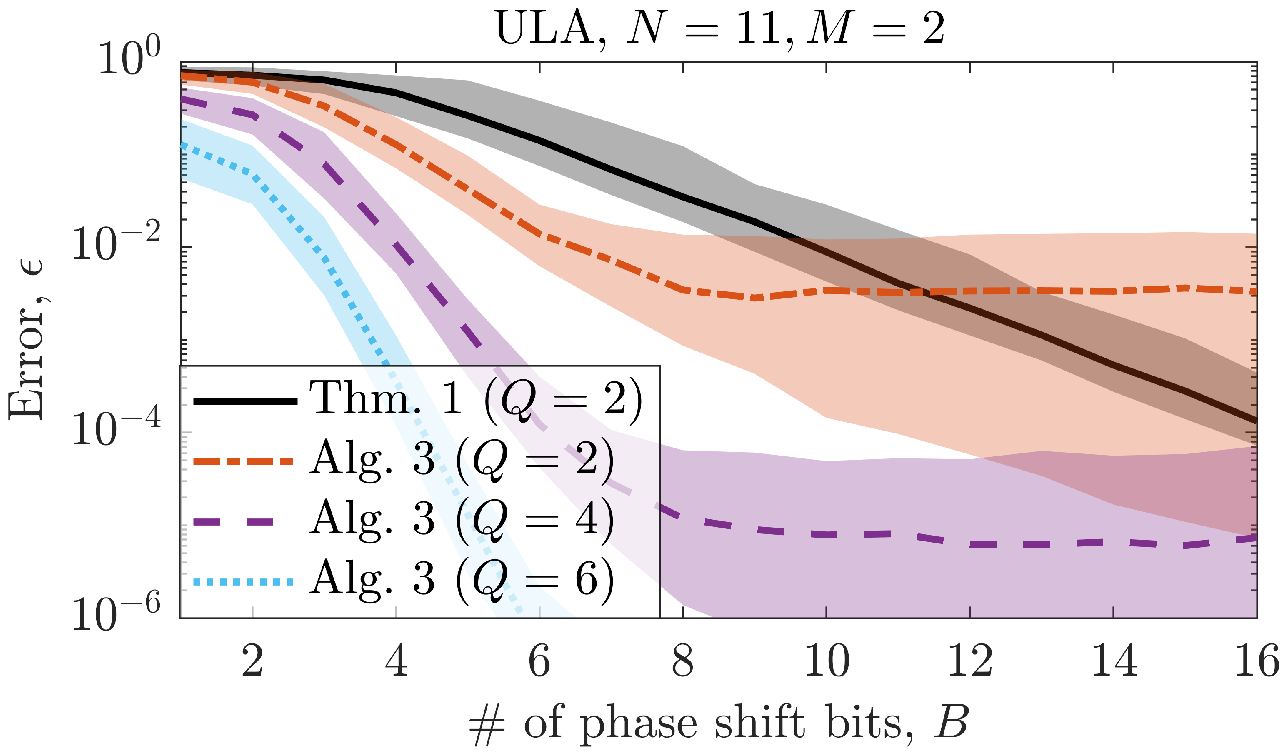}}
		{\includegraphics[width=.49\linewidth]{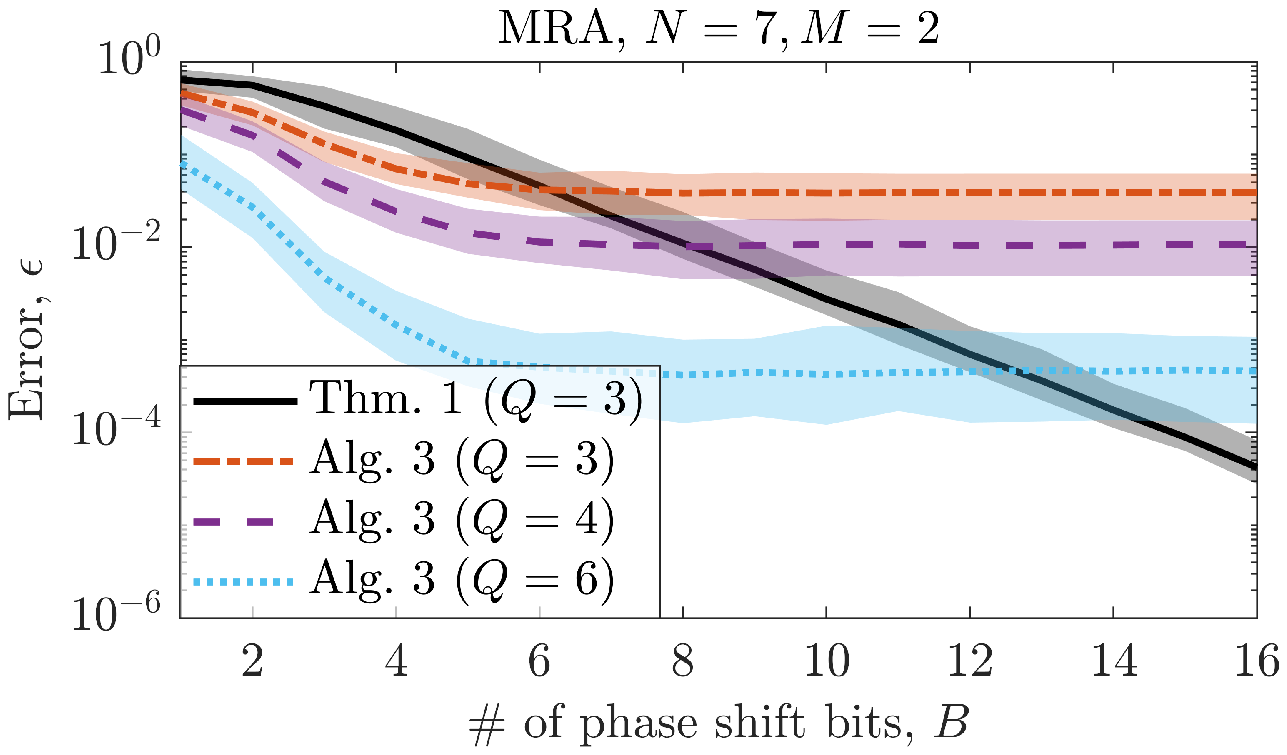}}
	\end{minipage}
	\caption{{Median error {and 90\% confidence intervals} of hybrid ULA (left) and MRA (right) beamformers {(two Tx/Rx front-ends)}. Algorithm~\ref{alg:greedy} {experiences diminishing returns in $ B $, but} achieves a lower median error than directly quantizing Theorem~\ref{thm:M2F2}, when $ B \leq 11$ (ULA) or $ B \leq 6$ (MRA).}}\label{fig:greedy_B_2}
\end{figure}
		
	\subsection{Point spread function of linear arrays}\label{sec:num_psf}
	We {now} {qualitatively study} the point spread function {of the arrays in Fig.~\ref{fig:arrays_lin}} {as a function of $ B $ and $ Q $. For simplicity, we limit ourselves to the Dolph-Chebyshev beampattern in Fig.~\ref{fig:psf_det}.}
	
	Fig.~\ref{fig:bp_lin_h} shows the realized PSF of the hybrid ULA (left column) and MRA (right column) with $ M=2 $ Tx/Rx front ends for {$ Q=1$ (top row), $Q=2 $ (bottom row),} and $ B\in\{1,5,\infty\} $. {When $ {B}\to \infty $, the ULA achieves the desired PSF using $ Q=1 $ component image by application of Theorem~\ref{thm:M2F2} to the fully digital solution found by Algorithm~\ref{alg:altmin_d}.} {The sparser MRA requires $ Q=2 $ for the same result.} When $ {B} $ is finite, {Algorithm~\ref{alg:greedy} needs to be employed. The elevated sidelobes of the PSFs are} reduced by increasing either $ {B} $ or $ Q $.
	\begin{figure}[t]
	\begin{minipage}[b]{1\linewidth}
	\centering
	{\includegraphics[width=.49\linewidth]{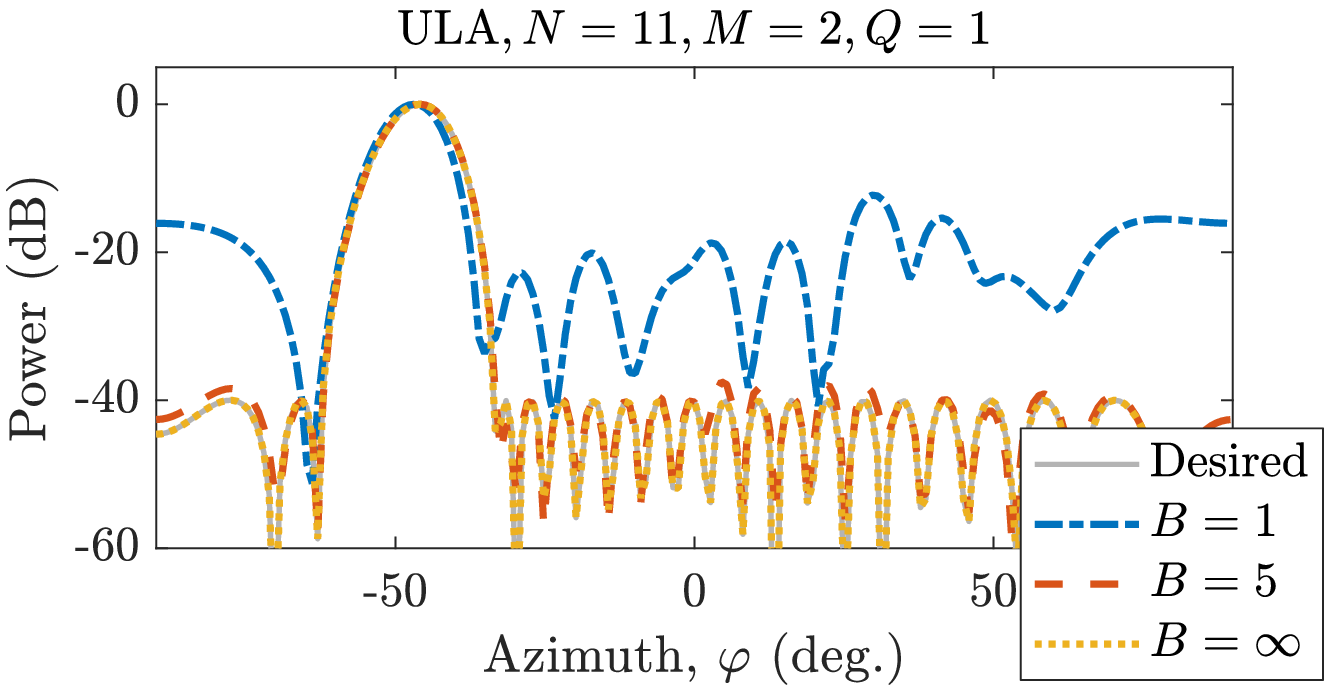}}
	{\includegraphics[width=.49\linewidth]{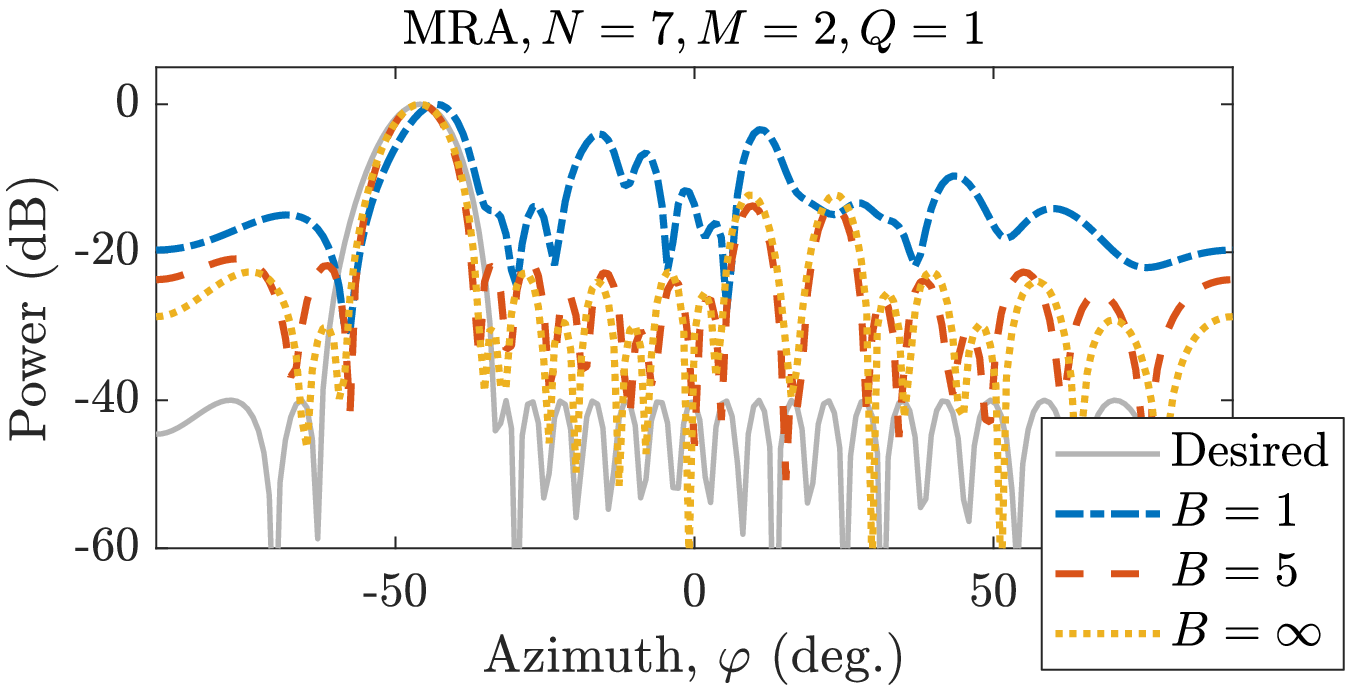}}
	{\includegraphics[width=.49\linewidth]{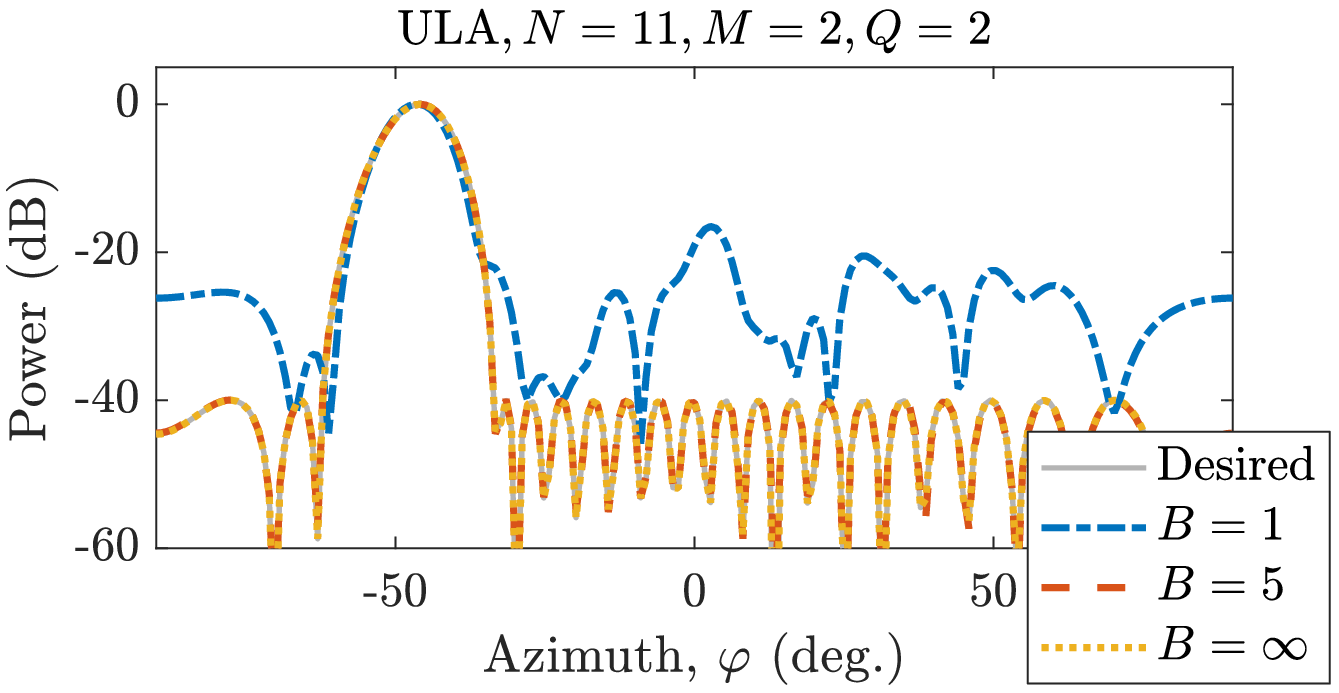}}
	{\includegraphics[width=.49\linewidth]{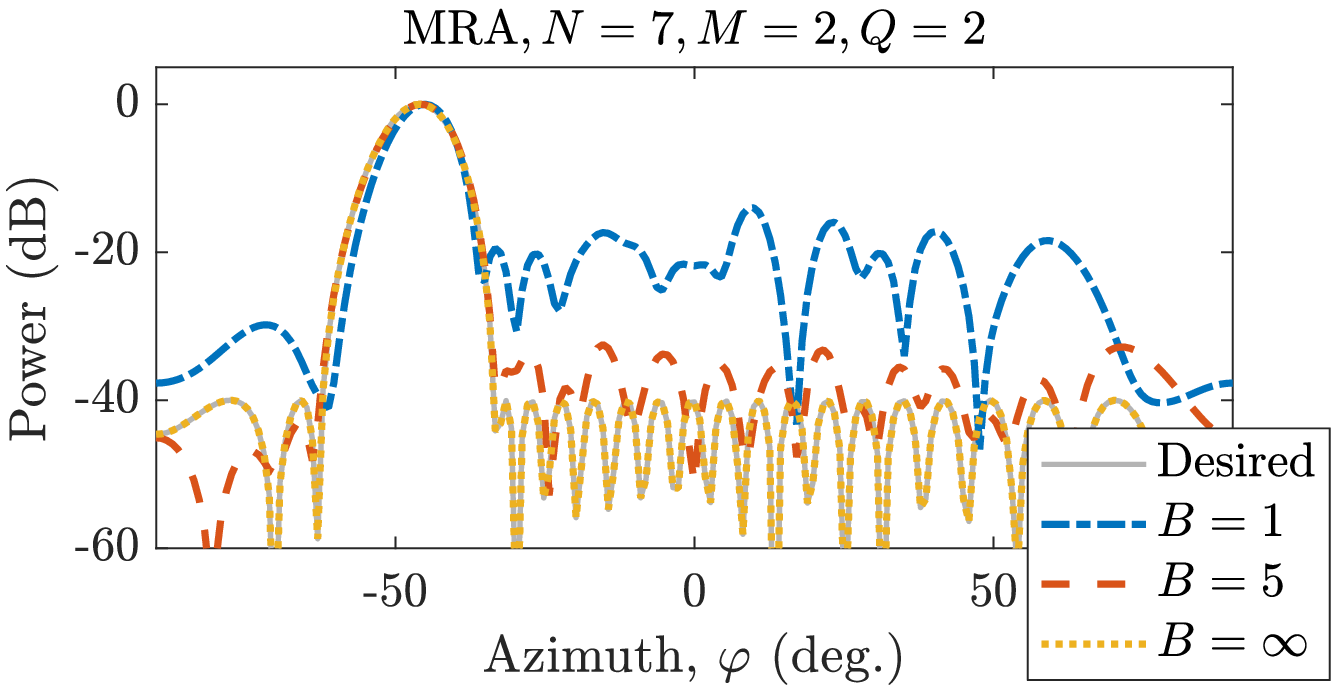}}
	\end{minipage}
	\caption{PSF of hybrid {ULA} (left) and MRA (right) beamformers (two Tx/Rx front ends). The {PSF} improves when increasing the number of component images $ Q $, or phase shift bits $ B$.}\label{fig:bp_lin_h}
	\end{figure}			

{Fig.~\ref{fig:bp_lin_a} shows the realized PSF of the fully analog ($ M=1 $) ULA (left column) and MRA (right column) for $ Q=4$ (top row), $Q=8 $ (bottom row), and $ B\in\{1,5,\infty\} $. Increasing $ Q $ decreases the mismatch between the desired and realized PSFs also in the analog case, although the rate of improvement is slower, and more component image are required compared to the hybrid case ({cf.} Fig.~\ref{fig:bp_lin_h}). When $ B\to \infty $, the ULA and MRA achieve the desired PSF using $ Q=4 $, respectively, $ Q=8 $ component images by application of Theorem~\ref{thm:M1F1Jinf}.}
	\begin{figure}[t]
	\begin{minipage}[b]{1\linewidth}
	\centering
	{\includegraphics[width=.49\linewidth]{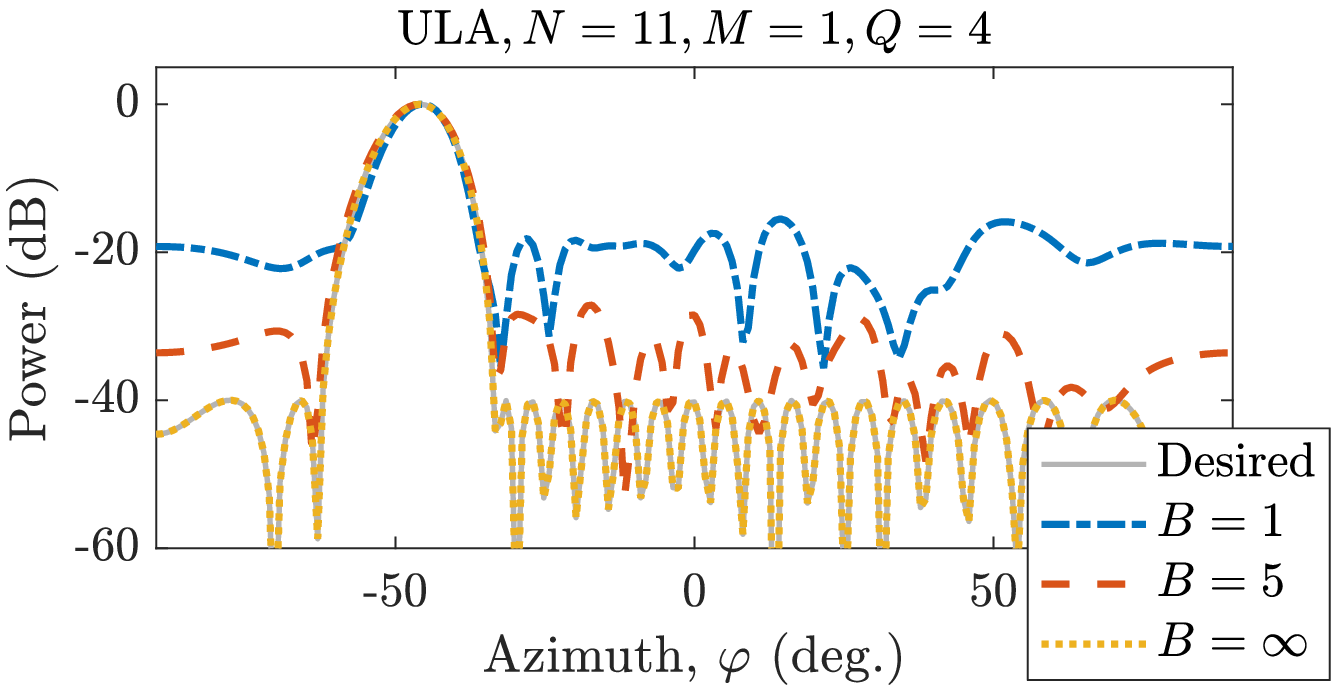}}
	{\includegraphics[width=.49\linewidth]{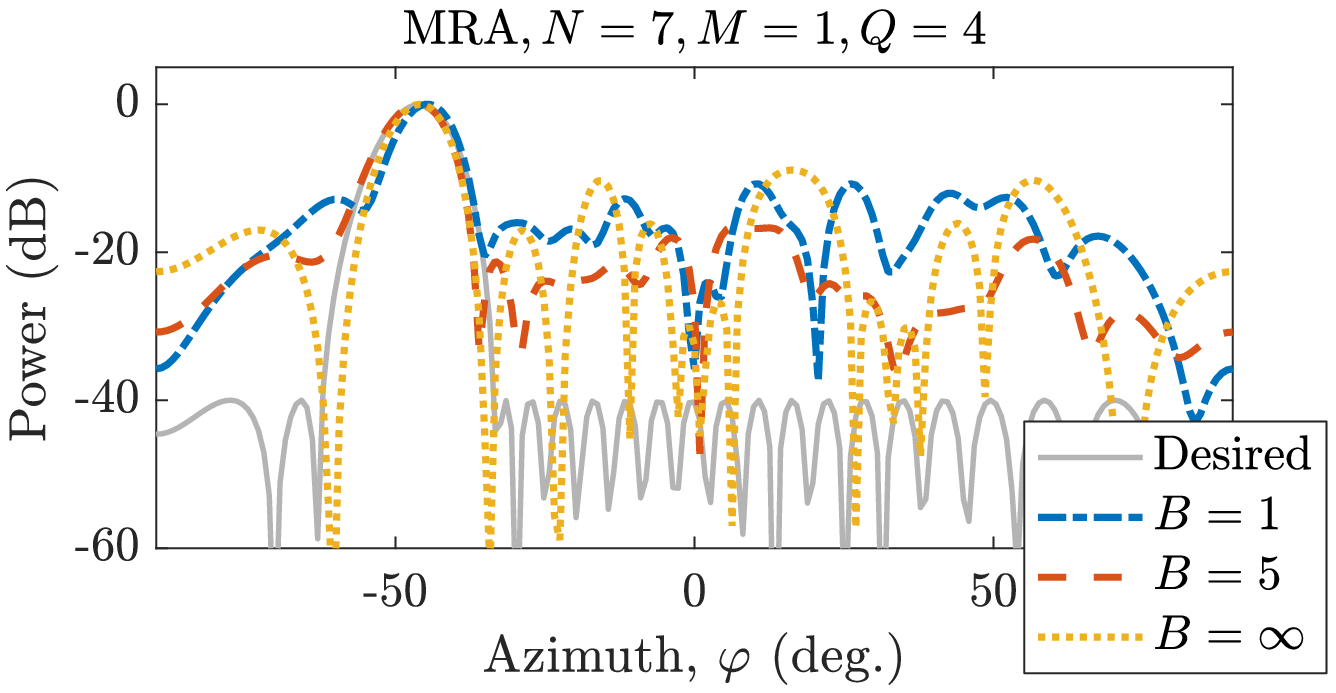}}
	{\includegraphics[width=.49\linewidth]{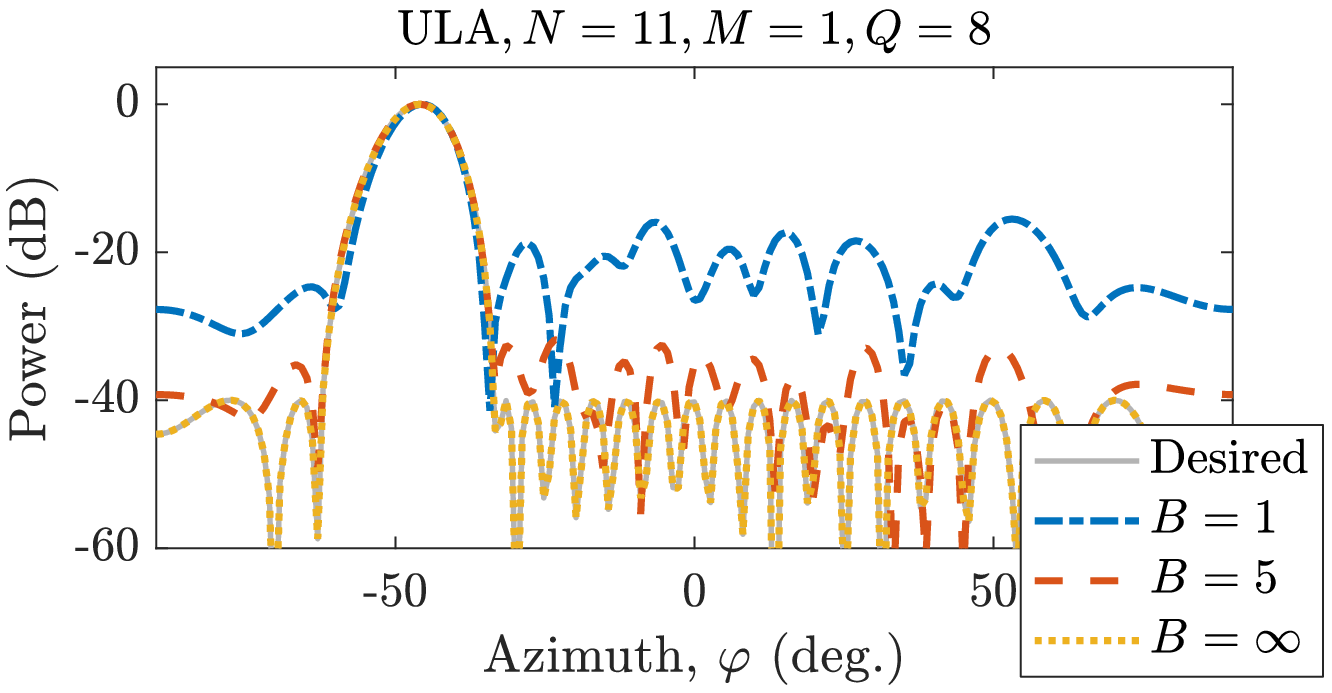}}
	{\includegraphics[width=.49\linewidth]{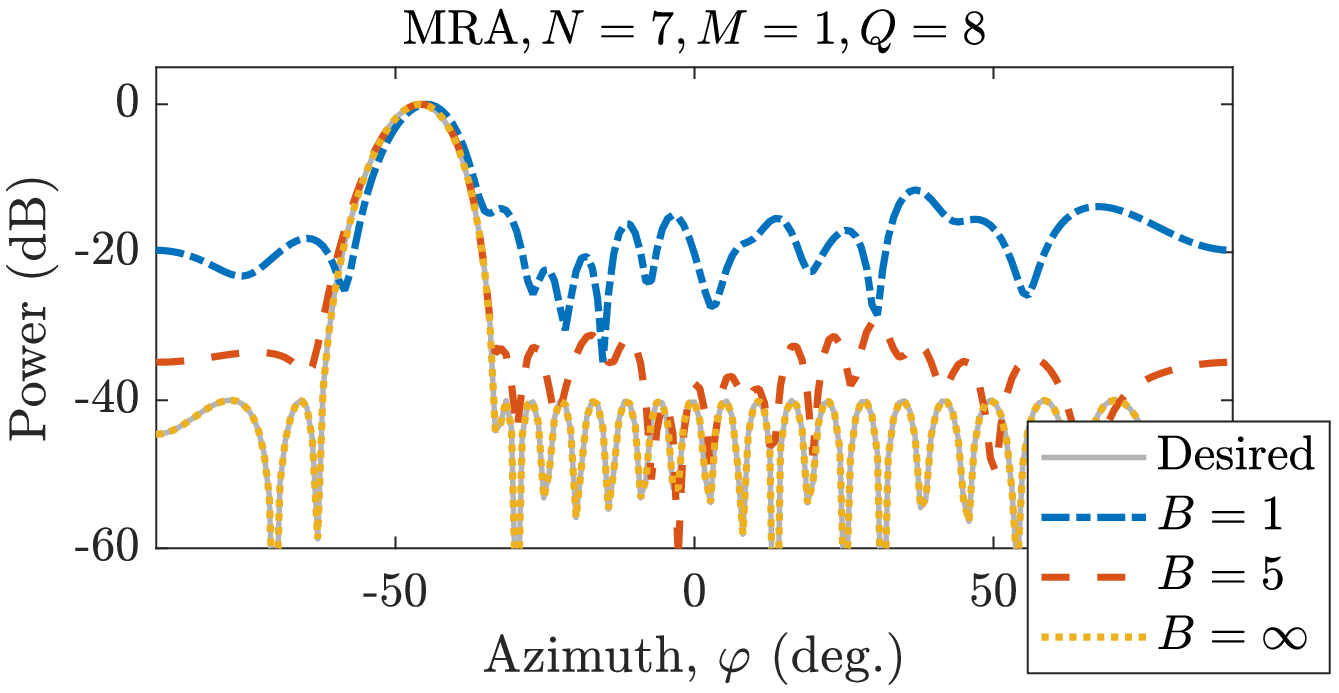}}
\end{minipage}
	\caption{PSF of {fully} analog {ULA} (left) and MRA (right) beamformers (one Tx/Rx front end). A high-fidelity PSF requires more component images than in the hybrid case.} \label{fig:bp_lin_a}
	\end{figure}

	\subsection{Trade-offs among main parameters $ M $, $ Q $ and $ B $}\label{sec:num_to}
	Next, we study the {trade-offs among} the three main {design} parameters $ M$, $Q$, and $ B $. {We fix the number of {quantization} bits $ B $} and evaluate the {relative} approximation error against the number of component images $ Q $ for different numbers of Tx/Rx front ends $ M $. We consider the {linear arrays in Fig.~\ref{fig:arrays_lin}} and the deterministic PSFs in Section~\ref{sec:num_def_psf_det}. Each PSF is steered in $ 100 $ different {scan} directions uniformly sampling the interval $ {[-\pi/2,\pi/2] }$.

	{Fig.~\ref{fig:sweep_ula} shows} {the median error {and 90\% confidence intervals} of the ULA (left column) and MRA (right column) as function of $ Q $ for $ B\!=\!1 $ (top row), $ B\!=\!5 $ (middle row), and $ B\!\to\!\infty $ (bottom row).} The number of Tx/Rx front ends $ M $ has a significant impact on the approximation error of the solution found by Algorithm~\ref{alg:greedy}. When $ M=1 $ (fully analog case), the {error} decreases at a slower rate as a function of $ {B} $ and $ Q $. When {$ M\geq 2 $} and $ B\to\infty $, the {error} {drops abruptly,} since Theorem~\ref{thm:M2F2} can be applied to achieve the {PSF of the} fully digital beamformer using {the same number of} component images. {By Theorem~\ref{thm:M1F1Jinf},} the fully analog beamformer requires four times as many component images as the fully digital beamformer. Consequently, the discontinuities for $ M\!=\!1, B\!\to\!\infty $ occur {when $Q\!=\!4,8,12,16,$ etc.}
	\begin{figure}[t]
		\begin{minipage}[b]{1\linewidth}
			\centering
			{\includegraphics[width=.49\linewidth]{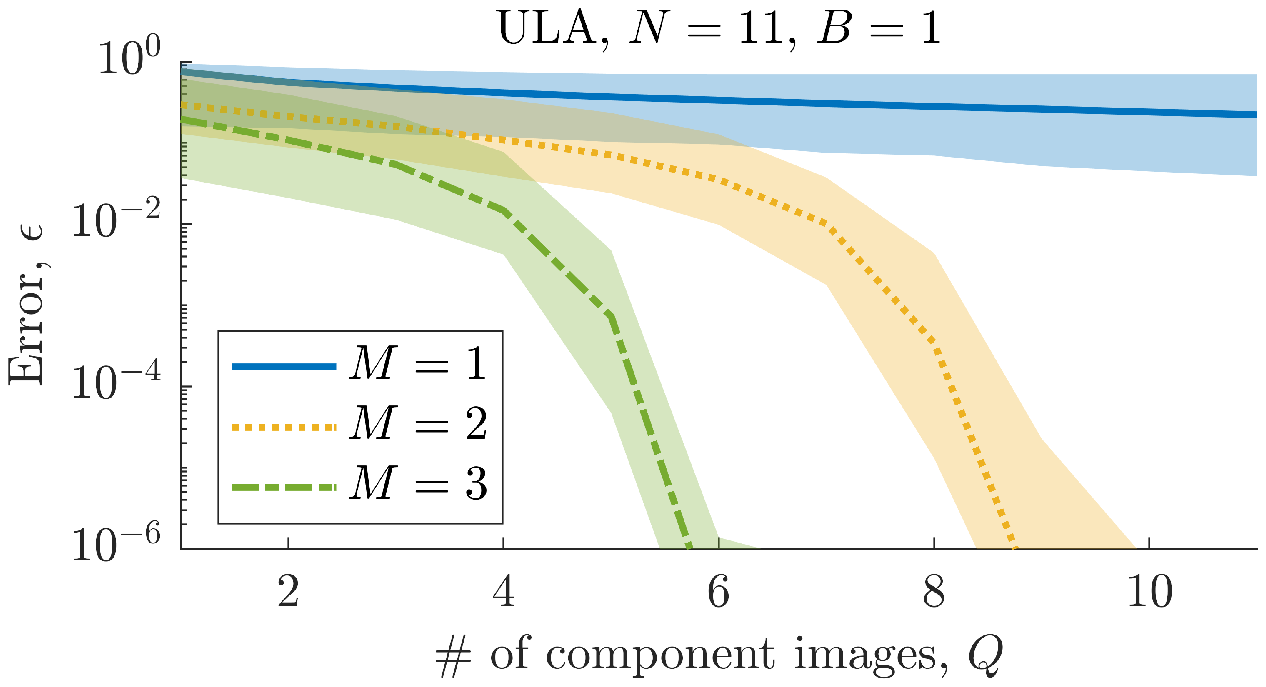}}
			{\includegraphics[width=.49\linewidth]{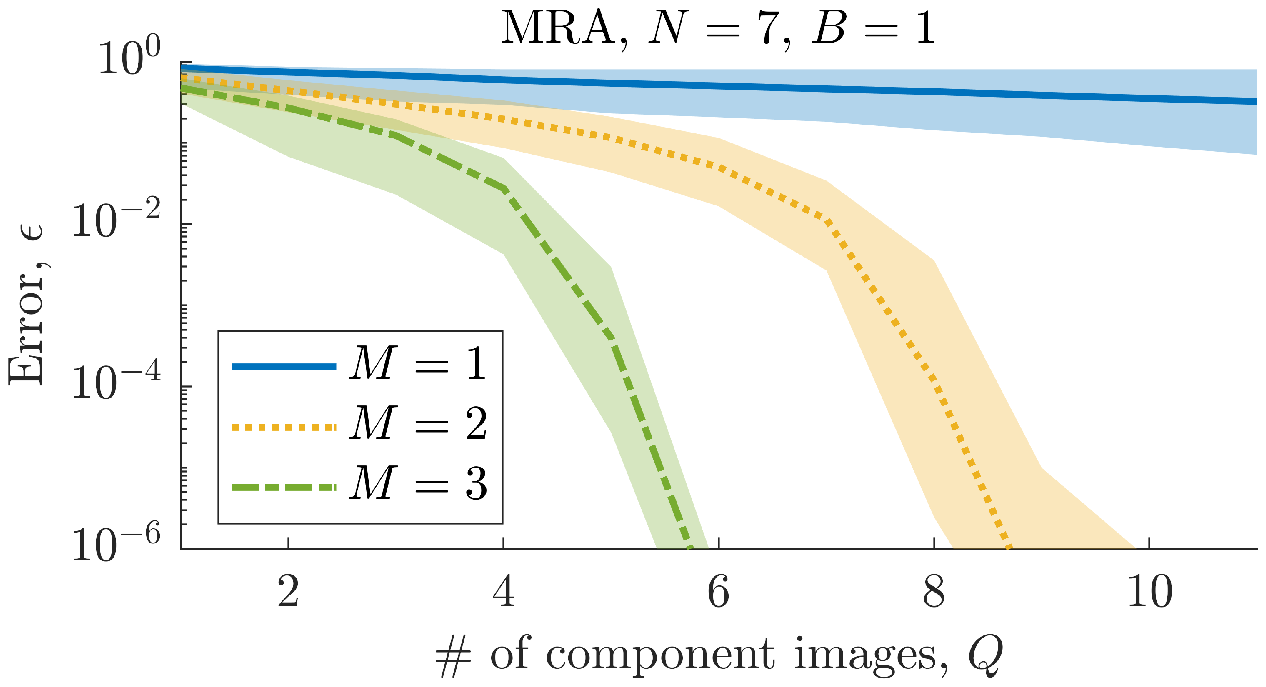}}
			{\includegraphics[width=.49\linewidth]{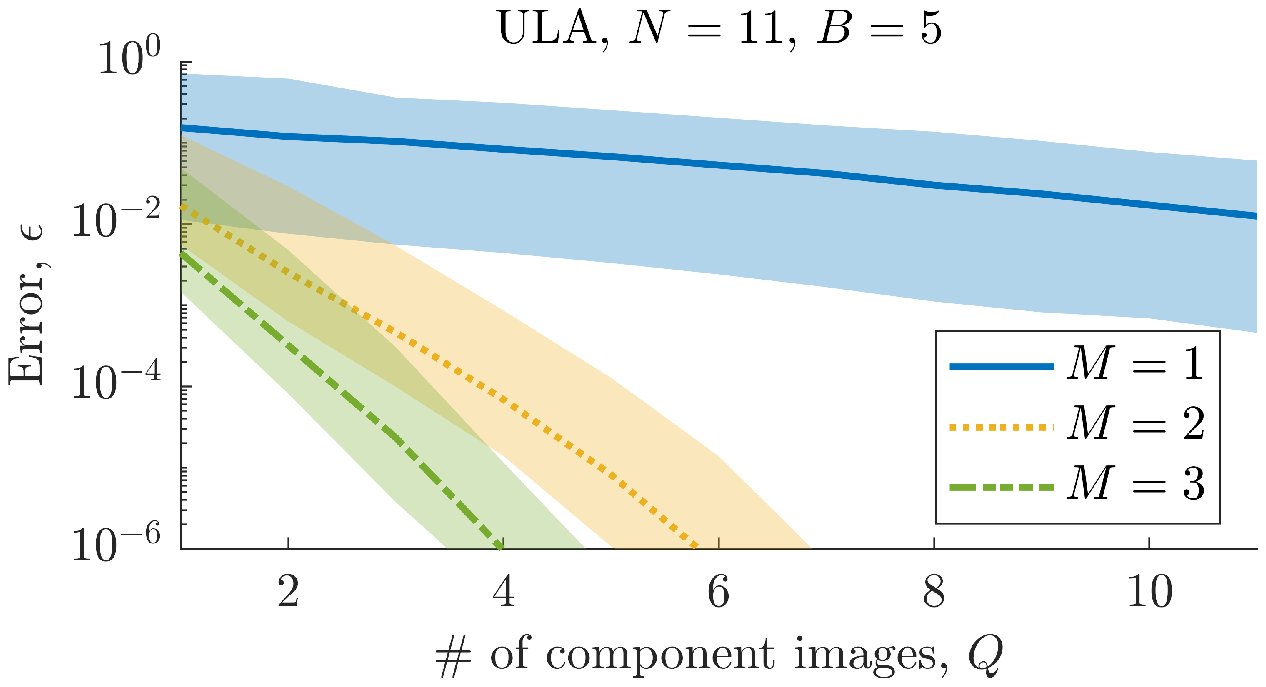}}
			{\includegraphics[width=.49\linewidth]{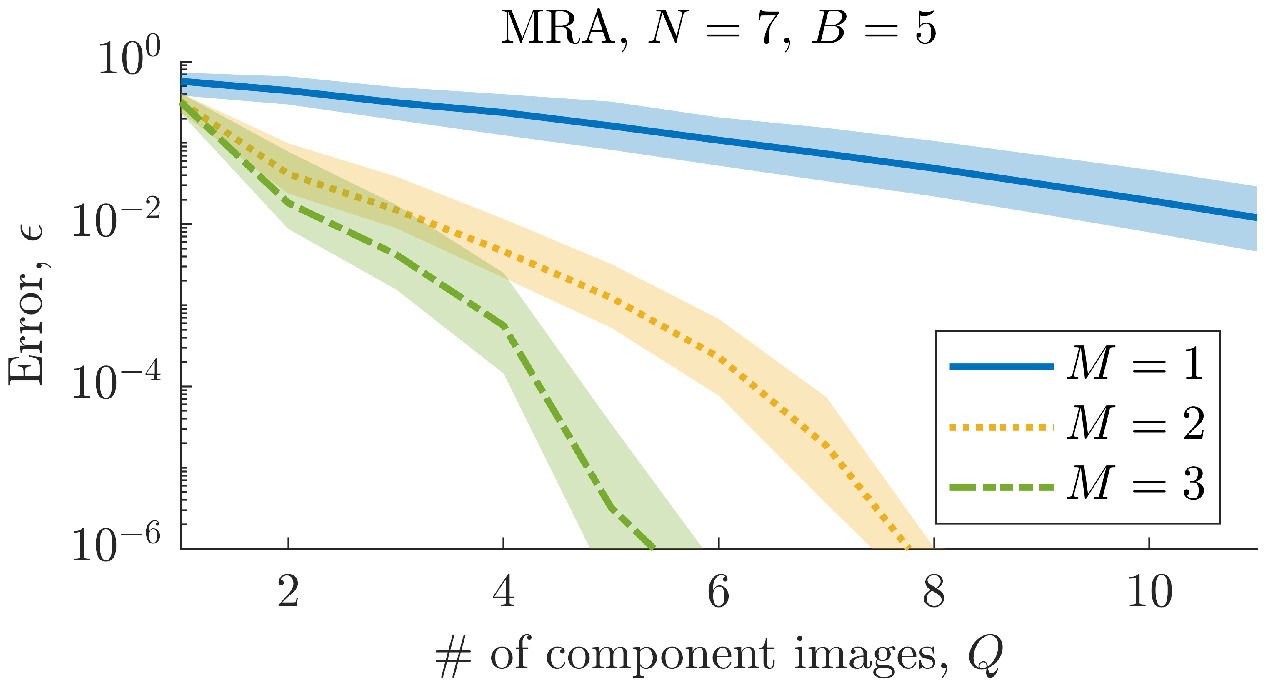}}
			{\includegraphics[width=.49\linewidth]{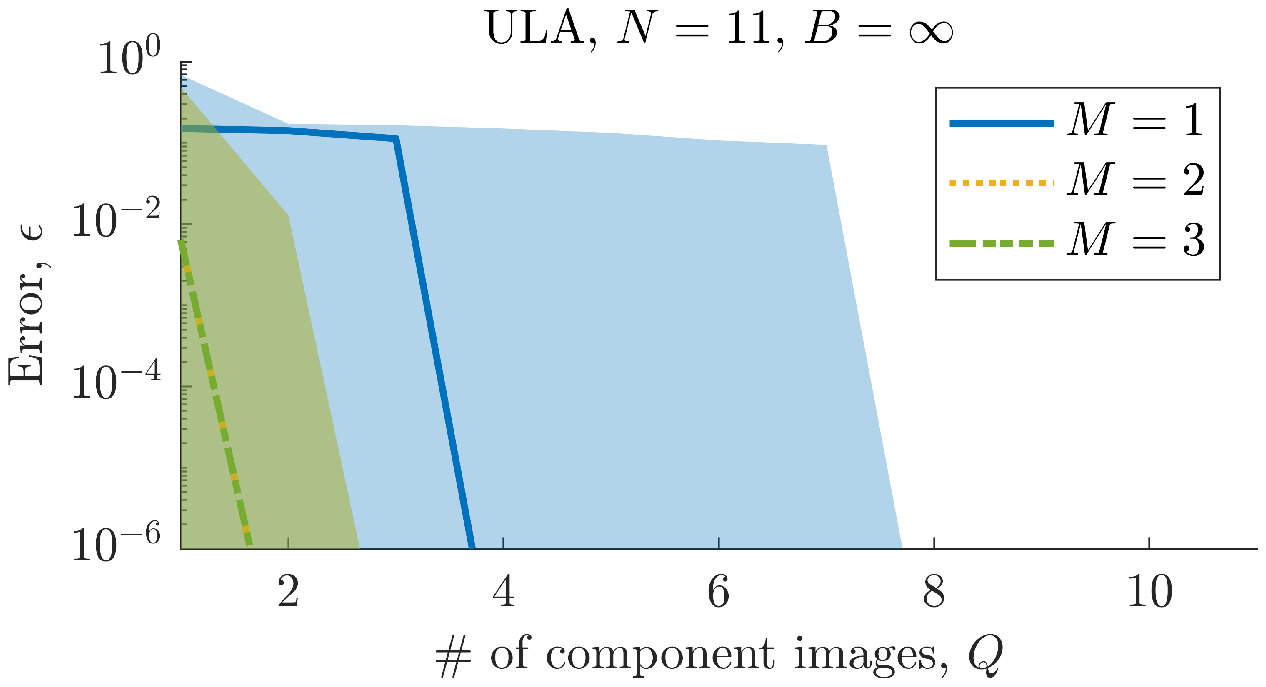}}
			{\includegraphics[width=.49\linewidth]{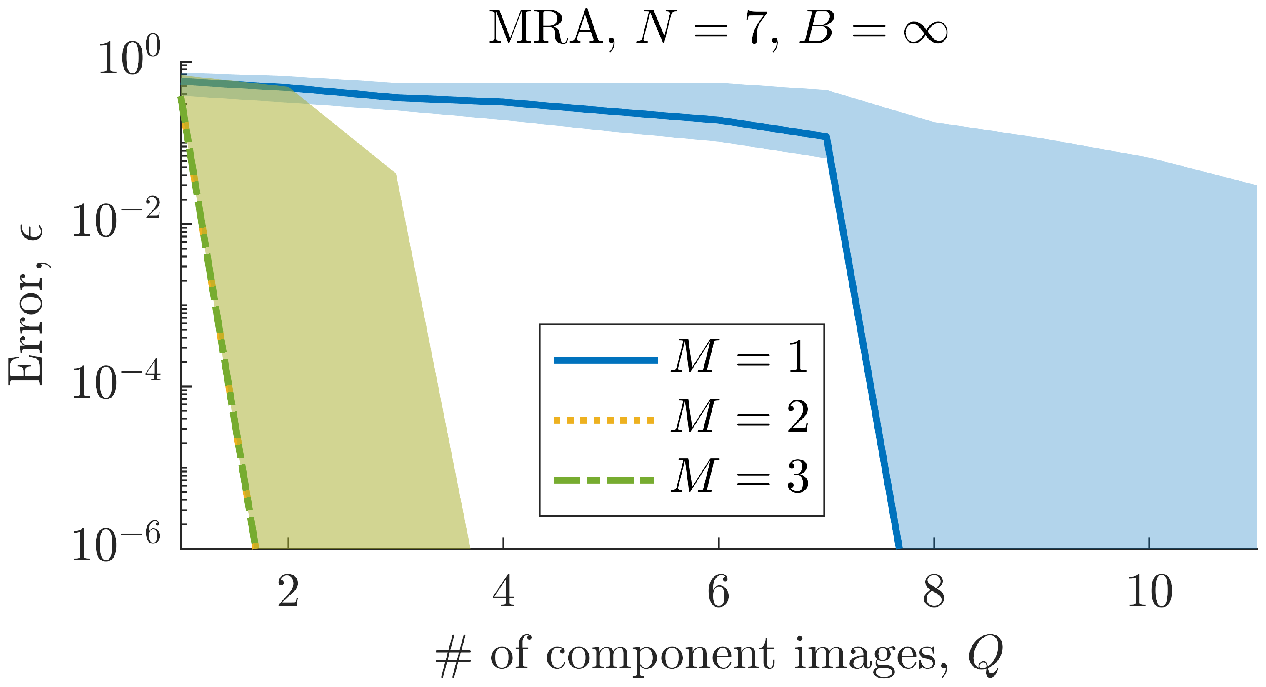}}
		\end{minipage}
		\caption{Median error and {90\% confidence intervals} of hybrid ULA (left) and MRA (right) beamformers. Generally, the solution improves more by increasing $ M$ or $Q $, rather than~$ {B} $.}\label{fig:sweep_ula}
	\end{figure}

	Fig.~\ref{fig:sweep_ula} suggests that $ M $ and $ Q $ have a larger impact on the realized PSF than $ {B} $. This is not surprising, since $ M $ and $ Q $ control the dimensions of $ \mathbf{F}_\txt{x} $ and $\mathbf{C}_\txt{x}$ in \eqref{eq:W_Bg}, whereas $ {B} $ only adjusts the quantization of $ \mathbf{F}_\txt{x} $. In other words, $ {B} $ does not affect the number of phase shifters and digital weights, unlike $ M $ and $ Q $. When $ {B} $ is finite, it is difficult to establish whether increasing $ M $ or $ Q $ will generally have a larger impact on the {error} of the realized {PSF}. However, when $ {B} $ is infinite, $ Q $ is practically the only free parameter, because $ M\!=\!2 $ suffices to achieve any fully digital factorization of $ \mathbf{W} $ by Theorem~\ref{thm:M2F2}.
	
	\subsection{Coherent imaging with planar arrays}\label{sec:examples_plan}
	{Fig.~\ref{fig:scatterers} shows the scattering scene {imaged by} the two sum co-array equivalent planar array configurations in Section~\ref{sec:num_def_plan}.} {The continuous rough surfaces of the reflecting objects are modeled by $ K = 6424$ i.i.d. points scatterers following a complex normal distribution: $ \gamma_k \sim \mathcal{CN}(\frac{1}{\sqrt{2K}},\frac{1}{2K})$. The variance of the measurement noise in \eqref{eq:gamma_hat} is $ \sigma^2=1$.} The desired (vectorized) two-dimensional co-array weighting is $ \mathbf{w}_\Sigma= \mathbf{w}_\txt{DC}\otimes \mathbf{w}_\txt{DC} $, where $ \mathbf{w}_\txt{DC}\in\mathbb{R}^{{17}} $ is a one-dimensional Dolph-Chebyshev window with $ -40 $ dB sidelobes. We evaluate the PSF and image at {$ {256\times 256} = 65536 $} pixels where the {reduced} azimuth and elevation angles {$ \sin(\varphi)$ and $ \cos(\theta)$} are sampled uniformly at {$ 256 $} points each in the interval {$ [-1,1] $.}
	\begin{figure}[t]
		\centering
		\includegraphics[width=0.5\textwidth]{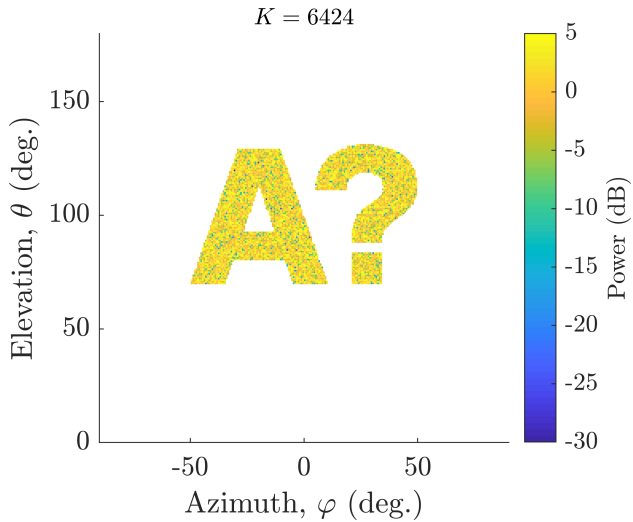}
		\caption{Point scatterer distribution and normalized reflectivity.}\label{fig:scatterers}
	\end{figure}

	Fig.~\ref{fig:planar_psf_images}~(a) shows the noiseless PSF and the noisy image of the scattering scene produced by the fully digital URA. Algorithm~\ref{alg:altmin_d} achieves the desired PSF using a single component image. In comparison, the fully digital BA {in Fig.~\ref{fig:planar_psf_images}~(b)} requires $ Q=6 $ component images to attain the same PSF. By Theorem~\ref{thm:M2F2}, the hybrid BA achieves exactly the same PSF as the fully digital BA, when $ M=2 $, $ {B}\to\infty $, and $ Q=6 $. Remark~\ref{thm:M1F2} {of Section~\ref{sec:prior_work}} also allows us to reduce the number of Tx/Rx front ends of the hybrid BA to $ M=1 $, and still achieve the PSF of the fully digital BA using $Q=6 $ component images. Alternatively, we could reduce the number of phase shifters from $ 128 $ to $ 64 $ at the expense of increasing the number of component images to $ Q=4\cdot 6 = 24$ {(cf. Theorem~\ref{thm:M1F1})}. 
	
	Fig.~\ref{fig:planar_psf_images}~{(c)} shows the PSF and image produced by the hybrid BA using Algorithm~\ref{alg:greedy} with $ {B}=5 $ bit phase {quantization}, $ M=2 $ {Tx/Rx front ends}, and {$ Q=8 $ component images.} The phase quantization slightly degrades the PSF compared to Fig.~\ref{fig:planar_psf_images}~{(b)}. However, the effect on the final image is {not drastic}. The main difference between the images {produced by the BA (both fully digital and hybrid) and the URA} is the lower noise level in the {latter}. Since the URA has more elements than the BA, it has at most {$20\log(289/64)^{\frac{3}{2}}\!\approx\!20 $} dB {\cite[Eq.~(20)]{johnson2005coherentarray}} higher array gain. {However, the difference in the final SNR may be smaller than the difference in array gain, depending on the transmit power used in each component image.}

	{Lastly, we note that the main computational effort in forming the discussed images stems from the number of times the beamforming weights need to be computed (see Section~\ref{sec:remarks_on_comp_complexity}). In the case of continuous phase shifters, solving \eqref{p:h} once suffices, since the array may be steered in an arbitrary direction by applying appropriate phase shifts to its elements. This was the case in Fig.~\ref{fig:planar_psf_images}~(a) and (b), where a MATLAB \cite{matlab2018} implementation of Algorithm~\ref{alg:altmin_d} yielded a solution in a matter of seconds on a $ 2.3 $ GHz Intel Core i5 processor. In the case of Fig.~\ref{fig:planar_psf_images}~(c), the beamforming weights were recomputed for each pixel due to the quantized phase shifters. {While a single call of Algorithm~\ref{alg:greedy} was on the order of a second, the complete image took {$32$} hours of processor time to compute (elapsed time was {$ 20$} hours).} {It is important to point out that, although it may be computationally intense, the beamforming weights can be computed offline and in parallel.}}
		\begin{figure}[h]
		\begin{minipage}[b]{1\linewidth}
			\centering
			\includegraphics[width=.39\textwidth]{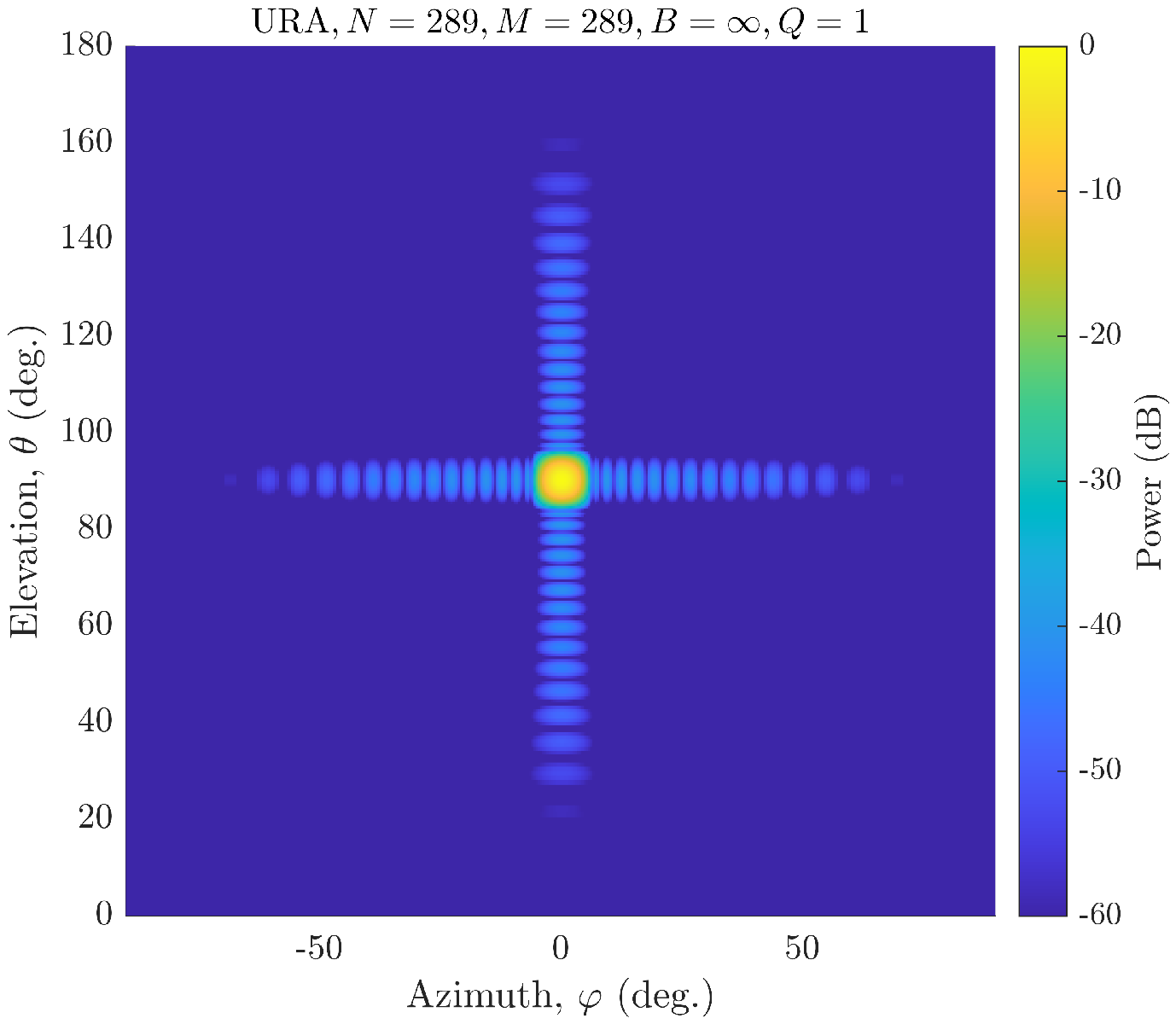}
			\includegraphics[width=.39\textwidth]{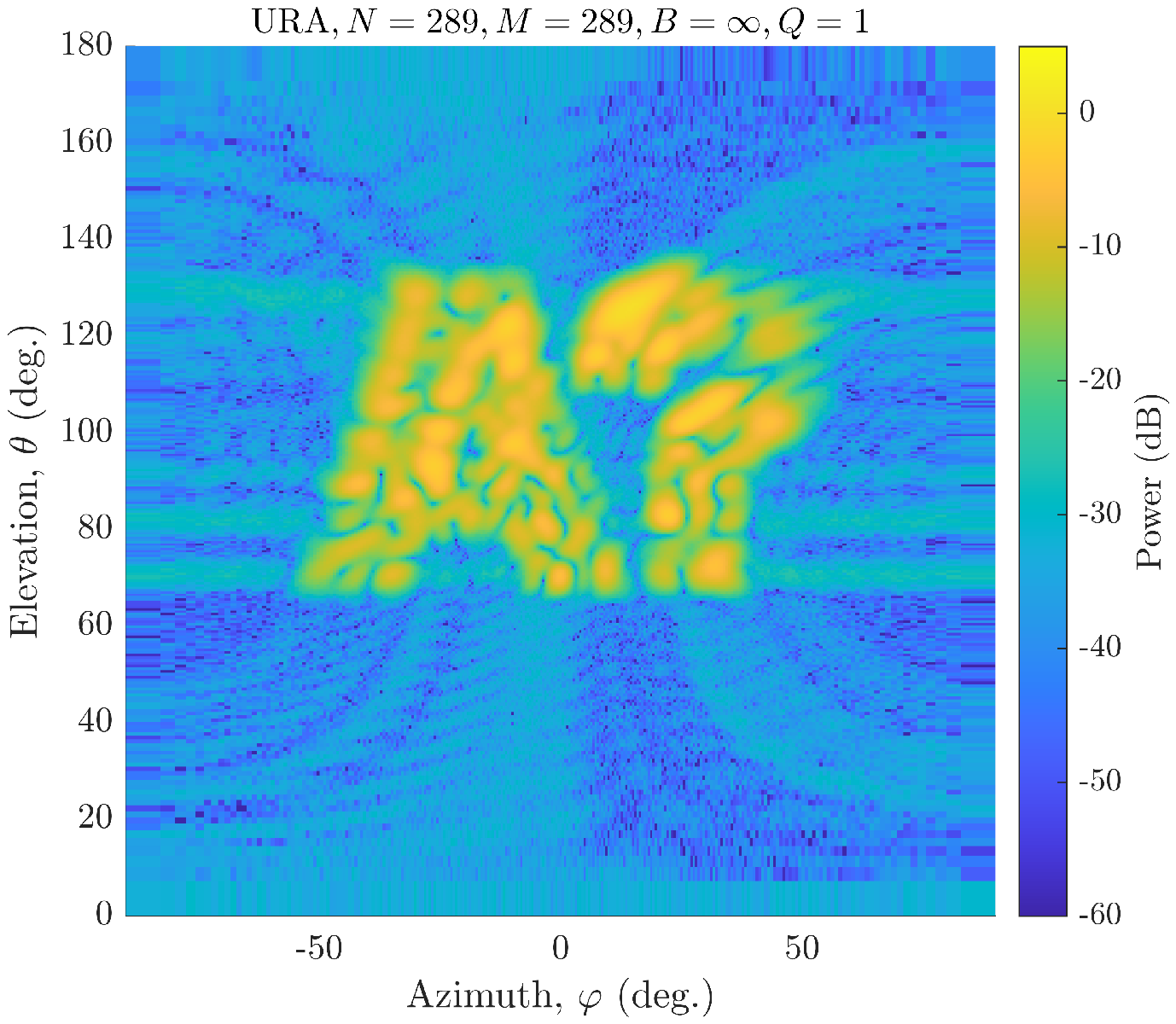}
			\centerline{(a) {Fully} digital URA ($M=289, Q=1 $)}\medskip
		\end{minipage}
		\hfill
		\begin{minipage}[b]{1\linewidth}
			\centering
			\includegraphics[width=.39\textwidth]{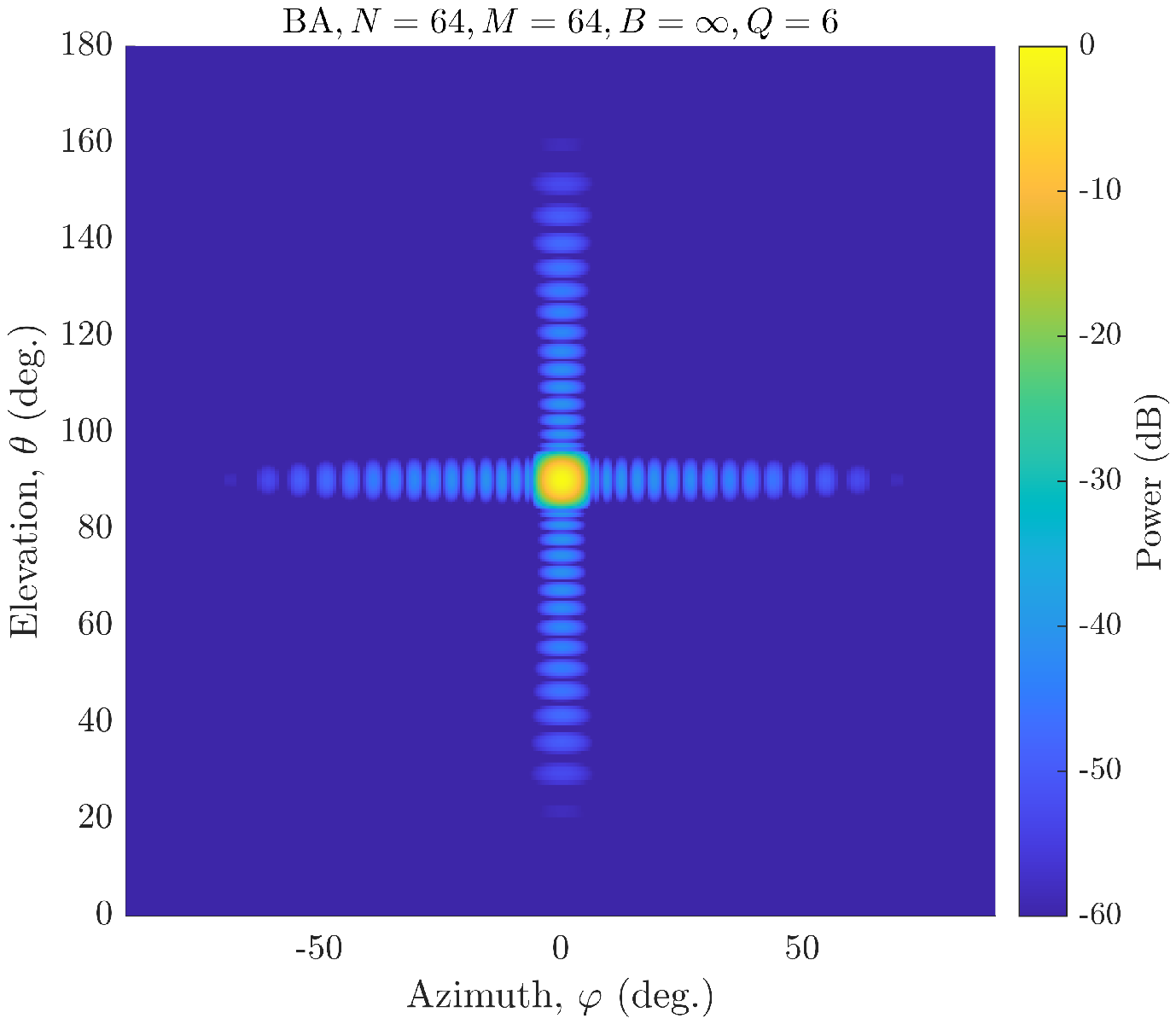}
			\includegraphics[width=.39\textwidth]{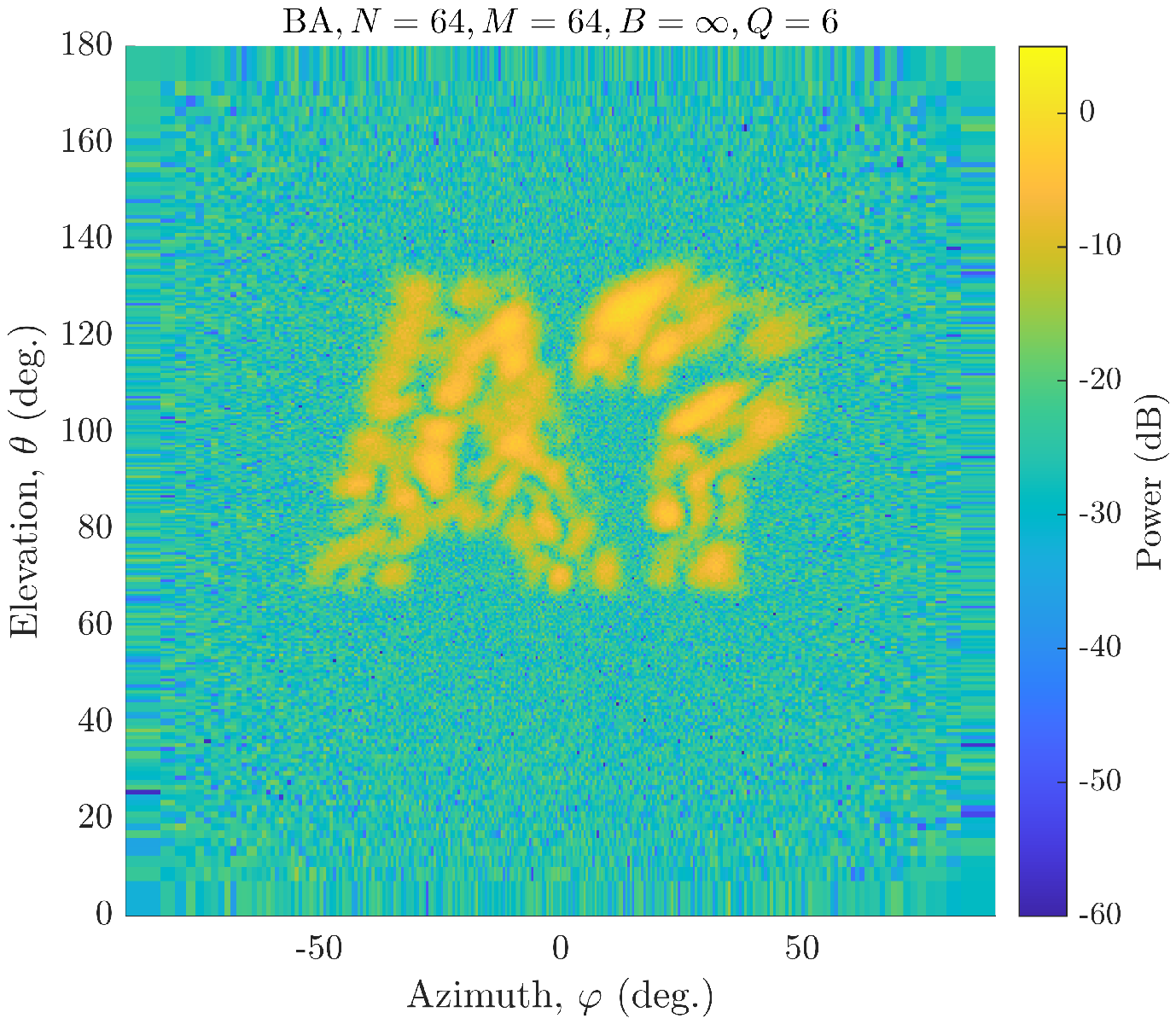}
			\centerline{(b) {Fully} digital BA ($M=64, Q=6 $)}\medskip
		\end{minipage}
		\hfill
		\begin{minipage}[b]{1\linewidth}
			\centering
			\includegraphics[width=0.39\linewidth]{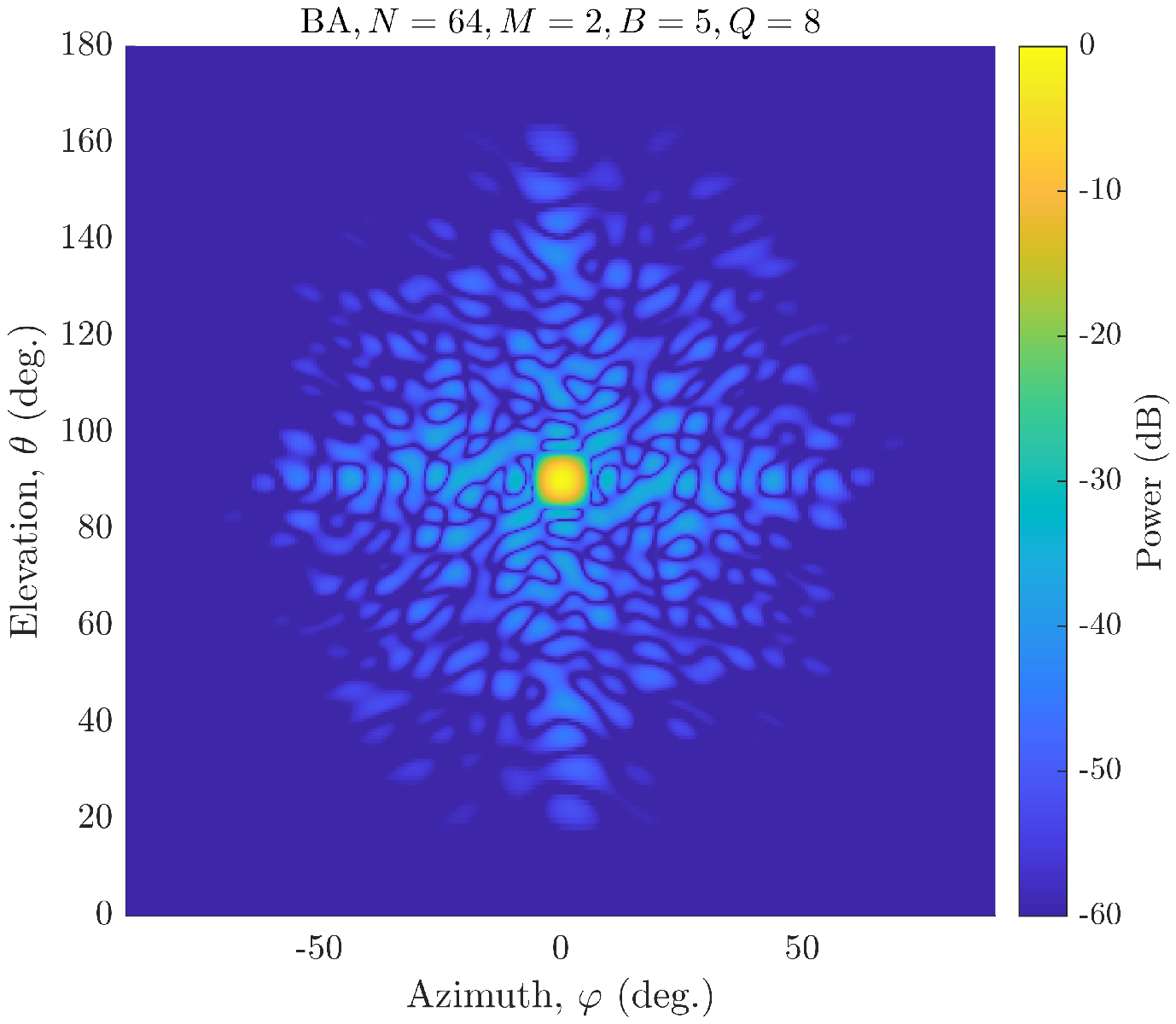}
			\includegraphics[width=0.39\linewidth]{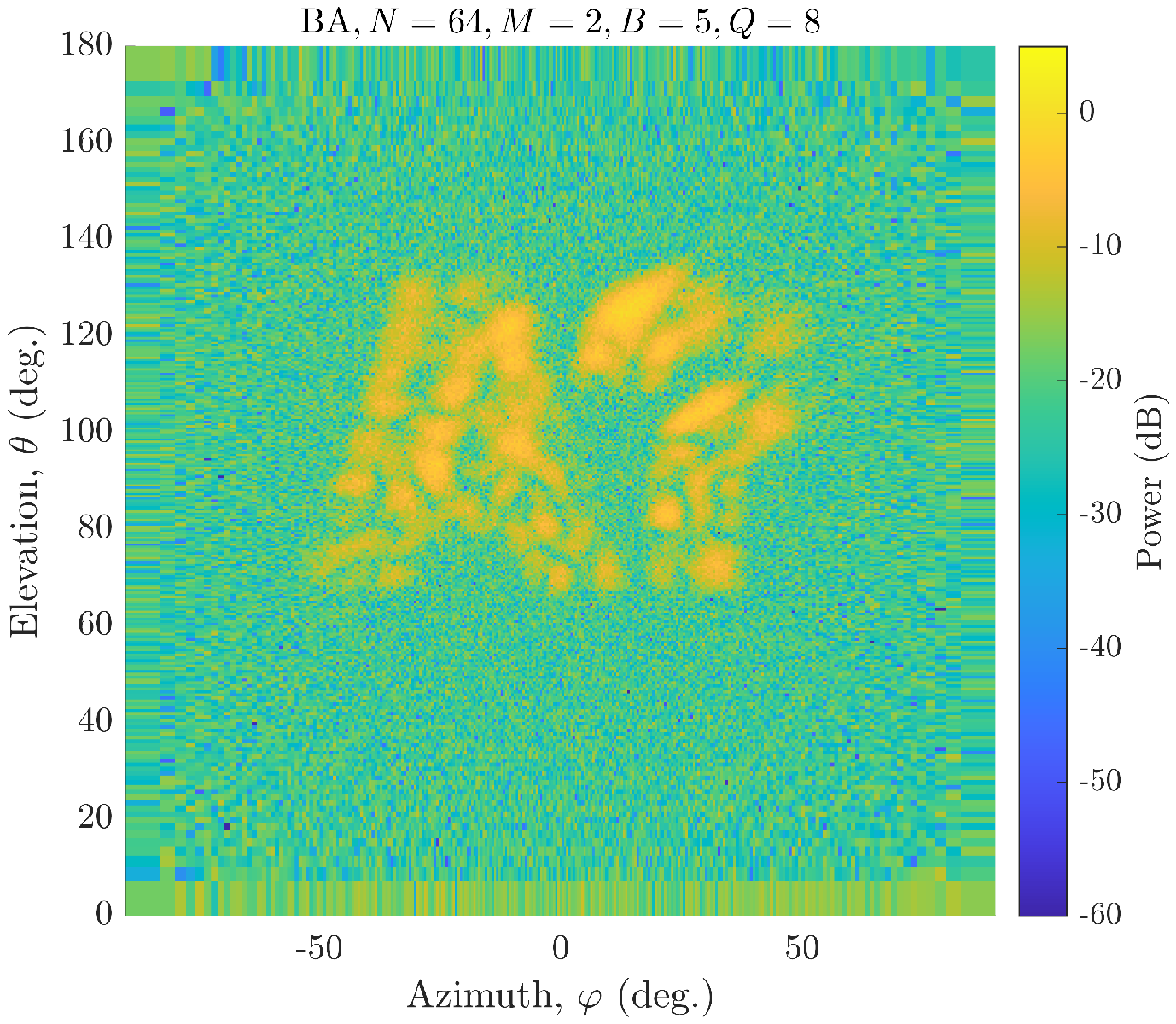}
			\centerline{(c) Hybrid BA ($M=2, Q=8,  {B}= 5 $)}\medskip
		\end{minipage}
		\caption{{PSF} (left) and image of {scatterer distribution} (right). The hybrid BA attains a comparable image to the digital URA, at the expense of more component images and lower SNR.}\label{fig:planar_psf_images}
	\end{figure}

	\section{Conclusions} \label{sec:conclusions}
	This paper considered active sensing using phased arrays with a hybrid beamforming architecture. The hybrid beamformers consist of a few Tx/Rx front ends, each connected to a network of {inexpensive} analog phase shifters with digitally controlled phase. {Such hybrid beamformers} {have} low cost and power consumption, {which} may be attractive in {applications, such as}, medical radar/ultrasound {or automotive radar}.
	
	We formulated an optimization problem, where the transmit and receive {hybrid beamforming} weights are jointly {found}, such that a desired PSF is achieved using as few component images as possible. We proposed numerical methods for finding solutions in both the fully digital, as well as the hybrid and fully analog cases. Furthermore, we derived bounds on the maximum number of component images required by some of these hybrid and analog architectures for attaining {the same} PSF {as} their fully digital counterparts. {Simulations demonstrated} that combining sparse arrays with hybrid beamforming allows for significant reductions in the number of elements and front ends. In particular, we showed {a design example, where} a hybrid sparse planar array {attained} the PSF of a $ 17\times 17 $ element fully digital uniform square array using $ 78 \%$ fewer elements and $ 99 \%$ fewer Tx/Rx front ends. These hardware savings come at the price of {an} increase in the number of component images {to $ 8 $} and {a} $ 20 $ dB reduction in array gain. {However, multiple transmissions may simultaneously increase SNR.} We observe that increasing the number of front ends or component images generally leads to {higher fidelity} PSFs than increasing the number of phase shift bits. {Generally,} only a few front ends {are} necessary for {achieving} the beamforming capabilities of a fully digital array. Indeed, two front ends are sufficient in the case of continuous phase shifters. 
	
	{In future work, the proposed hybrid beamforming framework could be extended to the more general MIMO case, with several (possibly correlated) waveforms. It would also be of practical interest to explore constraints such as quantized ADCs/DACs or unit-modulus {transmit} beamforming weights.}
		
\appendices

\section{Proof of Lemma~\ref{thm:M2F2_zhang}}\label{proof:thm:M2F2_zhang}
{Let} $ w_n\in\mathbb{C} $ be decomposed as 
\begin{align}
	w_n=|{c}_1|e^{j(\measuredangle {c}_1+\phi_n)}+|{c}_2|e^{j(\measuredangle {c}_2+\vartheta_n)},\label{eq:wn}
\end{align}
where {the} phases $ \phi_n,\vartheta_n\!\in\![0,2\pi) $ are functions of index {$ n\!=\!1,2,\ldots,N $}, but {the} complex amplitudes $ {c}_1,{c}_2\!\in\!\mathbb{C} $ are not. {By the law of cosines, angles $ \phi_n$ and $\vartheta_n $ can be written as}
\begin{align*}
	\phi_n &= \measuredangle w_n - \measuredangle {c}_1+\cos^{-1}\Big(\frac{|w_n|^2+|{c}_1|^2-|{c}_2|^2}{2|w_n||{c}_1|}\Big)\\
	\vartheta_n&= \measuredangle w_n-\measuredangle {c}_2 - \cos^{-1}\Big(\frac{|w_n|^2+|{c}_2|^2-|{c}_1|^2}{2|w_n||{c}_2|}\Big).
\end{align*}
{Assume without loss of generality that $ |{c}_1|\geq|{c}_2| $.} {It then follows from elementary trigonometry that \eqref{eq:wn} holds for all $ n $ only if} $ {c}_1,{c}_2 $ satisfy conditions $ |{c}_1|+|{c}_2|\geq \max_n(|w_n|)$, and $ |{c}_1|-|{c}_2|\leq \min_n(|w_n|) $. {For example, a particularly convenient choice is $ c_1=c_2=\max_n |w_n|/2 $, which leads to
	\begin{align*}
	\phi_n &= \measuredangle w_n +\cos^{-1}\Big(\frac{|w_n|}{\max_m |w_m|}\Big)\\
	\vartheta_n&= \measuredangle w_n - \cos^{-1}\Big(\frac{|w_n|}{\max_m |w_m|}\Big).
	\end{align*}		
	Using matrix notation then yields \eqref{eq:B_opt} and \eqref{eq:g_opt}.}\hfill\qedsymbol

\section{Proof of Lemma~\ref{thm:M1F1_approx}} \label{proof:thm:M1F1_approx}
We seek $ \argmin_{{c}\in\mathbb{C},\mathbf{f}\in\mathscr{F}(\infty)}\|\mathbf{w}-{c}\mathbf{f}\|_2^2$, or equivalently $ \argmin_{{c}\in\mathbb{C},\boldsymbol{\phi}\in\mathbb{R}^N} J(c,\boldsymbol{\phi})$, where {$ f_n\!=\!e^{j\phi_n} $ and}
\begin{align*}
	J(c,\boldsymbol{\phi})&=\sum_{n=1}^{N}\big||w_n|e^{j\measuredangle w_n}\!-\!|{c}|e^{j(\phi_n+\measuredangle{c})}\big|^2\\
	&=\sum_{n=1}^{N}|w_n|^2\!+\!|{c}|^2\!-\!2|w_n||{c}|\cos(\phi_n\!+\!\measuredangle {c}\!-\!\measuredangle w_n).
\end{align*}
The minimizer is $ \phi_n\!=\!\measuredangle w_n - \measuredangle {c} $, {which yields $ f_n\!=\!e^{j(\measuredangle w_n-\measuredangle {c})} $}. The least squares solution of $ {c} $ is then given by 
\begin{align*}
	{c}=(\mathbf{f}^\txt{H}\mathbf{f})^{-1}\mathbf{f}^\txt{H}\mathbf{w},
\end{align*}
where $ \mathbf{f}^\txt{H}\mathbf{f} = N$ and $\mathbf{f}^\txt{H}\mathbf{w} = \|\mathbf{w}\|_1 e^{j\measuredangle 
	{c}}$. Consequently, $ |{c}|\!=\!\|\mathbf{w}\|_1/N$, and 
\begin{align*}
	J\Big(\frac{\|\mathbf{w}\|_1}{N}e^{j\measuredangle 
		{c}},\measuredangle \mathbf{w} - \measuredangle {c}\Big)&=\sum_{n=1}^N |w_n|^2\!+\!\frac{\|\mathbf{w}\|_1^2}{N^2}-2\frac{|w_n|\|\mathbf{w}\|_1}{N}\\
	&=\|\mathbf{w}\|_2^2\!-\!\|\mathbf{w}\|_1^2/N.
\end{align*}
Since $ \measuredangle {c}$ is a free parameter, we may select $ \measuredangle {c}\!=\!0$ for simplicity, which yields \eqref{eq:b_opt_1} and \eqref{eq:g_opt_1}.\hfill\qedsymbol

\section{Proof of Lemma~\ref{thm:M1F1}} \label{proof:thm:M1F1}
{Note} that any $ \mathbf{w}\!=\!\mathbf{Fc}$ can be written as $ \mathbf{w}\!=\!\sum_{m=1}^M{c}_m \mathbf{F}_{:,m}$, where $ \mathbf{F}_{:,m} $ is the $ m $th column of matrix $ \mathbf{F}\!\in\!\mathbb{C}^{N\times M}$, and $ {c}_m\!\in\!\mathbb{C}$ is the $ m $th element of vector $\mathbf{c}\!\in\!\mathbb{C}^M $. It follows that 
\begin{align*}
	\mathbf{W}\!&=\!\sum_{\tilde{q}=1}^{Q}\mathbf{w}_{\txt{r},\tilde{q}}\mathbf{w}_{\txt{t},\tilde{q}}^\txt{T}\\
	&=\!\sum_{\tilde{q}=1}^{Q}\Bigg(\sum_{m_\txt{r}=1}^{M_\txt{r}}[\mathbf{c}_{\txt{r},\tilde{q}}]_{m_\txt{r}}[\mathbf{F}_{\txt{r},\tilde{q}}]_{:,m_\txt{r}}\Bigg)\Bigg(\sum_{m_\txt{t}=1}^{M_\txt{t}}[\mathbf{c}_{\txt{t},\tilde{q}}]_{m_\txt{t}}[\mathbf{F}_{\txt{t},\tilde{q}}]^\txt{T}_{:,m_\txt{t}}\Bigg)\\
	&=\!\sum_{{q}=1}^{QM_\txt{r}M_\txt{t}}{c}_{\txt{r},{q}}{c}_{\txt{t},{q}}\mathbf{f}_{\txt{r},{q}}\mathbf{f}_{\txt{t},{q}}^\txt{T}.
\end{align*}
{A feasible choice relating indices $ \tilde{q}, m_\txt{r} $ and $ m_\txt{t} $ to index $ q $ is} $ \tilde{q}\!=\!\lceil{q}/(M_\txt{r}M_\txt{t}) \rceil$; $ m_\txt{r} = \lceil (1+({q}-1) \bmod M_\txt{r}M_\txt{t})/M_\txt{t} \rceil $; and $ m_\txt{t}= 1+({q}-1) \bmod M_\txt{t}$. {This then yields \eqref{eq:bx_analog} and \eqref{eq:gx_analog}.}\hfill\qedsymbol
	
\bibliographystyle{IEEEtran}
\bibliography{IEEEabrv,references}

\end{document}